%% file: main.tex
\newif\ifabstract
\abstractfalse
\newif\iffull
\ifabstract \fullfalse \else \fulltrue \fi

\ifabstract
\documentclass[a4paper,UKenglish,cleveref, autoref, thm-restate]{lipics-v2021}
\fi

\iffull
\documentclass[11pt]{article}

\usepackage[utf8]{inputenc}
\usepackage[T1]{fontenc}
\usepackage[margin=1in]{geometry}
\usepackage{enumitem}
\usepackage{amssymb}
\usepackage{amsmath}
\usepackage{amsthm}
\usepackage{thmtools}
\usepackage{thm-restate}
\usepackage{color,soul}
\fi
\usepackage{lmodern}
\usepackage{bm}
\usepackage{tikz}
\let\realbfseries=\bfseries
\def\bfseries{\realbfseries\boldmath}

\def\defn#1{\textit{\textbf{\boldmath #1}}}

\usepackage{microtype}

{\makeatletter
 \gdef\customlabel#1#2{{%
   \protected@write\@auxout{}{\string\newlabel{#2}{{#1}{\thepage}{#1}{#2}{}} }%
  \hypertarget{#2}{}%
}}}
  
\newcommand{\eps}{{\varepsilon}}
\let\epsilon=\varepsilon

\newcommand{\changed}[1]{{\color{black}#1}}

\iffull
\usepackage{hyperref}
\newtheorem{theorem}{Theorem}
\newtheorem{claim}{Claim}
\newtheorem{definition}{Definition}
\newtheorem{lemma}{Lemma}
\newtheorem{remark}{Remark}
\newtheorem{observation}{Observation}
\newtheorem{corollary}{Corollary}
\fi

\ifabstract

\fi

\newcounter{section-preserve}
\newcounter{theorem-preserve}
\newcommand{\blank}[1]{}
\newtoks\magicAppendix
\magicAppendix={}
\newtoks\magictoks
\newif\iflater
\laterfalse
\long\def\later#1{\magictoks={#1}%
	\edef\magictodo{\noexpand\magicAppendix={\the\magicAppendix %
			\the\magictoks%
	}}
	\magictodo}
\long\def\both#1{\magictoks={#1}%
	\edef\magictodo{\noexpand\magicAppendix={\the\magicAppendix %
			\noexpand\setcounter{theorem-preserve}{\noexpand\arabic{theorem}}%
			\noexpand\setcounter{theorem}{\arabic{theorem}}%
			\noexpand\setcounter{section-preserve}{\noexpand\arabic{section}}%
			\noexpand\setcounter{section}{\arabic{section}}%
			\noexpand\let\noexpand\oldsection=\noexpand\thesection
			\noexpand\def\noexpand\thesection{\thesection}
			\noexpand\let\noexpand\oldlabel=\noexpand\label
			\noexpand\let\noexpand\label=\noexpand\blank
			\the\magictoks%
			\noexpand\setcounter{theorem}{\noexpand\arabic{theorem-preserve}}%
			\noexpand\setcounter{section}{\noexpand\arabic{section-preserve}}%
			\noexpand\let\noexpand\thesection=\noexpand\oldsection
			\noexpand\let\noexpand\label=\noexpand\oldlabel
	}}
	\magictodo
	\the\magictoks}
\def\magicappendix{\latertrue \the\magicAppendix}

\iffull
\long\def\both#1{#1}
\let\later=\both
\def\magicappendix{}
\fi

\author{Oswin Aichholzer \thanks{University of Technology Graz, Graz, Austria, \protect\url{oaich@ist.tugraz.at}}\and
  Erik D. Demaine%
    \thanks{Computer Science and Artificial Intelligence Laboratory,
      Massachusetts Institute of Technology, 32 Vassar St.,
      Cambridge, MA 02139, USA,
      \protect\url{{edemaine,virgi}@mit.edu}}
\and
  Matias Korman%
    \thanks{Siemens Electronic Design Automation, OR, USA,
      \protect\url{matias_korman@mentor.com}}
      \and
  Anna Lubiw%
     \thanks{Cheriton School of Computer Science, University of Waterloo,  Waterloo, ON, Canada, \protect\url{jaysonl@mit.edu,alubiw@uwaterloo.ca}. Supported by the Natural Sciences and Engineering Research Council of Canada (NSERC)}
\and
  Jayson Lynch\footnotemark[4]
\and
  Zuzana Mas\'arov\'a%
    \thanks{IST Austria, Am Campus 1, Klosterneuburg 3400, Austria, \protect\url{zuzana.masarova@ist.ac.at}}
    \thanks{Supported by Wittgenstein Prize, Austrian Science Fund (FWF), grant no.~Z 342-N31}
\and
  Mikhail Rudoy%
    \thanks{LeapYear Technologies, \protect\url{mrudoy@gmail.com}}
\and
  Virginia Vassilevska Williams\footnotemark[2] $^,$
    \thanks{Supported by an NSF CAREER Award, NSF Grants CCF-1528078, CCF-1514339 and CCF-1909429, a BSF Grant BSF:2012338, a Google Research Fellowship and a Sloan Research Fellowship.}
\and
  Nicole Wein\thanks{DIMACS, Rutgers University, \protect\url{nicole.wein@rutgers.edu}. Supported by a grant to DIMACS from the Simons Foundation (820931). This work was done while the author was at MIT.}}

\title{Hardness of Token Swapping on Trees}

\begin{document}

\date{}
\maketitle

\begin{abstract}
Given a graph where every vertex has exactly one labeled token, how can we most quickly execute a given permutation on the tokens? In \defn{(sequential) token swapping}, the goal is to use the shortest
possible sequence of \defn{swaps}, each of which exchanges the tokens
at the two endpoints of an edge of the graph. In \defn{parallel token swapping}, the goal is to use the fewest \defn{rounds}, each of which consists of one or more swaps 
on the edges of a matching.
We prove that both of these problems remain NP-hard when the graph is restricted to be a tree.

These token swapping problems have been studied by disparate groups of
researchers in discrete mathematics, theoretical computer science, robot
motion planning, game theory, and engineering. Previous work establishes NP-completeness on general graphs (for both problems), 
constant-factor approximation algorithms, and some poly-time exact algorithms for simple graph classes
such as cliques, stars, paths, and cycles.
%
Sequential and parallel token swapping on \emph{trees} were
first studied over thirty years ago (as ``sorting with a transposition tree'')
and over twenty-five years ago (as ``routing permutations via matchings''),
yet their complexities were previously unknown.

We also show limitations on approximation of sequential token swapping on
trees: we identify a broad class of algorithms that encompass all three known polynomial-time algorithms that achieve the best known approximation factor (which is $2$) and show that no such algorithm can achieve an approximation factor less than $2$.

\end{abstract}

\later{

}
\clearpage

\input{intro}
\input{NP-hard_weighted}
\input{NP-hard_unweighted}
\input{no_better_algs}
\later{\input{parallel.tex}}

\input{conclusion}

\section*{Acknowledgments}

This research was initiated at the 34th Bellairs Winter Workshop on
Computational Geometry, co-organized by Erik Demaine and
Godfried Toussaint, held on March 22--29, 2019 in Holetown, Barbados.
We thank the other participants of that workshop for providing
a stimulating research environment.

\bibliography{refs}
\input{appendix}

%




\end{document}

%% file: intro.tex
\section{Introduction}
Imagine $n$ distinctly labeled tokens placed without collisions
on the $n$ vertices of a graph~$G$. 
For example, these $n$ tokens might represent (densely packed) movable
\defn{agents}---robots, people, packages, shipping containers,
data packets, etc.---while the $n$ vertices represent possible agent locations.
Now suppose we want to move the tokens/agents around, for example,
to bring certain shipping containers to the loading side of a cargo ship.
In particular, we can suppose every token has a given start vertex and
destination vertex, and the goal is to move every token to its desired
destination.
Because every vertex has a token (agents are densely packed), a natural
reconfiguration operation is to \defn{swap} two adjacent tokens/agents,
that is, to exchange the tokens on the two endpoints of a given edge in~$G$.
In this paper, we study token reconfiguration by swaps from a given start
configuration to a given destination configuration with two natural objective
functions:

\begin{enumerate}
\item \textbf{(Sequential) Token Swapping}
  (a.k.a.\ ``sorting with a transposition graph''~\cite{akers1989group}):
Minimize the number of swaps, i.e., the total work required to reconfigure. 
\item \textbf{Parallel Token Swapping}
  (a.k.a.\ ``routing permutations via matchings''~\cite{alon1994routing}):
Minimize the number of rounds of simultaneous swaps (where the edges
defining the swaps form a matching, so avoid conflicting shared endpoints),
i.e., the total execution time or makespan required to reconfigure.
\end{enumerate}

These reconfiguration problems can be cast in terms of the symmetric group.
Each possible reconfiguration step---swapping along one edge in the
sequential problem, or swapping along every edge of a matching in the
parallel problem---is a particular permutation on the $n$ tokens
(an element of the symmetric group~$S_n$).
Assuming the graph is connected, these permutations generate~$S_n$,
defining a \defn{Cayley graph} $C$~\cite{cayley-graph}
where each node $\pi$ in $C$ corresponds to a permutation $\pi$
of the tokens (a collision-free placement of the tokens)
and an undirected edge connects two nodes $\pi_1,\pi_2$ in $C$
if there is a reconfiguration step (swapping an edge or matching in~$G$)
that transforms between the two corresponding permutations $\pi_1,\pi_2$.
Minimizing the number of reconfiguration steps between two configurations of
the tokens (sequential/parallel token swapping) is equivalent to finding the
shortest path in the Cayley graph between two given nodes corresponding to two
given permutations in~$S_n$.
In fact, sequential token swapping was first studied by Cayley in 1849~\cite{cayley} who (before inventing the Cayley graph) solved the problem
on a clique, i.e., without any constraint on which tokens can be swapped.

Since its introduction, token swapping has been studied by many researchers in
many disparate fields,
from discrete mathematics \cite{cayley,portier-vaughan1990star,vaughan1991bounds,vaughan1995algorithm,pak1999reduced,vaughan1999broom,kraft2015diameters}
and theoretical computer science \cite{jerrum1985complexity,knuth1998art3,ganesan2012efficient,yamanaka2015swapping,yasui2015swapping,AmirP15,miltzow2016approximation,bonnet2017complexity,yamanaka2018colored,chitturi2018sorting,kawahara2019time,chitturi2019sorting,tree-token-swapping}
to more applied fields including
network engineering as mentioned earlier \cite{akers1989group},
robot motion planning \cite{MotionPlanning_SoCG2018,surynek2019multi}, and
game theory \cite{gourves2017object}.

What is the complexity of token swapping?
In general, it is PSPACE-complete to find a shortest path between two given
nodes in a Cayley graph defined by given generators \cite{jerrum1985complexity}.
But when the generators include transpositions (single-swap permutations)
as in both sequential and parallel token swapping,
$O(n^2)$ swaps always suffice \cite{yamanaka2015swapping},
so the token-swapping problems are in NP.
Both sequential token swapping \cite{miltzow2016approximation}
and parallel token swapping \cite{banerjee2017new,kawahara2019time}
are known to be NP-complete on a general graph.

Sequential token swapping on general graphs is also known to be APX-hard \cite{miltzow2016approximation}, and even W[1]-hard with respect to the number of swaps \cite{bonnet2017complexity}. 
For the special case of graphs with constant treewidth and constant diameter, sequential token swapping is also known to be NP-hard \cite{bonnet2017complexity}. Additionally, for the special case of trees, but for the variant where the tokens have ``weights'' and ``colors'' sequential token swapping is known to be NP-hard \cite{tree-token-swapping}. From the algorithms side, there is 
a $4$-approximation
for sequential token swapping in general graphs \cite{miltzow2016approximation}.
Polynomial-time exact algorithms are known for a number of special classes of graphs including
cliques \cite{cayley}, paths \cite{knuth1998art3},
cycles \cite{jerrum1985complexity},
stars \cite{portier-vaughan1990star,pak1999reduced},
brooms \cite{vaughan1999broom,kawahara2019time,tree-token-swapping},
complete bipartite graphs \cite{yamanaka2015swapping}, and
complete split graphs \cite{yasui2015swapping}. The problem is also known to be fixed parameter tractable (where the parameter is the number of swaps) on nowhere dense graphs, which includes planar graphs and graphs
of bounded treewidth \cite{bonnet2017complexity}.
See also the surveys by Kim \cite{kim2016sorting} and
Biniaz et al.~\cite{tree-token-swapping}.

In this paper, we study the special case when the underlying graph is a tree. Sequential token swapping on a tree was first studied over thirty years ago, even before the problem was studied on general graphs.
Akers and Krishnamurthy \cite{akers1989group} studied the problem
in the context of interconnection networks.
Specifically, they proposed connecting processors together in a network defined
by a Cayley graph, in particular a Cayley graph of transpositions corresponding
to edges of a tree (what they call a \defn{transposition tree}),
so the shortest-path problem naturally arises when routing network messages.
They gave an algorithm for finding short (but not necessarily shortest) paths
in the resulting Cayley graphs, and characterized the diameter of
the Cayley graph (and thus found optimal paths \emph{in the worst case}
over possible start/destination pairs of vertices) when the tree is a star.
Follow-up work along this line attains tighter upper bounds on the diameter
of the Cayley graph in this situation when the graph is a tree
\cite{vaughan1991bounds,ganesan2012efficient,kraft2015diameters,chitturi2018sorting}
and develops exponential algorithms to compute the exact diameter of the
Cayley graph of a transposition tree \cite{chitturi2019sorting},
though the complexity of the latter problem remains open.

Sequential token swapping on a tree is the 
\changed{related}
problem
of computing the shortest-path distance
between two given nodes in the Cayley graph of a transposition tree.
For sequential token swapping on a tree, the literature exhibits a curious phenomenon whereby there are three 2-approximation algorithms that were all developed independently and all use completely different techniques. These algorithms are by Akers and Krishnamurthy \cite{akers1989group} in 1989, Vaughan and Portier \cite{vaughan1995algorithm} in 1995, and Yamanaka et al.~\cite{yamanaka2015swapping} in 2015. No better approximation factor than 2 is known. 

Parallel token swapping was also introduced in the context of network routing:
in 1994, Alon, Chung, and Graham \cite{alon1994routing} called the problem
``routing permutations via matchings''.
They focused on worst-case bounds for a given graph (the diameter of the
Cayley graph); in particular, they proved that any $n$-vertex tree
(and thus any $n$-vertex connected graph) admits a 
\changed{solution with less than $3n$ rounds,}
a bound later improved to $\frac{3}{2} n + O(\log n)$ \cite{zhang1999optimal}.
Like sequential token swapping, computing the exact diameter of the Cayley
graph of a given tree remains open. 

Parallel token swapping on a tree is the 
\changed{related} 
problem of computing the shortest-path distance between
two given nodes in such a Cayley graph.
Parallel token swapping is known
to be NP-complete in bipartite maximum-degree-$3$ graphs,
NP-complete even when restricted to just three rounds,
but polynomial-time when restricted to one or two rounds,
but NP-complete again for ``colored'' tokens restricted to two rounds
\cite{banerjee2017new,kawahara2019time}.
Two approximation results are known: an additive approximation for paths
which uses only one extra round \cite{kawahara2019time},
and a multiplicative $O(1)$-approximation for the $n \times n$ grid graph
\cite{MotionPlanning_SoCG2018}.%
\footnote{The results of \cite{MotionPlanning_SoCG2018} are phrased in terms
  of motion planning for robots, and in terms of a model where an arbitrary
  disjoint collection of cycles can rotate one step in a round.  However, the
  techniques quickly reduce to the model of swapping disjoint pairs of robots,
  so they apply to parallel motion planning as well.  They show that there is
  always a solution within a constant factor of the obvious lower bound on
  the number of rounds: the maximum distance between any token's start and
  destination.}
For other special graph classes, there are tighter worst-case bounds
on the diameter of the Cayley graph
\cite{alon1994routing,li2010routing,banerjee2017new}.
See Section \ref{sec:related} for additional related work.

\subsection{Our Results}
There have been many attempts to understand token swapping on a tree, but all have fallen short of determining its actual complexity. To summarize the previously stated results for sequential token swapping on a tree, there are three known 2-approximation algorithms, and no better approximation known. There are also exact algorithms for several special cases of trees, with the most general case being a broom (a path attached to a star). From the hardness side, attempts to prove that the problem is NP-complete have led to NP-completeness proofs for more general cases. In particular, token swapping on graphs of constant treewidth and diameter is NP-hard \cite{bonnet2017complexity}, and the ``weighted, colored'' variant of token swapping on trees is NP-hard \cite{tree-token-swapping}. This leads to our first main question: 

\begin{quote}
\emph{Question: Is sequential token swapping on a tree NP-complete?}
\end{quote}

This question has been implicit since sequential token swapping on a tree was first studied over 30 years ago, and the question has been explicitly stated
\changed{by Biniaz et al.~\cite{tree-token-swapping}
and by Bonnet et al.~\cite{bonnet2017complexity} who conjectured that the answer is yes.}

We resolve this question in the affirmative by providing a proof that sequential token swapping on a tree is NP-complete.

Next, we turn to the approximability of token swapping on a tree. The fact that there were three independently discovered 2-approximation algorithms, and nothing better is known, suggests that perhaps there is some barrier at approximation factor 2. This leads to our second main question: 

\begin{quote}
\emph{Question: Is there an inherent barrier to obtaining a $(2-\epsilon)$-approximation for sequential token swapping on trees?} 
\end{quote}

We address this question by showing that there is indeed a restriction on the \emph{types} of algorithms that can achieve approximation factor better than 2. To motivate the class of algorithms we rule out, it helps to examine known algorithms. Specifically, it was previously known that neither Akers and Krishnamurthy's ``happy swap'' algorithm \cite{akers1989group} nor Yamanaka et al.'s cycle algorithm \cite{yamanaka2015swapping} can possibly achieve an approximation ratio better than 2 \cite{tree-token-swapping}. 
These two algorithms share a natural property: every token $t$ always
remains within distance $1$ of the shortest path from $t$'s start vertex to
$t$'s destination vertex. A natural question is, can a better-than-$2$ approximation be achieved if one allows tokens to deviate from their shortest paths more, say to distance $10$ or $100$?

Motivated by this question, we define an \defn{$\ell$-straying} algorithm as an algorithm that never moves a token a distance more than $\ell$ from its shortest path. We prove a surprisingly strong limitation on $\ell$-straying algorithms: any less-than-$2$-approximation algorithm for sequential token swapping on trees must in general bring a token
\emph{arbitrarily far}---an $\Omega(n^{1-\epsilon})$ distance away---from its shortest path. That is, no $\ell$-straying algorithm for $\ell=o(n^{1-\epsilon})$ can achieve better than a 2-approximation.




The other known 2-approximation algorithm (besides \cite{akers1989group} and \cite{yamanaka2015swapping}), is the Vaughan-Portier algorithm
\cite{vaughan1995algorithm}, which in fact \emph{does} move tokens arbitrarily far from their shortest paths. That is, our result on $\ell$-straying algorithms does not imply a limitation on the Vaughan-Portier algorithm. To address this,
we also obtain the first proof that the Vaughan-Portier algorithm
\cite{vaughan1995algorithm} is no better than a $2$-approximation;
the best previous lower bound for its approximation factor was $\frac{4}{3}$
\cite{tree-token-swapping}. Thus, none of the known algorithms or even their generalizations can improve upon the approximation factor of $2$. 

For parallel token swapping on a tree, less is known than for the sequential version. In particular, there is no known approximation algorithm nor is there any known hardness for tree-like graphs. Thus, the complexity of this problem is completely unclear. This leads to our third main question: 

\begin{quote}
\emph{Question: What is the complexity of parallel token swapping on a tree?} 
\end{quote}

We address this question by showing that parallel token swapping on a tree is NP-hard.

In summary, our results are as follows:

\begin{enumerate}
    \item Sequential token swapping is NP-complete on trees.
    \item Parallel token swapping is NP-complete on trees, even on subdivided stars.
    \item Limitations on known techniques for approximating sequential token swapping on trees:
    \begin{enumerate}[label=\alph*)]
        \item No $\ell$-straying algorithm for any $\ell=O(n^{1-\eps})$ can achieve better than a 2-approximation.
        \item The Vaughan-Portier algorithm does not achieve better than a 2-approximation.
    \end{enumerate}
\end{enumerate}

\subsection{Our Techniques}\label{sec:techniques}


\paragraph*{NP-hardness of sequential token swapping on trees}

Our NP-hardness proof for sequential token swapping on trees is our
most technical and conceptually difficult result. Prior work has built towards this result by providing NP-hardness for generalizations of the problem, but there appear to be barriers against extending these techniques. In the following, we briefly review this prior work and compare it to our own.

Token swapping on trees is known to be NP-hard for the variant where tokens have weights as well as ``colors'' \cite{tree-token-swapping}. 
However, the use of weights and colors appears to be crucial to the reduction. Token swapping is also known to be NP-hard on graphs with treewidth 2 and diameter 6 \cite{bonnet2017complexity}. In particular, the graph in this construction is almost a tree in the sense that if you remove a single vertex the remaining graph is a forest. However, this single vertex has very high degree and is crucial to the construction. Given the apparent barriers against extending these known approaches to token swapping on trees, we take a completely different approach. 

We reduce from the
\defn{permutation generation} problem in Garey and Johnson~\cite[MS6]{Garey-Johnson} (also called the ``word problem for products of symmetric groups'' (WPPSG) in~\cite{garey1980complexity}). In comparison, the above prior work \cite{tree-token-swapping,bonnet2017complexity} reduces from the vertex cover problem, and the 3-dimensional matching problem, respectively. We observe that the permutation generation problem has a similar feel to token swapping, as it can be recast in terms of a token-swapping reachability
problem 
as follows:
%
\begin{quote}
\textbf{Star Subsequence Token-Swapping Reachability (Star STS):}
Given a star graph with center vertex $0$ and leaves $1, 2, \ldots, m$,
where vertex $i$ initially has a token $i$;
given a target permutation $\pi$ of the tokens;
and given a sequence of swaps $s_1, s_2, \ldots, s_n$,
where $s_j \in \{1, \ldots, m\}$ indicates a swap on edge $(0, s_j)$,
is there a subsequence of the given swaps that realizes~$\pi$?
\end{quote}
This recasting of the problem is a technicality that is not conceptually important. Appendix~\ref{appendixA} gives the straightforward reduction from
permutation generation to Star STS. 

As a first step towards reducing from Star STS to token swapping on trees, we reduce to 
\emph{weighted} token swapping on trees,
where each token has a non-negative integer \emph{weight},
and the cost of a swap is the sum of the weights of the two tokens being
swapped. Our reduction contains only tokens of weight 0 or 1. That is, the tokens of weight 0 are free to move, while the tokens of weight 1 cost to move. Our reduction from Star STS to 0/1-weighted token swapping on trees is quite simple. 
It is presented in Section \ref{sec:NP-hard-weighted}.

The situation becomes much more complicated when we extend this result from the 0/1-weighted setting to the unweighted setting. Now, we need to \emph{simulate} the weight-0 tokens using unweighted tokens. This introduces several complications. 

First, we will describe why weight-0 tokens are integral to our reduction from Star STS to 0/1-weighted token swapping on trees. The Star STS problem asks whether there \emph{exists} a subsequence of swaps that realizes the target permutation $\pi$. This subsequence could contain any number of swaps. In our reduction to 0/1-weighted token swapping, the swaps from this subsequence are represented using tokens of weight 0. This way, if there is a solution to the Star STS instance, then the cost of the 0/1-weighted token swapping is the same \emph{regardless} of how many swaps occurred in the solution to the Star STS instance. 
This introduces a challenge for unweighted token swapping for the following reason.
For 0/1 weighted token swapping, we prove a statement of the form ``if the token swapping cost is exactly $K$ then there is a solution to the Star STS instance'', while for unweighted token swapping, we prove a statement of the form ``if the token swapping cost is within a particular \emph{range} then there is a solution to the Star STS instance''. The second statement is much more difficult to prove because we need to argue that the additional swaps in this range do not allow the tokens to move around in a clever way to admit a solution even when there is no Star STS solution. In fact, as we discuss next, natural modifications of the weighted construction \emph{do} admit such clever ways to create counterexamples. 

The most basic first attempt to remove the weights from the weighted construction is simply to replace all weight-0 tokens with unweighted tokens. This construction admits a straightforward counterexample due to the increased cost of swapping these formerly weight-0 tokens. Thus, we would like to make the contribution of the formerly weight-0 tokens negligible in comparison to the weight-1 tokens. A natural attempt is to replace each weight-1 token with a \emph{long path} of tokens. However, as it turns out, there is a surprising and subtle counterexample to this strategy. To overcome this counterexample, we introduce a set of ``padding tokens'' throughout the graph whose role is to block any deviant movement of the original tokens.

The resulting proof is very involved. To give a sense of the complexity, our 0/1-weighted hardness proof fits in just a
couple of pages, while our unweighted proof spans around thirty pages. Our unweighted proof is presented in Section \ref{sec:NP-hard-unweighted}.

%
%
%

\paragraph*{NP-hardness of parallel token swapping on trees }

We prove that parallel token swapping on trees is NP-hard, even when restricted to \defn{subdivided stars}. This result is presented in
Section \ref{sec:parallel}. 
As for sequential token swapping, we reduce from the Star STS problem. 

Our construction is reminiscent of our construction for sequential token swapping,
\changed{although}
the details differ significantly. In particular, we use the single high-degree vertex in the subdivided star as a bottleneck to limit the available parallelism. We
develop ``enforcement'' tokens that need to swap through the high-degree
vertex to force congestion at specific times.
This proof's complexity is between the weighted and unweighted sequential
hardness proofs.

\paragraph*{Limitations on known techniques for approximation algorithms}

%
%
To prove that neither an $\ell$-straying algorithm nor the Vaughan-Portier algorithm can achieve better than a 2-approximation,
we use a problem instance that has been previously used to prove that Akers and
Krishnamurthy's and Yamanaka et al.'s algorithms cannot achieve an approximation ratio better than 2 \cite{tree-token-swapping}. To prove our results, we show that while there exists a solution to the instance with $K$ swaps (for some $K$), (1) every $\ell$-straying algorithm performs $2K$ swaps, and (2) the Vaughan-Portier algorithm performs $2K$ swaps. The existence of a solution with $K$ swaps was already shown by \cite{tree-token-swapping}, so it remains to show that the above 
\changed{algorithms}
require $2K$ swaps. 

We emphasize that the proofs in \cite{tree-token-swapping} for Akers and
Krishnamurthy's and Yamanaka et al.'s algorithms, as well as our proof for the Vaughan-Portier algorithm, are for \emph{specific} algorithms, while our proof for $\ell$-straying algorithms shows limitations against a very wide class of possible algorithms. Thus, our proof for $\ell$-straying algorithms requires much more general reasoning about how tokens can possibly move around the graph.




\subsection{Additional Related Work}\label{sec:related}

In the ``colored'' generalization of token swapping
\cite{yamanaka2018colored}, some subsets of tokens
have identical labels (colors), so each of these subsets (color classes)
can be permuted arbitrarily.  This setting arises naturally in applications
where agents (tokens) do not densely fill locations (vertices):
one color class can represent ``empty'' vertices.
Colored sequential token swapping is NP-hard even with just three color
classes, but polynomial-time with two color classes \cite{yamanaka2018colored}.
Miltzow et al.'s $4$-approximation algorithm for sequential token swapping
\cite{miltzow2016approximation} generalizes to this setting.
Exact solutions for paths and stars generalize to colored sequential token
swapping, even with weights on tokens \cite{tree-token-swapping}.
NP-hardness for weighted colored sequential token swapping on trees was proved
recently \cite{tree-token-swapping}.


Other models for moving tokens in graphs have been introduced.
In ``sliding tokens'' \cite{NCL_TCS}, some vertices do not have tokens,
vertices with tokens must always form an independent set,
tokens are all identical (one color class), and
the goal is to decide whether reconfiguration is even possible
(a variation of independent set reconfiguration
\cite{NPReconfiguration_TCS}).
This problem is PSPACE-complete even for planar graphs \cite{NCL_TCS},
but polynomial on trees \cite{TokenReconfigurationTrees_TCS2015}.
The associated Cayley graphs are called ``token graphs'', 
and their properties,
such as $k$-connectivity, are well-studied; see \cite{fabila2020token}.
In ``movement problems'' \cite{Movement_TAlg,Movement_APPROX2011,MovementFPT_TALG},
the problem is further relaxed to allow multiple tokens on the same vertex,
and the goal is for the vertices with tokens to induce a subgraph with a
particular property such as connectivity.
In ``sequentially swapping'' \cite{SequentialTokenSwapping_JGAA2019},
the entire reconfiguration sequence consists of moving one token
(conceptually, a blank space) along a nonsimple path in the graph,
swapping each token encountered with the previous position in the path,
and the goal is to minimize the length of the path.
This model is a natural graph generalization of the Fifteen Puzzle
(which is itself NP-hard on an $n \times n$ grid
\cite{ratner1990n2,FifteenPuzzle_TCS}).
This problem can be solved exactly on trees, complete graphs, and cycles;
and with two color classes, it is APX-hard
\cite{SequentialTokenSwapping_JGAA2019}.
There is recent work exploring a version of token swapping where some specified pairs of tokens may not swap~\cite{defant2020friends}.

Many of these problems have applications in AI, robotics, motion planning,
and quantum circuit compilation.
See \cite{surynek2019multi} for a recent survey on the practical side.

There are also game-theoretic models for token swapping, where vertices are
selfish players each of which has a strict total order on their preference
among the tokens (called objects).  A swap between adjacent vertices is
possible only when it benefits both vertices (each player obtaining a higher
object in their total order).  In this setting, deciding whether a given
configuration can be reached is NP-complete but polynomial time on trees;
deciding whether a given token can reach a specified vertex is NP-complete
even for a tree but polynomial on paths and stars; and
finding an equilibrium (Pareto-efficient) configuration is NP-complete
but polynomial on paths and stars and open on trees
\cite{gourves2017object}.

%% file: NP-hard_weighted.tex
\section{Weighted sequential token swapping on trees is NP-hard}
\label{sec:NP-hard-weighted}

In this section, we prove that weighted token swapping on a tree is NP-hard,
even when the token weights are in $\{0,1\}$.
Our purpose is to introduce the general idea that is used in our main
NP-completeness proof for the unweighted case
given in Section~\ref{sec:NP-hard-unweighted}.

We first precisely define 
the decision problem
\defn{weighted sequential token swapping on trees}
\cite{tree-token-swapping},
abbreviated in this section to \defn{weighted token swapping}.
The input consists of: a tree on $n$ vertices
with distinct initial positions (vertices) and distinct target positions for the $n$ tokens; non-negative integer weights on the  tokens;  and a maximum cost~$K$.
The cost of a swap is the sum of the weights of the two tokens involved.
The decision problem asks whether there is a sequence of swaps that moves all
the tokens to their target positions and such that the sum of the costs
of the swaps 
is at most~$K$.
 
 It is not clear whether the weighted token swapping problem lies in NP; however, it is in NP if $K$ is given in unary.\footnote{Consider 
 a minimum length swap sequence of weight at most $K$.
 The number of swaps involving a nonzero weight token is at most %
 $K$.
Because the sequence has minimum length, 
it can be shown that no two zero-weight tokens swap more than once.
Thus the sequence has length at most $K + n^2$  and provides a polynomial-size certificate, showing that the problem is in NP.
}

We prove that weighted token swapping is NP-hard when $K$ is given in unary.
The reduction is from Star STS (defined in Section \ref{sec:techniques}). 
In order to distinguish tokens and vertices in the original star from those in the tree that we construct, we  will call  the tokens of the Star STS instance \defn{items} and we will call the leaves \defn{slots}. 
Then the instance of Star STS consists of: a star with  center 0 and slots $1, \ldots, m$, each  of which initially has an item of its same label;  a permutation $\pi$ of the items; and a sequence $s_1, \ldots, s_{n}$ of slots  that specify  the allowed swaps.
Figure~\ref{fig:plan}(a) shows an  example input for $m=4$ slots, $5$ items and a sequence of length $n=7$.

\begin{theorem}
\label{thm:weighted-NP-hard}
Weighted token swapping on trees is NP-hard.
\end{theorem}
\begin{proof}
Suppose we are given an instance of Star STS as described above.
We may  assume without loss of generality that every slot appears in the sequence (otherwise remove that slot from  the problem) and that no slot appears 
twice in a row in the sequence.  

\begin{figure}[bt]
\centering
\includegraphics[width=\textwidth]{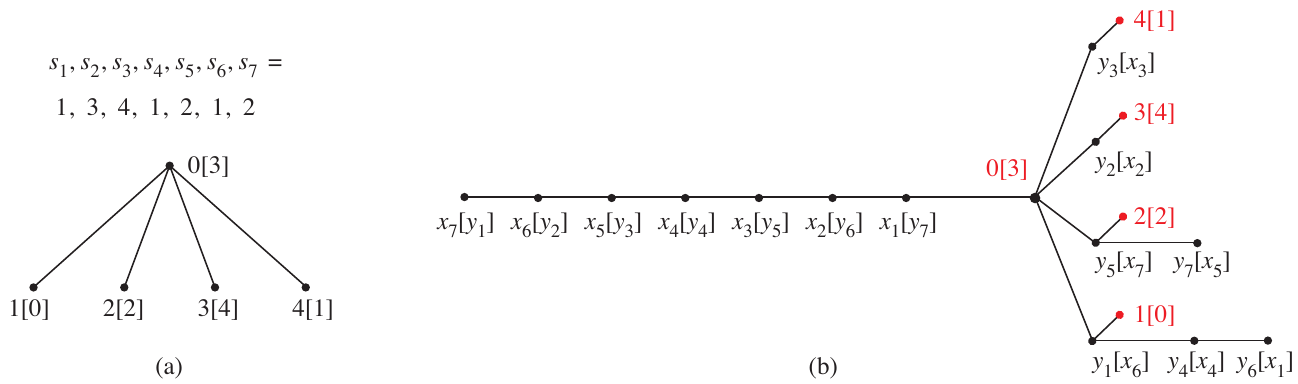}
\caption{
(a) An instance of Star STS with $m=4$ slots and a sequence of length $n=7$; notation $a[b]$ indicates that token $a$ is initially at this vertex and token $b$ should move to this vertex.  This instance  has no solution because item 4 should move to slot 3, which is possible only if slot 3 appears after slot 4 in the sequence.
(b) The corresponding  instance of  weighted token swapping with the ordering gadget on the left and $m=4$ slot gadgets attached to the root, each with a nook vertex shown in red.
After  swapping  item token $0$ at the root with  token $y_1$ in the first (bottom) slot gadget  there is an opportunity to swap tokens $0$ and $1$ for free along the nook edge of the first slot gadget, before moving $y_1$ to its target position at the end of  the ordering gadget and moving $x_1$ to its target position  at the end of the first slot gadget.
}
\label{fig:plan}
\end{figure}

Construct a tree with a root, an  \defn{ordering gadget} which is a path  of length $n$  attached to  the  root, and $m$ \defn{slot gadgets} attached to  the root. 
Slot gadgets are defined below.
See Figure~\ref{fig:plan}(b) where the ordering gadget of length  $n=7$ appears on the left and there are $m=4$ slot gadgets attached to the root.
We picture the tree with the root in the middle, and use directions left/right as in  the figure.

Let $n_i$ be the number of  occurrences  of slot $i$ in the input sequence.
Observe that $\sum_{i=1}^m n_i = n$.
\defn{Slot gadget} $i$ consists of a path of $n_i$ vertices, plus an extra leaf attached to the 
leftmost
vertex of the path. This extra leaf is called the \defn{nook} of the slot gadget. 
The  $m$ nooks and the  root 
are  in one-to-one correspondence with the slots and the center of the original star (respectively), and we place 
\defn{item tokens} at these vertices whose names,  initial positions, and  final positions correspond exactly to those of the items of the input star.
These item tokens are given a weight of 0.

There are $2n$ additional \defn{non-item tokens} $x_1, \ldots,  x_{n}$ and $y_1, \ldots, y_{n}$.  
These all have weight 1.  The $x_j$'s are initially placed along the ordering gadget path, in order, with $x_1$ at the right and $x_{n}$ at the left.
The ordering path is also the target position of the $y_j$'s in reverse order with 
$y_1$ at  the left and $y_{n}$ at the right.

Suppose slot $i$ appears in the sequence as $s_{j_1}, s_{j_2}, \ldots, s_{j_{n_i}}$ with indices in  order $j_1 < j_2 < \cdots < j_{n_i}$.  Then tokens $y_{j_1}, y_{j_2}, \ldots, y_{j_{n_i}}$ are initially placed along the path of slot gadget $i$, in order with smallest index at the left.
The path of slot gadget $i$ is also the target position of the tokens
$x_{j_1}, x_{j_2}, \ldots, x_{j_{n_i}}$ in reverse order with smallest index at the right.

Consider, for each non-item token  $x_j$ or $y_j$, the distance from its initial location to its target location. 
This is a lower bound on the cost of moving that token. We set the max cost  $K$ to the sum of these lower bounds. This guarantees that every $x_j$ or $y_j$ only travels along its shortest path.
Observe that this reduction takes polynomial time.
We now prove that a YES instance of  Star STS maps to a YES instance of  token swapping and vice versa. 

\medskip\noindent
{\bf YES instance of Star STS.}
Suppose the Star STS instance has a solution.
The  ``intended''  solution to the token swapping instance  implements 
each $s_j$ for $j=1, \ldots, n$ as follows.  Suppose $s_j = i$. 
By induction on $j$, we claim that token $y_j$ will be 
in the leftmost vertex of slot gadget $i$ when it is time to implement $s_j$.   
Swap token $y_j$ with the item token $t$ currently at the root.  Then item token $t$ has  the opportunity to swap for free with the item token in nook $i$. We perform this free swap if and only if swap $s_j$ was performed in the solution to Star STS.
Next, swap tokens $y_j$ and $x_j$---we claim by  induction that $x_j$ will be in the rightmost vertex of the ordering gadget.  
Finally, move token $y_j$ to its target position in the ordering gadget, and move $x_j$ to its target  position in slot gadget $i$.
It  is straightforward to  verify the induction assumptions.
Every $x_j$ and $y_j$ moves along its shortest path so the cost of the solution is equal to the specified bound $K$.

\smallskip\noindent
{\bf YES instance of weighted token swapping.}
Suppose the weighted token swapping instance has a solution with at most $K$ swaps.
Because $K$ is the sum of the distances of the non-item tokens from their target positions,  
each $x_j$ and $y_j$ can only move along the shortest path to  its target position.
Token $x_j$ must move into the root before $x_{j+1}$, otherwise they would need to swap before that, which means moving $x_j$ the wrong way along the ordering gadget path. Similarly, $y_j$ must move into the root before $y_{j+1}$, otherwise they need to swap after that, which means moving $y_{j+1}$ the wrong way.

Furthermore, $x_j$ must move into the root before $y_{j+1}$ otherwise the ordering gadget would contain $x_j,  \ldots, x_{n}$, and $y_1, \ldots, y_j$, a total of $n+1$ tokens, which is more tokens than there are vertices in the ordering gadget. 

We also claim $y_j$ must move into the root before $x_{j+1}$.
The nooks can only contain item tokens, since no $x_j$ or $y_j$ can move into a nook.  This accounts for every item token except for  one ``free''  item token.
Now suppose $x_{j+1}$ moves into  the root before $y_j$.   Then  the ordering gadget contains $x_{j+2},  \ldots, x_{n}$, and $y_1, \ldots, y_{j-1}$, a total of $n-2$ tokens.   Even if  the free item token is in the ordering gadget, there are not enough tokens to fill the ordering gadget.    

Thus the $x_j$'s and $y_j$'s must use the root in order, first $x_1$ and $y_1$ in some order, then $x_2$ and $y_2$ in some order, etc.
Finally, we examine  the swaps of item tokens.
A swap between two item tokens can only  occur when the free item token is at the parent of a nook. Suppose 
this happens in slot gadget~$i$.
Then some token $y_j$ must have left the slot gadget,  and the corresponding token  $x_j$ has not yet entered the slot gadget, which means that neither of the tokens $y_{j+1}$  or $x_{j+1}$  has moved into the root.  This implies that any swap  of item tokens is associated with a unique $s_j = i$, and such swaps must occur in  order of $j$, $1 \le j \le n$.   Thus the swaps of item tokens can be  mimicked by swaps in  the original Star STS sequence, and the Star STS instance has a solution.
\end{proof}

%% file: NP-hard_unweighted.tex
\section{Sequential token swapping on trees is NP-complete}
\label{sec:NP-hard-unweighted}

Recall that  the  token swapping problem is a decision  problem: given a tree with initial and target positions of the tokens, and given a non-negative integer $K$, can the tokens be moved from their initial to  their target positions with at most $K$ swaps.
Membership in NP is easy to show: any problem instance can always be solved with a quadratic number of swaps by repeatedly choosing a leaf and swapping its target token to it.  Thus, a certificate can simply be the list of swaps to execute. %

\subsection{Overview}

We reduce from Star STS and 
follow the same reduction plan as for the weighted case, building a tree with an ordering gadget and $m$ slot gadgets, each with  a nook for an item token.  However, there are several challenges. 
First of all, the swaps with item tokens are no longer free, so we need to  take 
\changed{their number}
into account. 
Secondly, we cannot  know  the number of item token swaps precisely since it will  depend on the number of swaps (between  $0$ and  $n$) required by  the  original Star STS instance. 
Our plan is to make this ``slack'' $n$ very small compared to the  total number of swaps needed for the constructed token swapping  instance.   To do this, we will make the total number  of swaps very big by replacing each of the non-item tokens $x_j$ and $y_j$ by a long path of non-item tokens called a ``segment.''  This raises  further  difficulties, because nothing forces the tokens in one segment to stay  together, which means that they  can sneak around and occupy nooks, freeing item tokens to  swap amongst themselves in unanticipated ways. 
To remedy this we add further ``padding segments''  to the construction.   The construction details are given in Section~\ref{sec:construction}.

Proving that our unweighted construction is correct is much more involved than in the weighted case.  There are two parts to the proof.

\begin{description}
\item[Part 1:] \customlabel{1}{part:1}
{\bf YES instance of Star STS $\to$ YES instance of token swapping.}
In Section~\ref{sec:ForwardDirection} 
we show that if there is a solution to an instance of Star STS then the 
constructed token swapping instance 
can be solved with at most $K$ swaps 
(where $K$ will be specified in the construction).
This solution uses exactly $H = K-n$ swaps to get the non-item tokens to their target positions, and uses between $0$ and $n$ additional swaps to get the item tokens to their target positions. Note that $H$ counts swaps of both item and non-item tokens and is more than just the cost of moving each non-item token along its shortest path.

\item[Part 2:] \customlabel{2}{part:2}
{\bf YES instance of token swapping $\to$ YES instance of Star STS.}
In Section~\ref{sec:BackwardDirection} 
we show that if the  constructed  token swapping instance can be solved with at most $K$ swaps then the original Star STS instance has a solution.
We  show  that $H$ swaps are needed to get the non-item tokens to their destinations.  We then show that with only $n$ remaining swaps,  the motion  of item tokens is so constrained that they must behave ``as intended'' and therefore correspond to swaps in the original Star STS instance.  

\end{description}

\subsection{Construction of token swapping instance}
\label{sec:construction}

Suppose we have an instance of Star STS where the star has center 0 and leaves $1, \ldots, m$ and each vertex has a token of its same label.  
We have a permutation $\pi$ of the tokens with $\pi(0) = 0$ 
and a sequence $s_1, \ldots, s_{n}$  with $s_j \in \{1, \ldots, m\}$  that specifies the allowed swaps.
As in the weighted case in Section~\ref{sec:NP-hard-weighted}, we will refer to the tokens of the Star STS instance as \defn{items} and the leaves as \defn{slots}.
We assume without loss of generality that every slot appears in the sequence (otherwise remove that slot from  the problem) and that no slot appears twice in a row in the sequence.

Construct a tree as in the weighted case except that each individual $x_j$ and $y_j$ is replaced by a  sequence of $k$ vertices (and tokens) with $k=(mn)^c$ for a large constant $c$, to be set later. 
Each such sequence of length $k$ is called  a
\defn{big segment}.
Refer to Figure~\ref{fig:plan-segments} where the tree is drawn with  the root in the middle, the ordering gadget to the left, and the slot gadgets to the right.  
We will refer to left and  right  as in the figure.
 The target ordering of  tokens within a big segment behaves as though the segment  just slides along the shortest path to its target, i.e.,  the left to right order of tokens  in a segment is the  same in the  initial  and  the final  configurations.

\begin{figure}[!ht]
\centering
\includegraphics[width=\textwidth]{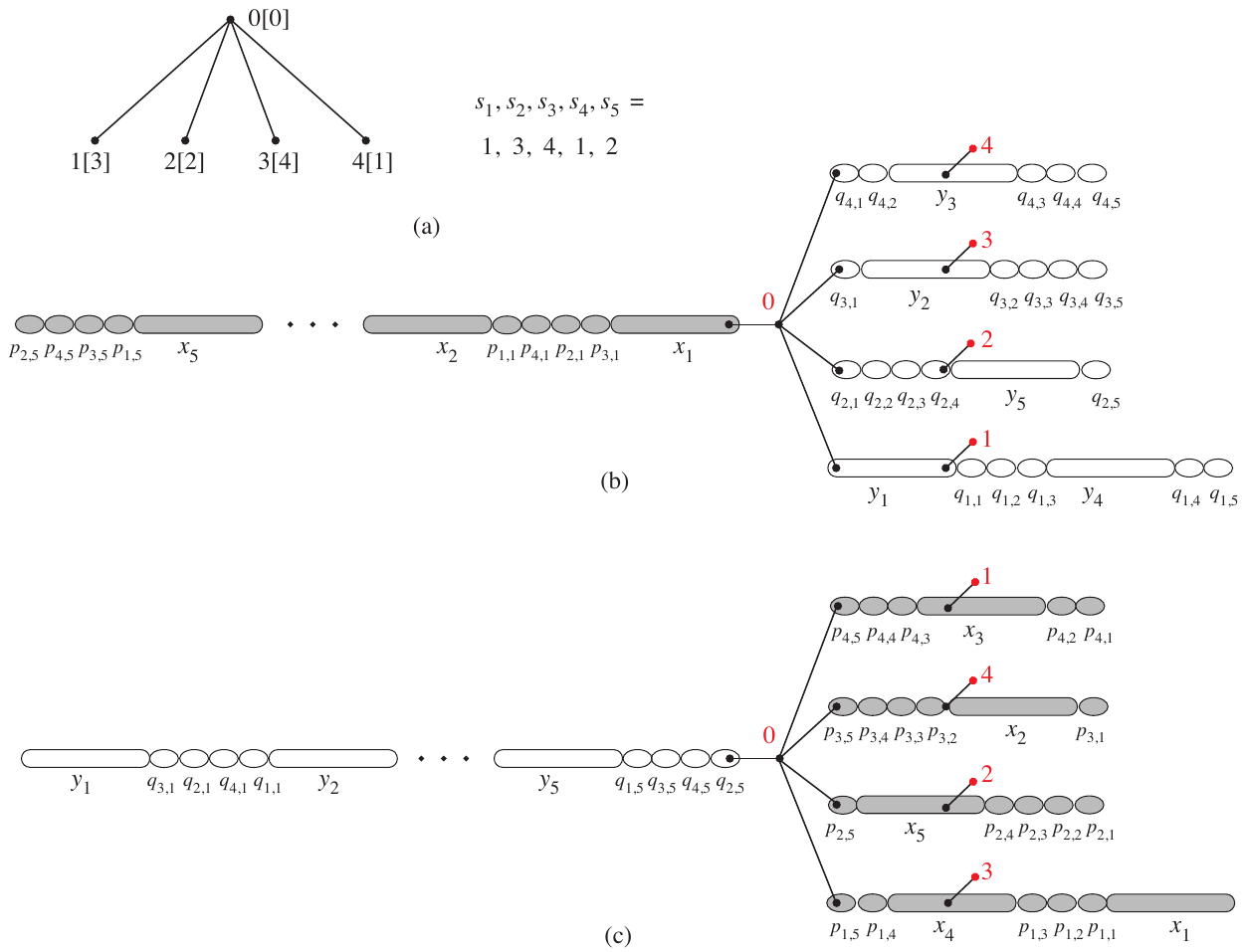}
\caption{(a) An instance of Star STS with $m=4$ slots and a sequence of length $n=5$.
(b) The corresponding instance of token swapping with the initial token  positions. The  root has item token $0$ (coloured red). The ordering gadget lies to the left of the root.  There are 4 slot gadgets to the right  of the root. 
A long oval indicates a big segment of length $k$, and a short oval indicates a padding segment of length $k' = k/n^8$.   Each nook vertex (coloured red)  is attached to the $k^{\rm th}$ vertex from the root along the slot gadget path.  Item tokens are coloured red;
non-slot tokens are in the segment ovals  coloured gray; and 
slot tokens are in  the segment ovals  coloured white.
(c) The target token positions. 
In the first round of  the ``intended''  solution,
big segments $y_1$ and $x_1$ first change places.  As $y_1$ moves left,  item token $0$ moves to nook parent 1 where it may swap with item token 1.
As $x_1$ moves right, the item token moves back to the root. 
Then segment $y_1$ moves to the far left of the ordering gadget and $x_1$ moves to the far right of the first slot gadget.   Next, padding segments $p_{3,1}$ and $q_{3,1}$ change places across the root and move to their target locations; then $p_{2,1}$ and $q_{2,1}$; then $p_{4,1}$ and $q_{4,1}$; and finally $p_{1,1}$ and $q_{1,1}$.
Note the ordering $q_{3,1},q_{2,1},q_{4,1},q_{1,1}$ of padding segments that lie to the  right  of $y_1$ in the final configuration---$q_{1,1}$ is last because $s_1=1$ and $q_{3,1}$ is first because $s_2=3$. 
}
\label{fig:plan-segments}
\vspace{-2em}
\end{figure}

The nook vertex in each slot gadget is attached to the vertex at distance $k$ from the root in the slot gadget, 
and this vertex is called the \defn{nook parent}.  The edge between the nook vertex and the nook parent is called the \defn{nook edge}, see Figure~\ref{fig:plan-segments}.
In the ``intended'' solution, the big segments leave the slot gadgets and enter the ordering gadget in the order $y_1, \ldots, y_n$.

Although we will not give details, 
it turns out that  
the construction so far allows ``cheating'' via interference between slot gadgets.
To prevent this, 
we add a total of %
$2nm$ 
\defn{padding segments} each of length 
$k' = k/n^8$. 
The  intuition is in the ``intended'' solution, after $y_j$ enters the ordering gadget,  one padding segment  from each slot gadget will enter the ordering gadget, and then $y_{j+1}$ will do  so.
We now give the details of the padding segments in the initial/final configurations of the slot/ordering gadgets.  
In the initial configuration there are $nm$ padding segments  $q_{i,j}, i= 1, \ldots,  m, j=1, \ldots, n$ in the slot gadgets, and $nm$ padding segments $p_{i,j}, i= 1, \ldots,  m, j=1, \ldots, n$ in the ordering gadget.   In the final configuration, the $q_{i,j}$'s are in the ordering gadget and the $p_{i,j}$'s are in the slot gadgets. 
In the initial configuration, 
slot gadget $i$, $i=1, \ldots, m$, contains $n$ padding segments $q_{i,j}, j=1, \ldots, n$.  They  appear in order from left to right with big segments mixed among them.  Specifically, if
 $y_{j'}$ is a big segment in slot gadget $i$ then $y_{j'}$ appears just before $q_{i,j'}$.
In the final configuration of the ordering gadget there are $m$ padding segments after each of the $n$ big segments.  The padding segments after $y_j$ are $q_{i,j}, i=1, \ldots, m$.  They appear in a particular left-to-right order:  $q_{s_j,j}$ is last, $q_{s_{j+1},j}$ is first, and the others appear in order of index $i$. Within one padding segment, the left-to-right ordering of tokens is the same in the initial and final configurations. 

The initial configuration of the ordering gadget is obtained from the final configuration by: reversing (from left to right) the pattern of  
big segments and padding segments; changing $y$'s to $x$'s, and changing $q$'s  to  $p$'s.  
Similarly, the final configuration of slot gadget $i$ is obtained from the initial configuration of slot gadget $i$ in the same way.

Let $A$ be the set of non-item tokens, and  for any token $t$, let $d_t$ be the distance between $t$'s initial and target  positions.
To complete the reduction, we will set $H$ to be $\frac{1}{2}\sum_{t \in A} (d_t + 1)$ 
and set the bound $K$ to be $H + n$.  The decision question is whether  this token  swapping  instance can  be  solved with  at most $K$ swaps. 
The reduction takes polynomial time.

\later{

In future sections we use the following definitions and  notation.
\clearpage
\begin{definition}
\label{defn:move}
\label{defn:partner}
\label{defn:basic}
\ 
\begin{itemize}
\item $T$ is the  set of tokens.  $T = A \cup I$.  $I$ is the set of $m+1$ \defn{item tokens} whose  initial and target positions are the nook vertices and  the root.  Item token 0 has initial and target position at the root. $A$ is the set of \defn{non-item tokens} and  consists of 
$nk + nmk'$ \defn{slot tokens} whose initial positions are in the slot gadget, and equally many \defn{non-slot tokens} whose initial positions are in the ordering gadget.  %
\item For a token $t \in A$, $d_t$ is the distance between  its initial and target positions.  
\item $H = \frac{1}{2}\sum_{t \in A} (d_t + 1)$.  $K = H+n$.
\item Big segments $x_j$ and $y_j$ are \defn{partners}.  Padding segments $p_{i,j}$ and  $q_{i,j}$ are \defn{partners}. 
\item Every swap consists of two \defn{moves}, where  a \defn{move} means one token moving across one edge. A \defn{left} move goes left in Figure~\ref{fig:plan-segments}, i.e., towards the root in a slot gadget or away from the  root in the ordering gadget.  
A \defn{right} move does the opposite. 
\end{itemize}
\end{definition}

\subsection{Part \ref{part:1}: YES instance of Star STS \texorpdfstring{$\to$}{-->} YES instance of token swapping}
\label{sec:ForwardDirection}
In this section we prove that if we start with a YES instance of Star STS then the constructed token swapping instance has a solution with at most $K=n+H$ swaps.
We first show how a solution to the Star STS instance gives an  ``intended'' solution to the constructed token swapping instance.
We follow the same idea as in the proof for the weighted case (Theorem~\ref{thm:weighted-NP-hard}).

\medskip\noindent
{\bf The ``intended''  solution.}
For $j=1, \ldots, n$, implement $s_j$ as follows.
\begin{itemize}
\item 
Move the item token $t$ currently  at the root %
along
the path of  slot gadget $s_j$ to the nook parent.    
Item token $t$ now has the opportunity to swap with the item token in the adjacent nook. Perform this swap if and only if swap $s_j$ was performed in the solution to the Star STS instance. 
Let $t'$ be the resulting token at the nook parent. 
Move big segment 
$y_j$ to its target position 
in  the ordering gadget.
Move big segment 
$x_j$ to its target position 
in the slot gadget $s_j$. This moves token $t'$ to the root.

\item 
Next, for each  slot gadget $i$, move the padding segment nearest the root out the slot gadget to its target position in the  ordering gadget and move its partner padding segment out of the ordering gadget to its target position in slot gadget $i$. These pairs of segment moves must be done in the order in which the padding segments $q_{i,j}$ appear in the final configuration (left to right). Note that item token $t'$ returns to the root after each of the $m$ pairs of segment moves.
\end{itemize}

It is straightforward to show that this solution is correct, i.e., that at the beginning of round $j$, big segment $y_j$ is next to   the root at the  left end of slot  gadget $s_j$ and big segment $x_j$ is next to the root at the  right end of the ordering gadget, and similarly, that after the movements of these big segments, the successive pairs of padding segments are in position next to the root when we are ready to move them.

At the  heart of the intended  solution is a subsequence of moves that are used to correctly re-position the non-item tokens while leaving the  item tokens  fixed.  We formalize this ``scaffold  solution'' for use in future  sections.

\medskip\noindent
{\bf  The ``scaffold'' solution.}  This is the same as the intended solution except that no swaps are performed between  item tokens.  Thus, item tokens $1, 2, \ldots, m$ remain  in their initial nook vertices, and item token $0$ returns to the root after each 
round (each loop  on $j$).

\medskip\noindent
{\bf Properties of the intended/scaffold solution.}

\begin{description}

\item{(P1)} Every non-item token $t$ travels on the path from its initial position to its target  position and makes $d_t$ moves.

\item{(P2)} There is no swap between two slot tokens nor between two  non-slot tokens. 

\item{(P3)}
The pairs of non-item tokens that swap with each other can be characterized with respect to the initial configuration as follows.
Let $t$ be a slot token that lies in big/padding segment $g$, and suppose that the partner of $g$ is the big/padding segment  $g'$.
Let $\pi_1$ be the ordered list of segments between  the root and $g$ (not including $g$).  Let $\pi'_1$ be the partners of the segments of $\pi_1$.
Let $\pi_2$ be the ordered list of segments from $g'$ (including $g'$) to the right end of the ordering gadget.  
Then the non-slot tokens that $t$ swaps with are the tokens that lie  in the segments $\pi'_1, \pi_2$.   See Figure~\ref{fig:intended-swaps}.
In particular, note that the number of tokens from $\pi'_1$ is less than the distance from $t$ to the root.

\item{(P4)} In the scaffold solution, token $0$ swaps with every non-item token.

\item{(P5)} The order in which the slot tokens reach the root is precisely the left-to-right order of their final positions in the ordering gadget. Similarly, the order in which the non-slot tokens reach the root is precisely the right-to-left order of their initial positions in the ordering gadget. 

\item{(P6)} The final positions of the slot tokens in the ordering gadget are such that between any big segment and any other segment from the same slot gadget, there is a segment from each of the other slot gadgets. Combining this with (P5) implies the following property. Consider a time interval during which at least one token from a big segment from a slot gadget $i$ and at least one token from any other segment from slot gadget $i$ reach the root. During this time interval, every token in at least one segment from each of the other slot gadgets reaches the root.

\end{description}

\begin{figure}[htbp]
\centering
\includegraphics[width=.4\textwidth]{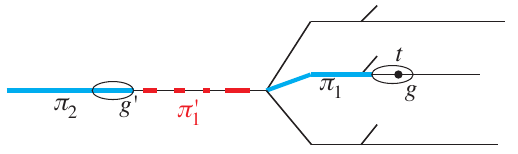}
\caption{In the intended/scaffold solution, slot token $t$ from segment $g$ swaps  with  the non-slot tokens with initial positions in $\pi'_1, \pi_2$.  Here $\pi'_1$ (coloured red)
consists of the segments that are the partners of segments in $\pi_1$.
}
\label{fig:intended-swaps}
\end{figure}

\medskip\noindent
{\bf Bounding the number of token swaps for a YES instance of Star STS.} 

\begin{lemma}
\label{lem:forward}
Given a YES instance of 
Star STS,
the total number of swaps in the corresponding token swapping instance is at most $K$. 
\end{lemma}
\begin{proof}
Recall  that $H = \frac{1}{2}\sum_{t \in A} (d_t + 1)$ and $K = H + n$.
We show that the  number of swaps used by the intended  solution is at most
$K$.

First consider the scaffold solution. The number of moves made by a  non-item token $t$  in  the  scaffold solution is $d_t$  by (P1).  In addition, item token  0 swaps once with each non-item token by (P4), for a total of $|A|$ moves by item tokens.  There are no other moves.  Thus the total number of moves in the scaffold solution is $\sum_{t \in A} (d_t +1)$, and the number of  swaps is half that, i.e., $H$.

The intended solution has additional swaps between pairs of item tokens---these are the swaps that occur on the nook edges.  They  correspond exactly to the  swaps $s_j$ that are chosen for the Star STS solution, thus, %
at most $n$.

In sum the total number of swaps is at most $H + n  = K$.
\end{proof}

\subsection{Part \ref{part:2}: YES instance of token swapping \texorpdfstring{$\to$}{-->} YES instance of Star STS}
\label{sec:BackwardDirection}

In  the  previous section we showed that if we have a YES instance of Star STS then the corresponding token swapping instance can be solved in at most $K$ swaps. The goal of this section is to prove the reverse claim. %

\begin{lemma}
\label{lem:backward}
Consider an  instance of Star STS whose corresponding token swapping instance has a solution using at most $K$ swaps.  Then the Star STS instance is a YES instance.  
\end{lemma}

Note that Lemmas~\ref{lem:forward} and \ref{lem:backward} combined show the correctness of our reduction, which finishes the NP-completeness proof.

Throughout the  section we assume that in the solution to the token swapping instance no   pair of tokens swaps more than once, which we justify as follows:

\begin{claim}
\label{lem:unique}
If a pair of tokens swaps more than once in a token swapping solution, then there is a shorter solution. 
\end{claim}
\begin{proof}
Suppose tokens $x$ and $y$ swap more than once in a sequence of swaps.
Eliminate the last two swaps of $x$ and $y$ and exchange the  names $x$ and $y$ in the subsequence between the two removed swaps.  The resulting shorter sequence yields the same final token configuration.
\end{proof}

Our argument will be based on the number of  moves rather than the  number of swaps. 
Recall from Definition~\ref{defn:move} that every  swap  consists of two  moves.
By   assumption we have a solution that uses at most $K$ swaps, which is at most $2K$ moves
and $2K = 2H + 2n = 2(\frac{1}{2}\sum_{t \in A} (d_t + 1)) + 2n = \sum_{t \in A} d_t + |A| + 2n$. 
Clearly,  we  need at  least $d_t$ moves to get each non-item token $t \in A$ to its target location.  In the first part of the argument (in Subsection~\ref{sec:Back moves for any instance}) we identify $|A|$ additional  moves (of item or  non-item tokens) that are dedicated to the task of getting non-item tokens to their target locations. This leaves at most $2n$ moves unaccounted for. 
In the second part of the argument (in Subsection~\ref{sec:Back moves for NO instance}),  we show that with so few remaining moves, the movement of item tokens is very constrained and corresponds to a solution to the original Star STS instance.

Before starting the rather technical subsections, we give more detail on these two high  level aspects: how  we classify and account for moves; and how we convert the  constrained movement of item tokens into a solution to the original Star STS instance.  %

\paragraph*{Accounting for moves.} 

Rather than arguing about all the moves, we will argue about special moves called ``contrary moves.''  For slot and non-slot tokens, the ``contrary moves'' go contrary to the direction the token should go, i.e., for a slot token---which should go  left---a contrary move is a right move, and vice versa for non-slot tokens.  We formalize this and extend the definition to item tokens as follows.

\begin{definition}
\label{defn:contrary-move}
\label{defn:main-path}
A \defn{contrary} move  of a slot token is  a right move.  See Figure~\ref{fig:moves}.
A \defn{contrary} move  of a non-slot token is a left move.
A \defn{contrary} move of an  item token is a  right move (this is somewhat arbitrary).
Let $c_t$  be the number of contrary moves made by token $t$.
The \defn{main path} of a slot token $t$ is the path in the tree between the rightmost vertex of $t$'s slot gadget and the leftmost vertex of the ordering gadget.
\end{definition}

The concept of a main path will  be used later in the proof. 
Observe that every  slot token $t$ must make $d_t$ left  moves along  its main path.  

\begin{figure}[htb]
    \centering
    \includegraphics[width=.55\textwidth]{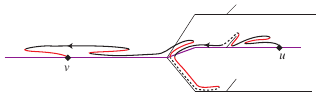}
    \caption{Classifying the moves as a  slot token moves from its  initial position  $u$ to its target position $v$.  \defn{Contrary} (right) moves are coloured red. Left moves are coloured black, with the solid portions occurring on the \defn{main path} (coloured purple).
    }
    \label{fig:moves}
\end{figure}

\begin{claim}
\label{claim:contrary-moves}
The total number of  moves in the solution is $\sum_{t \in  A} d_t + 2 \sum_{t \in T} c_t$.
\end{claim}
\begin{proof} Note that the first sum is over non-item tokens and the second sum is over all tokens.  Every slot token $t$ must make $d_t$ left moves, one on each edge of the path from its initial to final position. 
The other moves made by  token $t$ come  in pairs, where  each pair consists of a left  move and a right (contrary) move on some edge $e$. Thus   
the total number  of moves made by a slot token $t$ is $d_t + 2c_t$.  
For a non-slot token $t$ the same argument (with directions reversed) shows that the total number of moves is $d_t + 2c_t$.

For item tokens, observe that the set of  nodes of the tree that contain item tokens is the same in the initial and the target configurations.  This implies that, if an item token makes a contrary (right) move on an edge $e$, then some (possibly other)  item token must move left on $e$.  Thus the total number of moves made by  item tokens is $2\sum_{t \in I} c_t$.
Putting  these together, the total number   of  moves in the solution is 
$\sum_{t \in  A} d_t + 2 \sum_{t \in T} c_t$.
\end{proof}

We obtain a bound on the number of contrary  moves.

\begin{lemma}
\label{lemma:contrary-move-bound}
A token swapping solution with at most $K$ swaps has at most $ |A|/2 + n$ contrary  moves.
\end{lemma}
\begin{proof}
An upper bound of $K$ swaps is an upper bound of $\sum_{t \in A} d_t + |A| + 2n$  moves.   By Claim~\ref{claim:contrary-moves}, the number of moves is $\sum_{t \in  A} d_t + 2 \sum_{t \in T} c_t$.   Thus $\sum_{t \in T} c_t \le |A|/2 + n$. 
\end{proof}

In Subsection~\ref{sec:Back moves for any instance} we  will identify in the solution a set $\cal M$ of  $|A|/2$ contrary moves that are associated with getting slot tokens to their target locations, one contrary  move per slot token. This leaves at most $n$ contrary moves outside of $\cal M$.

In Subsection~\ref{sec:interface} we make a careful catalogue of which contrary  moves are included in $\cal M$, thus providing an interface to the final  Subsection~\ref{sec:Back moves for NO instance}.  In that final subsection we argue that if the big segments and padding segments are long enough (i.e., for 
sufficiently large $k$),
and given the bound of at most $n$ contrary moves outside $\cal M$, the movement of the tokens is constrained to be ``close'' to the intended solution. 
We do this by  proving various properties.  For example, to give the flavour  of these arguments, we prove that to get two item tokens into the ordering gadget at the same time takes more than $n$ contrary moves outside $\cal M$.  A major challenge in proving all these properties is showing that the more than $n$ contrary  moves we count are indeed disjoint from $\cal M$.  
Finally, the properties allow us to convert the movement of item tokens into a solution to the original Star STS instance as we describe next.

\paragraph*{Converting a token swapping solution to a star STS solution.}

Recall  that at any point in time during the intended solution, there is an item token in each nook plus one ``free'' item token (initially the one at the root) that moves around.  The free item token moves to a nook parent and may  swap with the item token $t'$ in the nook, in which case $t'$ becomes the new free item token.
We cannot hope for such structure in our solution since  many item tokens can be  outside the nooks.  However, we can still identify one ``free'' item token and one item  token per slot gadget. 

\begin{definition}
At any   point  in time during the swap sequence the \defn{free} item token is the item token that most recently was at the root.  
\end{definition}

In the initial and target configurations, item token $0$ is the free item token. We will  prove that 
apart from 
the free item token, there is always exactly one item token in each slot gadget. 
How can the free item token change? 
An item token from a slot gadget, say  slot  gadget $i$, must reach the root to become the  new free item token. 
Because of the property just stated (that, apart from the free item token, there is always exactly one item token in each  slot gadget),  
at this point in time, the old free item token must be in slot gadget $i$ 
(it  may have entered the slot gadget at the most recent swap or at some previous time). 

\begin{definition}
The \defn{exchange} sequence $\chi$  is constructed from the changes in the free item token as follows:
initialize $\chi$ to  the empty string; 
when the item token that was just in slot gadget $i$ becomes  the free item token,  append $i$ to $\chi$. 
\end{definition}
At the end of Subsection~\ref{sec:Back moves for NO instance} we will  prove  the following statement. 

\begin{lemma} 
\label{lem:free-item-token}
If $k$, the length of a big segment, is sufficiently  large, then
the  free item token  and  the exchange sequence $\chi$ satisfy the following properties.
\begin{enumerate}
\item At any point in time, apart from the free item token, there is exactly one item token in each slot gadget.
\item $\chi$ is a  subsequence of the sequence of swaps given as input to the Star STS instance.
\end{enumerate}
\end{lemma}

With this lemma  in hand, we can prove the  main result of this section:

\begin{proof}[Proof of Lemma~\ref{lem:backward}] 
We get a direct correspondence between  the token  swapping solution and a solution to the Star STS instance: 
the free item token corresponds to the item at the root in the Star STS instance; the token in slot gadget $i$ corresponds to the token at slot $i$ in the Star STS instance;  and 
an  exchange $i$ in $\chi$ corresponds to a swap  between the token at the root and the token at slot $i$.
By Lemma~\ref{lem:free-item-token}, the exchange sequence  $\chi$ corresponds to a subsequence of the allowed swap sequence. 
Thus the solution to the token swapping instance provides a solution to the Star STS instance.
\end{proof}

\subsubsection{Contrary moves to get non-item tokens  to their destinations}
\label{sec:Back moves for any instance}

In this section we will show that a solution to  the constructed token  swapping instance must contain $|A|/2$ contrary  moves dedicated to the  task of getting the  non-item tokens to their target destinations.  
We do this by building a set $\mathcal  M$ that contains one distinct contrary move per slot token.  Since  there are $|A|/2$ slot tokens, $\mathcal M$ will have $|A|/2$  contrary  moves. 

For intuition about  why  contrary moves are necessary, observe that in the intended/scaffold solution every  slot token swaps with an item token as  it travels leftward from its initial position to its final position, and  
this swap is a right (contrary)  move for  the item token.  Of course, a general solution need  not have this property.  However, note that in the initial configuration item token 0 lies in the path between the  initial and final positions of any slot  token $t$.  The informal idea is that if item token 0 swaps  off this path, then some  other ``intruder'' token $r$  takes its place, and we will show that $r$ must make a contrary move.

We now make this intuition  formal.
Recall  from Definition~\ref{defn:main-path} that for a slot token $t$ with initial  position in slot gadget $i$ its main path  is the path in the tree between the rightmost vertex of $t$'s slot gadget and the leftmost vertex of the ordering gadget.

\begin{definition}
An \defn{imposter} for a slot token $t$ is a token $r$ such that: 
(1) $t$ makes  a  left  move on its main path while swapping with $r$;
and  (2) $r$ is either an item token or a token  that $t$ does not swap with in  the scaffold solution. 
\end{definition}

\begin{claim}
\label{claim:imposter}
Every slot token $t$ has at least one imposter.
\end{claim}
\begin{proof} 
Token $t$ must make at least $d_t$ left moves along  its  main   path.
(Recall  that $d_t$ is the distance from $t$'s initial position to  its target position.)
The number of non-item tokens that $t$ swaps with in the scaffold solution is exactly one less than $d_t$.  Furthermore,  by Lemma~\ref{lem:unique} $t$ does not swap  with the same token more than once.   Thus one of $t$'s  left moves  on its main  path must be a swap with an imposter. 
\end{proof}

\paragraph*{Choosing contrary moves.}

Consider a solution to the  token swapping   instance.
(In  fact, we will only  use the property that the sequence  of swaps gets every non-item token to  its destination.) 
We will give a procedure that iterates through  the slot tokens $t$, and chooses one of $t$'s imposters, $r$, together with a unique contrary move of $r$. 
The contrary move of $r$ will not always occur  when $r$ swaps with  $t$.
Our procedure will add the chosen contrary move to a set $\mathcal{M}$ (initially empty).  Since there are $|A|/2$ slot tokens  and  we choose a unique contrary  move for each one, the final size of $\mathcal{M}$ will be $|A|/2$.

For each slot  token  $t$ consider the left moves of $t$ along  its main  path.  Check  the   following cases in  order.

\noindent
{\bf Case 1.} If $t$ swaps with an  item token, choose one such item token $r$. 
Then $r$ is an  imposter for $t$. Since $t$ moves left,  $r$ moves right which is a contrary move for $r$.  Add this contrary move of $r$ to $\mathcal{M}$.  

\noindent
{\bf Case 2.}
If $t$ swaps with another slot token,  choose one such slot token $r$. Then $r$ is an imposter for $t$ since no two slot tokens swap in the scaffold solution by (P2). 
Since $t$ moves left, $r$ moves  right  which is a  contrary  move for $r$.
Add this contrary move of $r$ to $\mathcal{M}$.  

We may now suppose that $t$ only swaps with non-slot tokens as it makes left  moves  on its main path. 
By Claim~\ref{claim:imposter}, $t$ swaps with at least one imposter.  
We next check \emph{where} this happens.  Suppose that $t$'s initial position is in slot gadget $i$.

\noindent
{\bf Case 3.}  If $t$ swaps with an imposter in slot gadget $i$, choose one such imposter $r$.   Suppose the  swap happens on edge $e$. Then $r$ moves  right.  In the scaffold solution  $t$ swaps with every non-slot token whose target destination is in slot $i$.  Thus $r$'s target destination is outside slot gadget $i$ and $r$ must move left on $e$ at some later time. 
As the contrary  move, choose the  earliest left move of $r$ on  edge $e$ after the  $r$-$t$ swap, and add this contrary move to $\mathcal{M}$.

We may now suppose that $t$ swaps with imposters only in the  ordering gadget, and  all  these imposters are non-slot tokens.

\noindent
{\bf Case 4.} If $t$ swaps with an imposter $r$ that was at the root  at some time before the $r$-$t$ swap, choose one such $r$.
Suppose the $r$-$t$ swap occurs on edge $e$.  Then $r$ moves  right on $e$, and must have moved left on $e$ between the time it was at the root and the time it swapped with $t$.  
As the contrary  move, choose the most recent left move of $r$ on  edge $e$ before the  $r$-$t$ swap, and add this contrary move to $\mathcal{M}$. 

We  are left with the set of tokens  $t$ that only have imposters of the following special form.

\begin{definition} A non-slot token  $r$ is a  \defn{non-leaving imposter} for slot token  $t$ if:
$t$ makes a left move in the ordering  gadget while swapping with $r$; (2) $r$ does not swap with $t$  in the  scaffold solution;  and  (3) $r$ has remained inside the ordering gadget up  until the  $r$-$t$ swap. 
\end{definition}

Let $S$ be the set of slot tokens $t$ such that every imposter for $t$ is a non-leaving imposter.
Note  that by Claim~\ref{claim:imposter}, every  $t \in S$  has a  non-leaving imposter.  In order  to complete the procedure to build $\cal M$, we must choose a non-leaving imposter and a contrary move for each token in $S$.

In later sections we will need to limit the number of non-slot tokens that are chosen as non-leaving imposters so we will choose the non-leaving imposters for tokens  in  $S$  according to Lemma~\ref{lem:FewImposters} below.  This has no effect on the rest of the current argument, so for now, just assume that we  have some assignment of non-leaving imposters to the tokens in $S$.  Let $S(r) \subseteq S$ be the  tokens  of $S$ that are assigned to imposter $r$.  Let $S_i(r) \subseteq  S(r)$ be the tokens of $S(r)$ whose initial positions are in  slot gadget $i$.  Note  that the sets $S_i(r)$ partition $S$.
Let $D_i$ be the non-slot tokens whose final positions are in slot gadget $i$.

\begin{lemma}
\label{lemma:enough-contrary-moves}
Token $r$ makes at least $|S_i(r)|$ contrary  moves while swapping with tokens of $D_i$.
\end{lemma}

With this lemma  in hand, we  can  complete the  construction of  $\mathcal M$.  

\noindent
{\bf Case 5.} Token $t$ lies in $S$, which we have  partitioned into the sets $S_i(r)$.    For each $i$ and $r$ such that $S_i(r)$ is non-empty, choose  $|S_i(r)|$ contrary  moves of $r$ where $r$ swaps  with tokens of  $D_i$ and add them to $\mathcal M$.  Note that this adds a total of $|S|$ contrary  moves to $\mathcal M$ because the contrary moves of $r$ for different $i$'s are disjoint since the  sets $D_i$ are disjoint.

\begin{proof}[Proof of Lemma~\ref{lemma:enough-contrary-moves}]
Consider the  initial  configuration  of tokens.
Let $t_1, \ldots,  t_f$ be the tokens of $S_i(r)$  ordered by  their (initial) positions in slot gadget  $i$ with $t_1$ closest to the root. Then $|S_i(r)| = f$ so we need to show that $r$  makes $f$ contrary moves while swapping with tokens of $D_i$.
Let $g$ be the big/padding segment containing  $t_1$, and let $g'$ be the partner of $g$.  Let $h_1$ and $h_f$ be the distances from the root to $t_1$ and $t_f$, respectively.

Since $r$ is an imposter for $t_1$, there was no swap between $r$ and $t_1$ in the scaffold solution. 
But in the scaffold solution, 
by  property (P3),  $t_1$ swaps with every token to the left of (and including)  $g'$.  Therefore 
$r$ must lie to the right of $g'$. 
Also, by property  (P3),  fewer than $h_1$ tokens of  $D_i$ lie to the right of  $g'$.  Thus, fewer than $h_1$ tokens of  $D_i$ lie to the right of $r$.

Let $Z$ be the  set of non-slot tokens that swap with $t_f$ during left moves of $t_f$ on its main path in slot gadget $i$. 
By Lemma~\ref{lem:unique}, $t_f$ does not swap with the same token  twice, so $|Z| \ge h_f$.
 No  token of $Z$ is  an imposter for $t_f$ because $t_f$ only has non-leaving imposters. Thus  $Z$ consists  of non-slot tokens whose destinations are in slot gadget $i$, i.e., $Z \subseteq D_i$. Also note that $r \notin Z$ since $r$ stays in  the ordering gadget until it swaps with $t_f$.  

Let $Z'$ be the  tokens of $Z$ that lie to the left of $r$ in the initial configuration.
Since there are less than $h_1$ tokens of $D_i$ to the right  of  $r$, we have $|Z'| >  h_f - h_1 \ge f-1$. 
Thus $|Z'| \ge f$.
Finally, we claim that all the tokens of $Z'$ swap with $r$, and do so during contrary (left) moves of $r$.  This  is   because $r$ remains  inside the ordering gadget until it swaps with $t_f$, but all the tokens of $Z'$ must exit the ordering gadget so they can swap with $t_f$ on its journey to the root.   
Thus $r$ makes at least $f$ contrary  moves while swapping with tokens of $D_i$, as required.
\end{proof}

\paragraph*{Uniqueness of contrary moves.}
We  claim  that the  set $\mathcal M$ constructed above contains $|A|/2$ distinct contrary moves.
First note that we added a contrary move to $\mathcal M$ for each slot token, and there are $|A|/2$ slot tokens.  It remains to  show that we never added the same contrary move to $\mathcal M$ multiple times.
Consider a contrary move in $\mathcal{M}$ of a token $r$.  We can uniquely identify when it was added to $\mathcal{M}$ as follows. 

\begin{itemize}
\item If $r$ is an item  token then the move was added to $\mathcal{M}$ in Case 1, and the associated slot token $t$ is the one that $r$ swaps with  during this contrary move.  
\item If $r$ is a slot token then the move was added to $\mathcal{M}$ in Case 2, and $t$  is the token  that $r$ swaps with.
\item If $r$ is a non-slot token and its contrary (left) move  is on an edge $e$ in a slot gadget, then the move was added to $\mathcal{M}$ in Case 3, and in $r$'s most recent right move on edge $e$ it swapped with token $t$.
\item If $r$ is a non-slot token and its contrary (left) move  is on an edge $e$ in the ordering  gadget  and $r$
visited the root before this contrary move, then the move was added to $\mathcal{M}$ in Case 4, and in $r$'s most recent  right  move on edge $e$  it swapped with  token $t$.
\item    Finally, if $r$ is a non-slot token and its contrary (left) move  is on an edge in the ordering  gadget  and $r$
did not visit the root before this contrary move,
then the move was added  to  $\mathcal M$ in Case 5. As  already noted, Case 5 added $|S|$ distinct contrary  moves to $\mathcal M$ to pay for the tokens $t \in S$.  
\end{itemize}

Thus, we have shown what we set out to show at the beginning of this section: a solution to our token  swapping instance must have $|A|/2$ contrary  moves dedicated to the  task of getting the  non-item tokens to their target destinations. However, to argue that the contrary moves 
that we will count
in future sections are disjoint from those counted in this section, we need to be more precise about how imposters are chosen.

\paragraph*{How to choose non-leaving imposters.}
We will prove the following lemma, which limits the number of non-leaving imposters. %

\begin{lemma}\label{lem:FewImposters}
We can choose non-leaving imposters so that the total number of unique non-leaving imposters is less  than  or equal to %
$n(m+1)$.    
\end{lemma}
\begin{proof}
Recall that $S$ is the set of slot tokens $t$ such that every imposter for $t$ is  a non-leaving imposter.  Tokens  of $S$ are exactly  the tokens for which we must choose a non-leaving  imposter. 
We show how to choose non-leaving imposters so that 
if two tokens in $S$ were initially in the same segment (the same $y_j$ or $q_{i,j}$) then they   are assigned the same imposter.
This proves that the number of unique non-leaving  imposters is less than or equal to the number of segments initially in  the slot gadgets. 
There are $n$ big segments and $nm$ padding segments in the slot gadgets,  so the  bound is $n(m+1)$, as required.

Let  $G$ be the set  of slot tokens from one segment.
Observe that every  token  in $G$ swaps with the same set  of tokens in the scaffold solution.
Let $G'$ be $G \cap S$.
Then every  $g \in G'$ has a non-leaving imposter and does not have any other kind of  imposter.
We must show that all tokens in $G'$ can share the same non-leaving imposter.

Consider the last time each token in $G'$ enters the ordering gadget, and let $t$ be the last token to do so. 

\begin{claim} $t$ swaps with an imposter after it enters the ordering gadget for the last time.
\end{claim}
\begin{proof}
Let $d_t = h  + h'$ where $h$ is the distance from $t$'s initial position to the root and $h'$ is the distance from the root to $t$'s target position.  Then  $t$ must make at least $h$ left moves  before it reaches the root for the first time, and must make at least $h'$ left moves after it reaches the root for the last time.  These moves all   occur on $t$'s  main path.
As noted in the proof of Claim~\ref{claim:imposter}, the number of non-item tokens that $t$ swaps with in the scaffold solution is $d_t -1$.  Since $t$ has already used up at least $h$ of them on its way to the root (where it met no  imposters along the way), there are at most $h' -1$ of them left.   Therefore $t$ must swap with an imposter during a left move after $t$ enters the ordering gadget for the last time. 
\end{proof}

Let $r$ be an imposter that $t$ swaps with after it enters the ordering gadget for the last time.   Note that $r$ must be a non-leaving imposter, since $t$ did  not have any  other kind of imposter.    

We claim that $r$ is also a  non-leaving imposter for every token in  $G'$. 
Since the tokens in  $G'$ only have non-leaving  imposters, it suffices to show that $r$ is an imposter for every token $g \in G'$. 
To do this, we prove that: 
(1) $g$ swaps with $r$ while $g$ makes a  left move on an edge of the ordering gadget (note that the ordering gadget  is  part of every  slot  token's main  path); 
and (2) $g$ does not swap with $r$  in  the scaffold solution.

By definition of $t$, when $t$ enters the ordering gadget for the last time, every other token in $G'$ is in the ordering gadget and will not leave again. 
After $t$ enters the  ordering gadget it swaps with $r$, and after that, $r$ reaches the root for the first time.
Thus, before $r$ reaches the root, it must swap with every token $g$ in $G'$ with $r$ moving right and $g$ moving left.
This proves~(1).

Since all tokens in $G$ swap with the same set of tokens  in  the scaffold solution,  the fact that $r$ does not swap with $t$ in  the scaffold solution means that $r$ does not swap with any token in $G'$ in the scaffold solution. This proves (2). 
\end{proof}

\subsubsection{Interface between Sections~\ref{sec:Back moves for any instance} and~\ref{sec:Back moves for NO instance}: Characterization of contrary moves}\label{sec:interface}

In this section we will characterize the contrary moves that we counted in Section~\ref{sec:Back moves for any instance}, so that we can argue that the contrary moves %
we will count in forthcoming Section~\ref{sec:Back moves for NO instance}  are disjoint from those already counted. 

We say that the \defn{initial gadget} of a non-item token $t$ is the gadget (either the ordering gadget or one of the slot gadgets) containing $t$'s initial position. Similarly, the \defn{final/destination gadget} of $t$ is the gadget containing $t$'s final position. 

Recall that $\mathcal{M}$ is the set of contrary moves counted in Section~\ref{sec:Back moves for any instance}. The following observation follows directly from the construction of $\mathcal{M}$. These correspond to {\bf Cases 1-5} in Section~\ref{sec:Back moves for any instance}.

\begin{observation}\label{obs:purple}
$\mathcal{M}$ can be written as $\mathcal{M}=M_1\cup M_2 \cup M_3 \cup M_4 \cup M_5$, where each $M_i$ is a set of contrary moves that are \emph{not} into or out of the nook, with the following properties:
\begin{enumerate}
\item $M_1$ is a set of contrary (right) moves of item tokens.
\item $M_2$ is a set of contrary (right) moves of slot tokens with the following properties: 
\begin{itemize}
    \item Each contrary move in $M_2$ occurs in a slot gadget $i$ where the \emph{other} token in the swap is a slot token $t$ whose initial position is in slot gadget $i$.
    \item $M_2$ contains at most one such contrary move for each slot token $t$. 
\end{itemize} \label{purp2}
\item $M_3$ is a set of contrary (left) moves of non-slot tokens with the following properties:
\begin{itemize}
    \item Each contrary move in $M_3$ of a token $s$ occurs on an edge $e$ in a slot gadget $i$ that is not $s$'s destination gadget.
    \item In $s$'s most recent right move on edge $e$, $s$ swapped  with a  slot token $t$ whose initial gadget is slot  gadget $i$.
    \item $M_3$ contains at most one such contrary move for each slot token $t$.
\end{itemize} \label{purp3}
\item $M_4$ is a set of contrary  (left)  moves of non-slot tokens with the following properties:
\begin{itemize}
    \item Each contrary move in $M_4$ occurs in the ordering gadget.
    \item For each contrary move in $M_4$, the token $s$ that performs the contrary move was at the root at some point previously.
\end{itemize}  %
\item $M_5$ is a set of contrary (left) moves of a set of at most $n(m+1)$ non-slot tokens, that occur in the ordering gadget. The bound of $n(m+1)$ follows from Lemma~\ref{lem:FewImposters}. 
\end{enumerate}
\end{observation}

The following result follows directly from Observation~\ref{obs:purple} and is an explicit description of various contrary moves that are \emph{not} in $\mathcal{M}$. We list these moves here for easy reference from Section~\ref{sec:Back moves for NO instance}.

\begin{corollary}\label{obs:interface}$ $
\begin{enumerate}
    \item Any contrary (left) move of a non-slot token that is in its destination slot gadget is not in $\mathcal{M}$.\label{int1}
     \item Any contrary (right) move of a slot token $s$ swapping with a token $t$ that occurs in slot gadget $i$, such that $i$ is $s$'s initial gadget and $i$ is \emph{not} $t$'s initial gadget, is not in $\mathcal{M}$.\label{int6}
       \item Any contrary move of a non-item token into or out of the nook is not in $\mathcal{M}$.\label{int4}
    \item Any contrary move of a non-item token that is in a slot gadget that is not its final gadget, such that the other token in the swap is an item token, is not in $\mathcal{M}$.\label{int5}
    \item If a slot token $t$ has $\lambda$ many swaps in either $t$'s initial slot gadget or the ordering gadget, in which the other token in each of these swaps is a slot token doing a contrary (right) move, then at least $\lambda-1$ of these contrary moves are not in $\mathcal{M}$.\label{int2}
    \item If a non-slot token $t$ has $\lambda$ many swaps in the ordering gadget, in which the other token in each of these swaps is doing a contrary (left) move and is a non-slot token that has not yet reached the root, then at least $\lambda-n(m+1)$ of these contrary moves are not in $\mathcal{M}$.\label{int3}
   
\end{enumerate}
\end{corollary}

\subsubsection{Moves to get item tokens to their destinations}
\label{sec:Back moves for NO instance}

Recall that we assume that the token swapping instance that was constructed from an instance of Star STS 
has a solution with at most $K$ swaps. 
By Lemma~\ref{lemma:contrary-move-bound}, the solution has at most  
$|A|/2 + n$ contrary moves.  
Previously, we found a set $\mathcal M$ of $|A|/2$ contrary moves in the solution  that  are required to  get the non-item tokens to their destinations.  This leaves at most $n$ contrary moves to complete the job of getting the item tokens to their destinations.  In this section we will show that if the big/padding segments are very long, then 
the movement of tokens is constrained enough that we can prove the structural properties of 
Lemma~\ref{lem:free-item-token}.

Recall from the construction that the length of a big segment is $k=(mn)^c$ for a large constant $c$. Set $c=25$. 
We will assume $n,m > 5$. %
The following observation will be useful. 

\begin{observation}\label{obs:contrary}
The total number of contrary moves in the token swapping solution is less than $k^{1+1/c} =kmn$, and the number of contrary moves not in $\mathcal{M}$ is less than $k^{1/c} = mn$.
\end{observation}
\begin{proof}
By Lemma~\ref{lemma:contrary-move-bound} the number of contrary moves is at most $|A|/2 + n$.
Plugging in $|A| = 2nk + 2nmk'$ and $k' = k/n^8$ this is $\le nk + nmk/n^8 + n < kmn = k^{1 + 1/c}$.  
The   number of contrary   moves not in $\mathcal M$ is at most $n$ and $n < mn = k^{1/c}$.
In both cases, we used $n,m > 1$.
\end{proof}

This section has four subsections:
\begin{enumerate}
\item In Section~\ref{sec:Segment Ordering}, we show that tokens are constrained in the sense that the segments reach the root in roughly the same order as in the scaffold solution. 
\item In Section~\ref{sec:Ordering Lemma}, we show that the tokens are constrained in the sense that each individual token cannot get too far out-of-order with its surrounding tokens.
\item In Section~\ref{sec:Item Tokens Stuck}, we use 
the above results 
to prove the key lemma 
that every slot gadget contains an item token ``near'' its nook at all times.
\item In Section~\ref{sec:Item and Segment Swap} we combine 
the above results 
to complete the proof of Lemma~\ref{lem:free-item-token}.

\end{enumerate}

\paragraph{Correspondence with the scaffold solution}\label{sec:Segment Ordering}%

In this section we will show that the movement of tokens in our solution is similar to that of the scaffold solution. The main lemma of this section is Lemma~\ref{lem:segbyseg}, which says that tokens reach the root in roughly the same order as in the scaffold solution. We start with a simple lemma that shows that not many tokens can
enter and exit one gadget (as made precise below).
We use ``gadget'' here to mean a slot gadget or the ordering gadget. Note that a token leaves the gadget it is in if and only if it reaches the root.  We will find it convenient to use both expressions when  we talk about the  movement of tokens. Note that a token can leave any gadget, i.e., reach the root, many times throughout the sequence of swaps.

\begin{lemma} \label{lem:swapNotOutsideOfGadget}
Let $S$ be a set of tokens, all  of which  are either slot or non-slot tokens.   Let $g$ be a gadget. 
Suppose that either of the following situations occur:
\begin{enumerate}
    \item At some time $\tau$, all tokens of $S$ are in gadget $g$, but  each token of $S$ was outside $g$ at some time before $\tau$ and at some time after $\tau$. 
    \item At two time points $\tau_1$ and $\tau_2$, all tokens of $S$ are in gadget $g$, but each token of $S$ 
    reached the  root at some time between $\tau_1$ and $\tau_2$.
\end{enumerate}
Then the tokens of $S$ make at least $\frac{1}{2}|S|^2$ contrary  moves, and consequently, $|S| \le k^{1/2 + 1/c}$.
\end{lemma}
\begin{proof}
Depending on whether %
the tokens of $S$
are slot or non-slot tokens, either moving the tokens out of, or %
into, the gadget $g$ will count as contrary moves. In either case the tokens of $S$ make at least 
$1+2+ \cdots + |S| > \frac{1}{2} |S|^2$ 
contrary moves.
If $|S| > k^{1/2 + 1/c}$,
then the number of contrary moves is $> \frac{1}{2}k^{1 + 2/c} > k^{1 + 1/c}$,
which contradicts Observation~\ref{obs:contrary}. 
\end{proof}

Recall that for any token $t$, its \defn{initial gadget} is the gadget that 
contains $t$'s initial position, and
its \defn{final gadget} is the gadget
that contains $t$'s final  position.
The next lemma says that if a large number of tokens have reached the root,
then  most of them have entered their  final gadgets for the  last time.

\begin{lemma}\label{lem:indest}
Let $S$ be a set of  at least $\ell=k^{1/2+2/c}$
tokens that are either all slot tokens or all non-slot tokens. If, at some point in  time, every token in $S$ has 
reached the root then every token in $S$ except for at most $\ell$ of them has entered its final gadget for the last time.
\end{lemma}
\begin{proof}
We consider the  situation  at %
the point in time stipulated in the lemma.
Let $R$ be the tokens of $S$ that have entered their final gadgets for the last time.   Our goal is to show that $|S \setminus R| \le \ell$.  
Observe that $S  \setminus R$ is a disjoint union of two sets: $S'$, the tokens of $S$ that are not in their final gadgets, and $S''$, the tokens of $S$ that are in their final gadgets, but will subsequently  leave (and then re-enter) those gadgets. Thus our goal is to show that $|S'| + |S''| \le \ell$.
We will use Lemma~\ref{lem:swapNotOutsideOfGadget}
to bound their sizes.

Consider $S'$.  
By the pigeonhole principle, at least one of the $m+1$ gadgets, say  gadget $g$, contains a set $X$ of at least $|S'|/(m+1)$ tokens from $S'$. Since $g$ is not the  final gadget for any token in $X$, they will all eventually leave gadget $g$. 
By  assumption, every token $t \in X$ has  reached the  root before the current time point, so $t$ has been outside $g$ at a previous time.
Thus by Lemma~\ref{lem:swapNotOutsideOfGadget}(1)
we have $|S'|/(m+1) \le  k^{1/2 + 1/c}$.

Next consider $S''$. 
Again by the  pigeonhole principle, at least one of the $m+1$ gadgets, say gadget $g$, contains at least $|S''|/(m+1)$ tokens from $S''$.
These tokens are currently in $g$, will each leave $g$ at some later time (by  definition of $S''$), and will all eventually be back in $g$ (since $g$ is their final gadget).
Thus by Lemma~\ref{lem:swapNotOutsideOfGadget}(2), we have
$|S''|/(m+1) \le k^{1/2 + 1/c}$.

Now $|S'| +  |S''| \le 2(m+1)k^{1/2 + 1/c} \le k^{1/2 + 2/c}$ for $m,n > 2$ giving  the  required bound.
\end{proof}

Next, we prove a related lemma which says that, for any  slot gadget $g$,  if many tokens with initial positions in $g$ have reached the root
then  almost as many tokens with final gadget $g$ have entered $g$ for the last time.
This lemma is not needed 
until  Section~\ref{sec:Item Tokens Stuck},
but we place it here because its proof heavily depends on Lemma~\ref{lem:indest}. 

\begin{lemma}\label{lem:exchange}
Let $S$ be a set of at least $k^{1/2+3/c}$ tokens that all have 
slot gadget $g$ as their initial gadget.
Let $\tau_1$ be a point in time where no token in $S$ has reached the root,
and let $\tau_2$ be a later point in time after every token in $S$ has reached the root. Then, between $\tau_1$ and $\tau_2$, at least $|S|-k^{1/2+3/c}$ many tokens (not in $S$) whose final gadget is $g$ enter gadget $g$ for the last time.
\end{lemma}
\begin{proof}
Let $S'$ be the set of tokens that were not in slot gadget $g$ at time $\tau_1$, but are in slot gadget $g$ at time $\tau_2$. We will show that $(i)$ 
$S'$ is big enough for the statement to hold and 
$(ii)$ 
few tokens of $S'$ subsequently leave gadget $g$.

For Statement $(i)$, we first apply Lemma~\ref{lem:indest} to $S$. Set $S$ has size at least $\ell =k^{1/2+2/c}$ and at  time $\tau_2$ all the tokens of $S$ have reached the root, so, by the lemma, all except at most $\ell$  of them have entered their final gadgets (which are different from $g$) for the last time.  Since all tokens of $S$ were in gadget $g$ at time $\tau_1$ and at least $|S| - \ell$ of them are not in gadget $g$ at time $\tau_2$, therefore $|S'| \ge |S| - \ell$ (since the missing tokens must have been replaced).   

For Statement ($ii$), we apply Lemma~\ref{lem:swapNotOutsideOfGadget}(1) to $S'$. 
All the tokens of $S'$ are in gadget $g$ at time $\tau_2$,  all of them were outside gadget $g$ at a previous time, and therefore, the number that can leave gadget $g$ after time $\tau_2$ is at most 
$k^{1/2 + 1/c}$ which is less than $\ell$.

Putting  these together, the number of tokens that have entered gadget $g$ and will remain there until the end is at least $|S| - 2\ell$.  Since $2\ell < k^{1/2 + 3/c}$, this gives the bound we want.  
\end{proof}

Recall that in the scaffold solution, the segments reach the root in a particular order, as specified by property (P5) of the scaffold solution. Now we show that, although not all tokens in our solution obey this particular order, many tokens do. In particular, the \defn{core} of each segment, defined below, behaves like the scaffold solution. 

\begin{definition} \label{def:core} 
Let $S$ be a segment. Order the tokens of $S$
by the first time they  reach the root.
The \defn{core}, $C(S)$, is the set of all tokens in $S$ except for the first and last $k^{1/2+3/c}$ in this  ordering. 
\end{definition}

\begin{lemma}\label{lem:segbyseg}
Let $S_1$ and $S_2$ each be 
a segment (either a big segment or a padding segment) of slot tokens  such that tokens of $S_1$ reach the root before tokens of $S_2$ in the scaffold solution. Then, in our token swapping solution, all tokens of $C(S_1)$ reach the root before any token of $C(S_2)$ does so.
\end{lemma}
\begin{proof}
Note that because $S_1$ and $S_2$ consist of slot tokens,  their final gadget is the ordering gadget. Also, by property (P5) of the scaffold solution, the final positions of the tokens in $S_1$ are to the left of tokens of $S_2$  in the ordering gadget.  

Consider the order in which  the tokens of $S_1  \cup S_2$ reach the root for the first time.   We must show that (all  tokens of) $C(S_1)$ come before (all tokens  of) $C(S_2)$ in this ordering.

Assume for the sake of contradiction 
that a token $t_2$ of $C(S_2)$ comes before a token $t_1$ of $C(S_1)$ in the ordering.   
Let $L_1  \subseteq S_1$ 
be the last $k^{1/2+3/c}$ tokens of $S_1$ 
in  the ordering (i.e., the tokens of $S_1$ after $C(S_1)$)
and let $F_2$  be the  first $k^{1/2+3/c}$ tokens of $S_2$ in the ordering (i.e.,  the tokens of $S_2$ before $C(S_2)$).
Then $F_2$ comes before $t_2$, which in turn comes before $t_1$ (which also comes before $L_1$).

Let $\tau$ be the time point just after all tokens of $F_2$ have reached the root.  At time $\tau$, none of the tokens of $L_1$ have reached the root, so none of them are in  the ordering gadget.
Apply Lemma~\ref{lem:indest} to $F_2$ and the time point  $\tau$. The lemma  gives us a set $F'_2 \subseteq F_2$  of tokens that have entered the ordering gadget for the last time at time $\tau$, where $|F'_2| \ge k^{1/2+3/c} - k^{1/2-2/c}$.

However, as already noted, 
the final positions of the tokens in $S_1$ are to the left of tokens of $S_2$.  In particular,  
the final positions  of the  tokens in  $L_1$ are to the left of the tokens of $F'_2$.
Thus, $L_1$ (which is outside the ordering gadget at time $\tau$) and $F'_2$ must pairwise swap in the ordering gadget. During these swaps, the tokens in $F'_2$ are performing contrary (right) moves. This amounts to a total of $|F'_2||L_1|\geq (k^{1/2+3/c}-k^{1/2-2/c})(k^{1/2+3/c})>k^{1+1/c}$ contrary moves, a contradiction  to Observation~\ref{obs:contrary}.
\end{proof}

From the definition of core (and construction of gadgets) we get the following bounds on the size of each core: 

\begin{lemma}\label{lem:noncore}
If $S$ is a segment, then 
$|C(S)|>k^{1-9/c}$. Also, for any slot gadget $i$, the total number of tokens initially in slot gadget $i$ that are \emph{not} in the core of any segment is less than $k^{1/2+5/c}$.
\end{lemma}
\begin{proof}
Each segment is of size at least $k/n^8$, and exactly $2k^{1/2+3/c}$ tokens in each segment are not in the core. Thus, $|C(S)|\geq k/n^8-2k^{1/2+3/c}>k^{1-9/c}$ since $c>22$.

The second statement follows from the fact that any slot gadget contains at most $2n$ segments ($n$ padding segments and at most $n$ big segments), and for each segment, exactly $2k^{1/2+3/c}$ 
tokens
are not in the core. Thus, there are at most  $2n \cdot k^{1/2+3/c}<k^{1/2+5/c}$ many tokens initially in slot gadget $i$ that are not in the core of any segment.
\end{proof}

\paragraph{Individual tokens stay roughly in order}\label{sec:Ordering Lemma}$ $\\
\ 

The main  lemma  of this section (Lemma~\ref{lem:OutOfOrder}) shows that no individual token that is currently in a slot gadget can get too far out of order with respect to its surrounding tokens (see below for a precise definition of `out of order').
More precisely, we show that in any solution to the token swapping instance that has at most $K$ swaps, no intermediate token configuration can contain a set of $\ell=k^{1/2+2/c}$ tokens that are all out of order with a token $t$ in a slot gadget, as otherwise there would be too many contrary moves. We use the fact (proved in Lemma~\ref{tooManyContraryMoves}) that by Observation~\ref{obs:contrary}, the token swapping sequence can neither use $\frac{\ell}{2^4(m+1)}$ %
contrary moves that are disjoint from moves accounted in $\mathcal{M}$ (c.f.~Corollary~\ref{obs:interface}), nor it can use $\frac{\ell^2}{2^53^2(m+1)^2}$ %
contrary moves in total (irrespective of whether they are disjoint from $\mathcal{M}$ or not) -- as in both cases these would be more contrary moves than the token swapping instance can afford. %

Recall that for a slot token with initial  position in slot gadget $i$ or a non-slot token with target position in slot gadget $i$, its \defn{main path} is the path consisting of the slot gadget $i$ (not including the nook edge) and the ordering gadget, connected via the root.

\begin{definition}
Let $t$ and  $s$ be two slot tokens that share the same initial gadget or two non-slot tokens that share the same final gadget.
Observe that $t$ and $s$ have the same main path $p$. We say that $t$ is \defn{out of order} with $s$ if $t$ and $s$ are both on their main path $p$, but $s$ is on the opposite side of $t$ than it was initially. Note that this relation is symmetric, that is, if $t$ is out of order with $s$, then also $s$ is out of order with $t$. Moreover, we say that $t$ is \defn{out of order} with a set $S$ of tokens if $t$ is out of order with every element $s \in S$.
\end{definition}

In what follows, we first show that transforming certain token configurations into others uses too many contrary moves. To this end, we start by proving two auxiliary lemmas %
and by identifying Transformations 1--4. These will subsequently be used in the proof of the main Lemma~\ref{lem:OutOfOrder} of this section. 

Consider the four transformations in Figure~\ref{fig:transformations}. Each transformation shows a pair of configurations and assumes that a token swapping sequence transforms the first configuration into the second one. The type of token $t$ (slot or non-slot), as well as its position relative to a set $S$ of $\lambda$ tokens of the same type are marked in each configuration.

\begin{figure}[tbh]
	\centering
	\includegraphics[width=\textwidth]{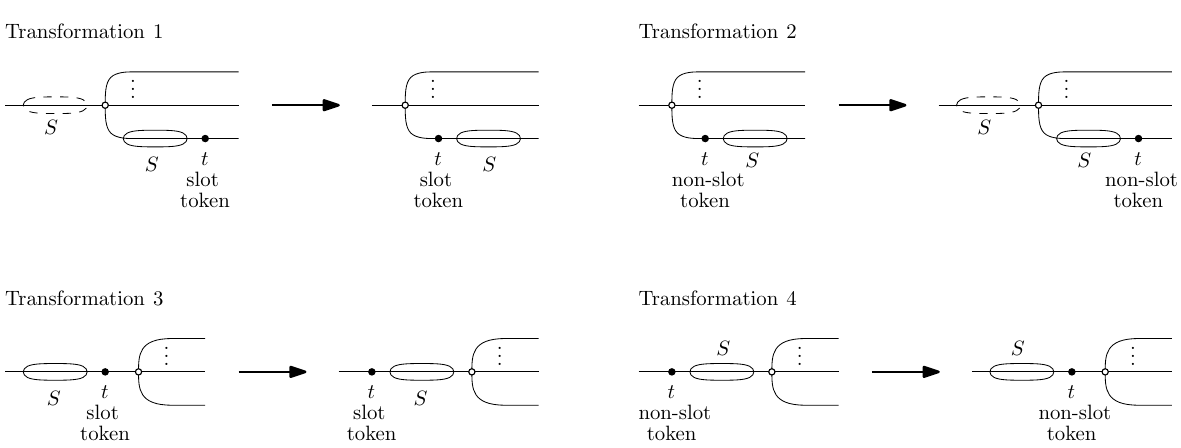}
	\caption{Four transformations. The empty circle denotes the root of the tree. In each configuration, the solid and dashed ovals mark the possible location of the set $S$ of $\lambda$  tokens. Token $t$ can be anywhere inside the marked gadget or at the root as long as it respects its position relative to the set~$S$. The marked gadget is assumed to be $t$'s initial or target gadget, depending on the context. The tokens from $S$ have the same initial and target gadget as $t$. %
	}
	\label{fig:transformations}
\end{figure}

\begin{lemma} \label{lem:transformations}
If $t$ exchanges its order with each of the $\lambda$ tokens of the set $S$ inside the same gadget as the one in which $t$ starts and ends a transformation, then realizing any of the four transformations from Figure~\ref{fig:transformations} takes either at least $ \lambda^2/(2\cdot3^2)$ contrary moves or at least $\lambda/3-n(m+1)$ contrary moves that are disjoint from contrary moves accounted in $\mathcal{M}$.  
\end{lemma}

\begin{proof}
All four transformations are argued by contradiction. Typically, the token $t$ will be able to exchange with tokens of $S$ in the marked gadget $i$ in up to 3 different ways, and so there will always be a subset $S'$ of at least $\lambda/3$ tokens from $S$ that swapped with $t$ in the same way.

{\bf  Transformation 1.} There are three possibilities: either the tokens of $S'$ swap directly with the token $t$ in the slot gadget $i$ (Case a), or each of the $S'$ tokens at some point enters the nook and the token $t$ passes it by moving in slot gadget $i$ %
(Case b), or the token $t$ enters the nook and the $S'$ tokens pass it by moving in the slot gadget %
(Case c).

In Case a, $t$ is a slot token in its initial slot gadget and each $s\in S'$ makes a contrary (right) move when swapping with $t$. By Corollary~\ref{obs:interface}(\ref{int2}), these are %
at least $\lambda/3$ contrary moves that were not accounted in $\mathcal{M}$.

In Case b, moving each slot token $s\in S'$ into the nook is a contrary (right) move. By Corollary~\ref{obs:interface}(\ref{int4}), these are at least $\lambda/3$ %
contrary moves that were not accounted in $\mathcal{M}$.
 
For Case c, note that all $\geq \lambda/3$ %
tokens of $S'$ are to the left of $t$ just before $t$ enters the nook for the first time, and all these tokens are to the right of $t$ after $t$ leaves the nook for the last time. Hence the block of at least $ \lambda/3$ %
slot tokens of $S'$ moved by at least $ \lambda/3$ %
vertices to the right, which takes at least $\frac{1}{2}\cdot (\frac{\lambda}{3})^2$ %
contrary (right) moves.

{\bf  Transformation 2.} The same three possibilities as for Transformation 1 need to be discussed. In Case a, when each non-slot token $s\in S'$ swaps directly with $t$, it makes a contrary (left) move in its destination gadget. These are at least $ \lambda/3$ %
contrary moves that by Corollary~\ref{obs:interface}(\ref{int1}) were not accounted in $\mathcal{M}$. In Case b, every time a non-slot token $s \in S'$ moves out of the nook, it makes a contrary (left) move. These are at least $ \lambda/3$ %
contrary moves that by Corollary~\ref{obs:interface}(\ref{int4}) were not accounted in $\mathcal{M}$. In Case c, the at least $\lambda/3$ %
non-slot tokens of $S'$ shift by at least $ \lambda/3$ %
vertices to the left, which takes at least $ \frac{1}{2}\cdot(\frac{\lambda}{3})^2$ %
contrary (left) moves.

{\bf Transformation 3.} The only possibility here is that each token $s \in S$ directly swaps with the slot token $t$ in the ordering gadget, meaning that the slot token $s$ makes a contrary (right) move. These are $\lambda$ %
contrary moves that by Corollary~\ref{obs:interface}(\ref{int2}) were not accounted in $\mathcal{M}$.

{\bf Transformation 4.} The non-slot token $t$ must again swap directly with each $s \in S$ in the ordering gadget but we distinguish two cases depending on whether the tokens of $S'$ have never left the ordering gadget before swapping with $t$ (Case a) or whether they did (Case b). Note that `never left the ordering gadget' refers to the entire token swapping sequence, possibly even before Transformation 4 started.
In Case a, each non-leaving non-slot token $s$ of $S'$ makes a contrary (left) move when swapping with $t$. 
These are at least $ \lambda/3$ contrary moves, out of which by Corollary~\ref{obs:interface}(\ref{int3}) at most $n(m+1)$ have been accounted in $\mathcal{M}$. Hence, there are at least $\lambda/3 - n(m+1)$ contrary moves disjoint from $\mathcal{M}$.

In Case b, each $s \in S'$ had visited the root before it swapped with $t$ and eventually all $\geq \lambda/3$ %
tokens of $S'$ lie to the left of $t$ in the ordering gadget.
Hence the tokens in $S'$ must have altogether made at least $1+2+\dots+\lambda/3 \geq \frac{1}{2}\cdot(\frac{\lambda}{3})^2$ %
contrary (left) moves (some of them possibly before the transformation started).
\end{proof}

The following lemma lists bounds that imply too many contrary moves in the token swapping instance.

\begin{lemma} \label{tooManyContraryMoves}$ $
\begin{enumerate}
    \item The total number of contrary moves in the token swapping solution is less than $\frac{\ell^2}{2^53^2(m+1)^2}$.\label{tooManyItem1}
    \item If $\lambda\geq\frac{\ell}{4(m+1)}$, then realizing any of the four transformations from Figure~\ref{fig:transformations} as described in Lemma~\ref{lem:transformations} requires too many contrary moves (i.e., either too many contrary moves in total, or too many contrary moves that are disjoint from contrary moves accounted in~$\mathcal{M}$).\label{tooManyItem2}
\end{enumerate}
\end{lemma}

\begin{proof}$ $
\begin{enumerate}
    \item Recall that $\ell = k^{1/2+2/c}$, we assume that $m,n>5$ and from the construction of the token swapping instance, $k=(mn)^c$. Hence, $\frac{\ell^2}{2^53^2(m+1)^2} = k^{1+1/c}\cdot \frac{m^3n^3}{2^53^2(m+1)^2}$ which exceeds $k^{1+1/c}$ contrary moves given in Observation~\ref{obs:contrary} whenever $m,n\geq 5$.
    
    \item By Lemma~\ref{lem:transformations}, the four transformations in Figure~\ref{fig:transformations} result in either at least $ \lambda^2/(2\cdot3^2)$ contrary moves or at least $\lambda/3-n(m+1)$ contrary moves that are disjoint from contrary moves accounted in $\mathcal{M}$. We show that either case exceeds the number of available contrary moves. %
    
    By assumption on $\lambda$, we have $\lambda^2/(2\cdot3^2) \geq \frac{\ell^2}{2^53^2(m+1)^2}$ which by the first statement of the current lemma are too many contrary moves.
    
    For the second case, we show that $\lambda/3-n(m+1) \geq \lambda/4 \geq k^{1/c}$. By Observation~\ref{obs:contrary}, this will imply that there are too many contrary moves disjoint from $\mathcal{M}$.
    To prove the first inequality, it suffices to show that $\lambda/12 \geq n(m+1)$. Recall that $\ell=k^{1/2+2/c}$, by construction of the token swapping instance $k=(mn)^c$ and we have set $c=25$. By assumption, $\lambda \geq \frac{\ell}{4(m+1)}$ which implies $\lambda \geq \frac{(mn)^{c/2}(mn)^2}{4(m+1)}$. %
    We obtain that $\lambda/12 \geq \frac{(mn)^{c/2}(mn)^2}{48(m+1)} \geq n(m+1)$ holds whenever $m,n \geq 2$. %
    Finally, $\lambda/4 \geq \frac{\ell}{16(m+1)}= k^{1/c} \cdot \frac{(mn)^{1+c/2}}{16(m+1)}$ %
    and this exceeds $k^{1/c}$ whenever $m,n\geq2$.
\end{enumerate}
\end{proof}

We can now prove the main lemma:

\begin{lemma} \label{lem:OutOfOrder}
At all times, for every non-item token $t$ that is currently in a slot gadget or at the root, there does not exist a set $S$ of $\ell=k^{1/2+2/c}$ tokens that are out of order with $t$.
\end{lemma}

\begin{proof}
Suppose for a contradiction that $t$ is out of order with a set $S$ of $\ell$ tokens of the same type (i.e., all slot or all non-slot tokens). 
Observe that with respect to the main path of $t$ at least $\ell / 2$ of these tokens are on the same side (either left or right) of $t$. Call this subset $S'$, $|S'| \geq \ell / 2$.  We will consider four cases: the token $t$ is either a slot or a non-slot token and the set $S'$ is either to the right or to the left of $t$, see Figure~\ref{fig:four_cases.pdf}. 

\begin{figure}[tbh]
	\centering
	\includegraphics[scale=0.8]{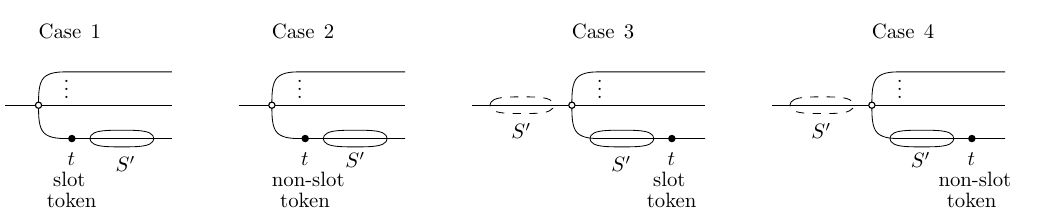}
	\caption{Four cases considered in proof of Lemma~\ref{lem:OutOfOrder}. The empty circle denotes the root of the tree.
                   The bottom-most slot gadget indicates the initial or the target slot gadget (depending on the context) of token $t$ and of all tokens from $S$.}	
	\label{fig:four_cases.pdf}
\end{figure} 

{\bf Case 1.} If $t$ is a slot token that is to the left of the set $S'$ of slot tokens in the current out-of-order configuration, then the initial configuration of the token swapping instance must have had $t$ in the same slot gadget on the right side of $S'$. Transforming the initial configuration into the current out-of-order configuration is exactly Transformation 1. If $\ell/4$ or more %
tokens from $S'$ exchange with $t$ inside the slot gadget, then by Lemma~\ref{tooManyContraryMoves}(\ref{tooManyItem2}) %
there are too many contrary moves. Otherwise there is a subset $S''$ of at least $|S'|/2\geq\ell/4$ tokens from $S'$ that exchange with $t$ outside of the slot gadget. Then, in order for $t$ to get out of this slot gadget, all $S''$ tokens must first get out of the gadget and then re-enter it again. Hence, Lemma~\ref{lem:swapNotOutsideOfGadget} applies and there are at least $\frac{1}{2}\cdot(\frac{\ell}{4})^2=\ell^2/2^5$ %
contrary moves, which is by Lemma~\ref{tooManyContraryMoves}(\ref{tooManyItem1}) too many.

{\bf Case 2.} If $t$ is a non-slot token that is to the left of the set $S'$ of non-slot tokens in the current out-of-order configuration, then the final configuration of the token swapping instance must have  $t$ in the same slot gadget on the right side of $S$. Transforming the current out-of-order configuration into the final configuration is exactly Transformation 2. By a similar argument as in Case 1, the Lemmas~\ref{tooManyContraryMoves}(\ref{tooManyItem2})%
, \ref{lem:swapNotOutsideOfGadget} and \ref{tooManyContraryMoves}(\ref{tooManyItem1}) show that $t$ cannot exchange order with the tokens from $S'$ inside nor outside of the slot gadget.

{\bf Case 3.} Assume that $t$ is a slot token and the set $S'$ of slot tokens is to the left of it, each such token either in the same slot gadget $i$ or in the ordering gadget. Consider the first time after the current out-of-order configuration that token $t$ reaches the root: this intermediate configuration has at least %
$\frac{\ell}{2(m+1)}$ tokens of $S'$, call this set $S''$, either in slot gadget $i$ (Subcase~a), or in a slot gadget $j\neq i$ (Subcase~b), or in the ordering gadget (Subcase~c). 

In Subcase a, %
all $\geq\frac{\ell}{2(m+1)}$ %
tokens of $S''$ must exchange their order with token $t$ inside the slot gadget $i$ between the current out-of-order configuration and the intermediate configuration. This is exactly Transformation 1 and by Lemma~\ref{tooManyContraryMoves}(\ref{tooManyItem2}) it causes too many contrary moves.

In Subcase b, moving the $S''$ slot tokens from their current position into slot gadget $j$ (for the intermediate configuration) takes at least $\frac{1}{2}\cdot(\frac{\ell}{2(m+1)})^2$ %
contrary (right) moves, which are again by Lemma~\ref{tooManyContraryMoves}(\ref{tooManyItem1}) more contrary moves than we can afford.

Finally, note that given the current out-of-order configuration, the target configuration of the token swapping instance must have $t$ to the left of $S'$ tokens in the ordering gadget. Thus, in Subcase c, transforming the intermediate configuration into the target configuration is exactly Transformation 3. If at least half of the $S''$ tokens, i.e.,~ $\frac{\ell}{4(m+1)}$ tokens, exchange with $t$ inside the ordering gadget, then by Lemma~\ref{tooManyContraryMoves}(\ref{tooManyItem2}) there are too many contrary moves. Otherwise at least $|S''|/2\geq\frac{\ell}{4(m+1)}$ tokens of $S''$ exchange with $t$ outside of the ordering gadget. For that, they must all first leave and then re-enter the ordering gadget. By Lemma~\ref{lem:swapNotOutsideOfGadget}, this results in at least $\frac{1}{2} \cdot (\frac{\ell}{4(m+1)})^2$ contrary moves which by Lemma~\ref{tooManyContraryMoves}(\ref{tooManyItem1}) are too many.

{\bf Case 4.} Assume that $t$ is a non-slot token and the set $S'$ of non-slot tokens is to the left of it, each such token either in the same slot gadget $i$ or in the ordering gadget. Consider the last time before the current out-of-order configuration that token $t$ was at the root: this intermediate configuration has at least $\frac{\ell}{2(m+1)}$ %
tokens of $S'$, call this set $S''$, either in slot gadget~$i$ (Subcase~a), or in a slot gadget $j\neq i$ (Subcase~b), or in the ordering gadget (Subcase~c). 

In Subcase a, %
all $\geq\frac{\ell}{2(m+1)}$ %
tokens of $S''$ must exchange their order with token $t$ inside the slot gadget $i$ between the intermediate and the current out-of-order configuration. This is exactly Transformation 2 and by Lemma~\ref{tooManyContraryMoves}(\ref{tooManyItem2}) it causes too many contrary moves.

In Subcase b, moving the $S''$ non-slot tokens out of slot gadget $j$ (in the intermediate configuration) to their current position takes at least $\frac{1}{2}\cdot(\frac{\ell}{2(m+1)})^2$ %
contrary (left) moves, which are again by Lemma~\ref{tooManyContraryMoves}(\ref{tooManyItem1}) more contrary moves than we can afford.

Finally, note that given the current out-of-order configuration, the initial configuration of the token swapping instance must have $t$ to the left of $S'$ tokens in the ordering gadget. Thus, in Subcase c, transforming the initial configuration into the intermediate configuration is exactly Transformation 4. If at least half of the $S''$ tokens, i.e.,~ $\frac{\ell}{4(m+1)}$ tokens, exchange with $t$ inside the ordering gadget, then by Lemma~\ref{tooManyContraryMoves}(\ref{tooManyItem2}) there are too many contrary moves. Otherwise at least $|S''|/2\geq\frac{\ell}{4(m+1)}$ tokens of $S''$ exchange with $t$ outside of the ordering gadget. By Lemma~\ref{lem:swapNotOutsideOfGadget}, this results in at least $\frac{1}{2} \cdot (\frac{\ell}{4(m+1)})^2$ contrary moves which by Lemma~\ref{tooManyContraryMoves}(\ref{tooManyItem1}) are too many.

\end{proof}

\paragraph{Every slot gadget contains an item token at all times}\label{sec:Item Tokens Stuck}$ $\\

The goal of this section is to prove the following lemma.

\begin{lemma}\label{lem:LeafLeaving} At all times there is at least one item token in each slot gadget at distance at least $k(1-\frac{1}{4n})$ from the root.
\end{lemma}

\begin{proof} 
Our approach will be to suppose for contradiction that there is a slot gadget that does not satisfy the lemma statement, which in particular means that the token at the nook of that slot gadget is a non-item token. We will argue that having a non-item token in the nook with no item token nearby creates too many contrary moves.

To  be precise, suppose that there is a time $\tau$
when some slot gadget, say slot gadget $i$, does  not  have an item token at distance at least $k(1-\frac{1}{4n})$ from the root.

Let $\tau_1$ be the last time before $\tau$
in which we had an item token occupying the nook,
and let $q$ be that item token. Similarly, we define $\tau_2$ as the first time after $\tau$
when
there is an item token occupying the nook, and let $q'$ be the item token. Note that both $\tau_1$ and $\tau_2$ must exist since the initial and target configurations are candidates for each of them, respectively. %

Let $v$ be the unique vertex of slot gadget $i$ whose distance to the root is $k(1-\frac{1}{4n})$. Equivalently, $v$ is the vertex $\frac{k}{4n}$ vertices to the left of the nook parent of slot gadget $i$. 
At time $\tau_1$, $q$ was at the nook of slot gadget $i$, and before time $\tau$, $q$ reached vertex $v$.
That is, between times $\tau_1$ and $\tau$, token $q$ traverses all of the $\frac{k}{4n}$ edges from the nook  parent to $v$. Let $R$ be the set of tokens that $q$ swaps with on $q$'s last traversal  of each of these edges before time $\tau$.  By  Claim~\ref{lem:unique},  these tokens are distinct, so $|R|=\frac{k}{4n}$.   

We now give an overview of the proof. We first show that after their swaps with $q$, many of the tokens in $R$ proceed to pass by the nook parent in the rightward direction (Claim~\ref{claim:rightnook}). Then, we identify a non-item token $t$ that is in the nook while many of the tokens in $R$ pass by the nook parent (Claim~\ref{claim:tnook}). Then, we show that while $t$ is in the nook, many tokens also pass by the nook parent in the leftward direction, to ``replace'' the tokens in $R$ that are passing in the rightward direction (Claim~\ref{claim:Tbig}). 

Next, we show that $t$ cannot have slot gadget $i$ as its initial nor final gadget (Claim~\ref{claim:fromelsewhere}). This is because if $t$ has slot gadget $i$ as its initial or final gadget then $t$ will get out of order with the tokens that are passing by the nook in the leftward or rightward direction, respectively. This contradicts Lemma~\ref{lem:OutOfOrder} (which says that a token cannot be out of order with too many tokens). 

The remaining case is when $t$ does not have $i$ as either its initial or final gadget. In this case we find a similar type of contradiction: we find a large number of tokens that have the same initial and final gadget as $t$ that move out of their initial gadget and into their final gadget while $t$ stays in slot gadget $i$ (Claim~\ref{claim:slotfromt}). We are able to find these tokens because the segments (actually just their cores) reach the root in order according to the scaffold solution by Lemma~\ref{lem:segbyseg}, and the scaffold solution includes segments from \emph{every} slot gadget at high frequency due to the padding segments. The consequence of this is that $t$ becomes out of order with the tokens from its initial gadget that have entered their final gadget while $t$ is in slot gadget $i$, which again contradicts Lemma~\ref{lem:OutOfOrder}.

We now formalize this intuition. Towards showing that many tokens in $R$ pass by the nook parent, we first show that many tokens in $R$ are non-slot tokens whose final gadget is $i$.

\begin{claim}\label{claim:sizer2}
$R$ contains at least $\frac{k}{4n}-k^{1/c}-k^{1/2+2/c}$ many non-slot tokens whose final gadget is $i$.
\end{claim}
\begin{proof}
When token $r\in R$ swaps with $q$, $r$ moves to the right. Thus, if $r$ is a slot token, then $r$ does a contrary move. By Corollary~\ref{obs:interface}(\ref{int5}), this contrary move is not in $\mathcal{M}$. Thus, at most $k^{1/c}$ tokens in $R$ are slot tokens. Also, by Lemma~\ref{lem:indest}, at most $k^{1/2+2/c}$ of the tokens in $R$ are non-slot tokens whose final gadget is not~$i$. 

Since $|R|=\frac{k}{4n}$ we have that the number of tokens in $R$ that are non-slot tokens whose final gadget is $i$ is at least $\frac{k}{4n}-k^{1/c}-k^{1/2+2/c}$ as claimed. %
\end{proof}

Let $R_2$ be the subset of tokens in $R$ that are non-slot tokens whose final gadget is~$i$. %
Now consider the movement of the tokens in $R_2$ between times $\tau_1$ and $\tau_2$. We claim that many tokens in $R_2$ move to the right of the nook parent.

\begin{claim}\label{claim:rightnook}
At least $k^{1-2/c}$ tokens in $R_2$ move to the right of the nook parent at some point between times $\tau_1$ and $\tau_2$.
\end{claim}
\begin{proof}
First, we claim that fewer than $k^{1/2+1/c}$ tokens in $R_2$ can ever reach vertex $v$ (the vertex at distance $\frac{k}{4n}$ to the left of the nook parent) after swapping with $q$. This is because each token in $R_2$ swaps with $q$ along a unique edge of distance at most $\frac{k}{4n}$ from the nook. Thus, moving $k^{1/2+1/c}$ of these tokens to a distance \emph{at least} $\frac{k}{4n}$ from the nook requires at least $\frac{(k^{1/2+1/c})^2}{2}>k^{1+1/c}$ left moves, which are all contrary moves since every token in $R_2$ is a non-slot token.  

Other than the $k^{1/2+1/c}$ tokens mentioned above, any token that does not satisfy the claim must remain between $v$ and the nook parent from when it swaps with $q$ until time $\tau_2$. We will show that there are at most $k^{1/c}$ many such tokens. Recall that at time $\tau$ no item token is in slot gadget $i$ at distance at least $k(1-\frac{1}{4n})$ from the root. In particular, at time $\tau$, $q'$ (the item token defined by time $\tau_2$) is either in slot gadget $i$ between the root and $v$ or not in slot gadget $i$. 

At time $\tau_2$, $q'$ is at the nook. Thus, between times $\tau$ and $\tau_2$, $q'$ swaps with every token in $R'$ with $q'$ moving to the right and the tokens in $R'$ moving to the left. Since every token in $R_2$ (and thus $R'$) is a non-slot token, these left moves are contrary moves. Furthermore, these contrary moves are not in $\mathcal{M}$ by Corollary~\ref{obs:interface}(\ref{int5}). Thus, $|R'|< k^{1/c}$.

Putting everything together, the number of tokens that move to the right of the nook parent at some point between times $\tau_1$ and $\tau_2$ is at least $|R|-k^{1/2+1/c}-k^{1/c}\geq\frac{k}{4n}-2k^{1/c}-k^{1/2+2/c}-k^{1/2+1/c}>k^{1-2/c}$ since $c>6$.
\end{proof}

Now, we are ready to define the token $t$ mentioned in the overview. Let $e$ be the edge between the nook parent and the vertex to its right.

\begin{claim}\label{claim:tnook}
There exists an interval $\mathcal{I}$ of time contained between $\tau_1$ and $\tau_2$ such that: $(1)$ a single non-item token $t$ is in the nook of slot gadget $i$ for the entirety of $\mathcal{I}$, and $(2)$ at least $k^{1-3/c}$ tokens of $R_2$ traverse edge $e$ in the rightward direction during $\mathcal{I}$.
\end{claim}
\begin{proof}
First observe that, by Corollary~\ref{obs:interface}(\ref{int4}), each move of a non-item token into and out of the nook creates a contrary move that is not in $\mathcal{M}$. Thus, the total number of different non-item tokens that ever go into the nook of any gadget (and in particular slot gadget $i$) is at most $k^{1/c}$. 

Let $R_3 \subseteq R_2$ be the tokens of $R_2$ that move to the right of the nook parent at some point between times $\tau_1$ and $\tau_2$. By Claim~\ref{claim:rightnook}, we know that $|R_3|\geq k^{1-2/c}$. 

For each token of $R_3$, consider the last time that it traverses edge $e$ to the right between times $\tau_1$ and $\tau_2$ and look at which token is in the nook during that traversal. By the pigeonhole principle and the above bounds, there must exist a non-item token $t$ that is in the nook of gadget~$i$ contiguously for an interval $\mathcal{I}$ of time during which at least $|R_3|/k^{1/c}\geq k^{1-3/c}$ tokens of $R_3$ traverse edge $e$ in the rightward direction. 
\end{proof}

 Let $\mathcal{I}$ be the interval of time from Claim~\ref{claim:tnook} and let $R_4\subseteq R_2$ be the set of at least $k^{1-3/c}$ tokens from Claim~\ref{claim:tnook}.
Also, let $\mathcal{S}$ be the set of swaps during interval $\mathcal{I}$ during which tokens in $R_4$ traverse edge $e$ in the rightward direction, taking the last such traversal for each token in $R_4$. %
Let $L$ be the set of tokens that $R_4$ swaps with during the swaps $\mathcal{S}$.

Next, we will show that many distinct tokens traverse the edge $e$ to the left during $\mathcal{I}$. We already know that there are many such traversals (performed by tokens in $L$), however we need to guarantee that many such traversals are performed by \emph{distinct} tokens. In particular, it is possible for $L$ to be composed of a single token that is going back and forth across~$e$. Furthermore, we require that these distinct tokens were initially in slot gadget $i$. These required properties are captured in the following claim.

\begin{claim}\label{claim:Tbig}
At least $k^{1-4/c}$ distinct tokens whose initial slot gadget is $i$ traverse $e$ in the left direction during interval $\mathcal{I}$.
\end{claim}
\begin{proof}
First suppose $|L|\geq |R_4|/4$. We already know that every token in $L$ traverses $e$ in the left direction while $t$ is in the nook, so it suffices to show that at least half of the tokens in $L$ were initially in slot gadget $i$ (since $|L|/2\geq |R_4|/8\geq k^{1-3/c}/8>k^{1-4/c}$). Suppose more than half of the tokens in $L$ are \emph{not} initially in slot gadget $i$. If at least $|L|/4$ of these tokens have slot gadget $i$ as their final gadget, then their leftward traversal of $e$ is a contrary move, which is not in $\mathcal{M}$ by Corollary~\ref{obs:interface}(\ref{int1}). Thus, at least $|L|/4$ of these tokens do not have gadget~$i$ as their initial or final gadget. In this case, these tokens must move both into slot gadget $i$ all the way to $e$, and out of slot gadget $i$. One of these two directions is composed of contrary moves, so since $e$ is at distance $k$ from the root, we have $k$ contrary moves for each of the $|L|/4$ tokens, which is a contradiction.

On the other hand, suppose $|L|< |R_4|/4$. This is the tricky case because for example $L$ could be just a single token that is moving back and forth across $e$, but we still need to show that many \emph{distinct} tokens traverse $e$ in the left direction. Because $|L|< |R_4|/4$, some tokens in $L$ participate in multiple swaps in $\mathcal{S}$. 
Let $\mathcal{S}'\subseteq\mathcal{S}$ be the set of swaps such that the participating token in $L$ has already previously performed a swap in $\mathcal{S}$ and will again perform another swap in $\mathcal{S}$. That is, for every token in $L$, at most two of its swaps in $\mathcal{S}$ are \emph{not} in $\mathcal{S'}$. Thus, $|\mathcal{S'}|\geq|\mathcal{S}|-2|L|=|R_4|-2|L|>|R_4|/2$. Defining $\mathcal{S'}$ is useful because we are guaranteed the following property: during each swap in $\mathcal{S'}$, the participating token in $L$ is moving leftward, however at some later point this same token must traverse $e$ in the \emph{rightward} direction in order to be set up to again traverse $e$ leftward during its next swap in $\mathcal{S}$. For every swap in $\mathcal{S'}$, consider the next time this token in $L$ traverses $e$ in the rightward direction and denote this set of swaps by $\mathcal{S''}$. By definition, $|\mathcal{S''}|=|\mathcal{S'}|$. Now we will condition on whether the tokens in $L$ that participate in swaps in $\mathcal{S'}$ are slot tokens or non-slot tokens. 

{\bf Case 1.} Suppose that for at least half of the swaps in $\mathcal{S'}$, the participating token in $L$ is a slot token. Then, for at least half of the swaps in $\mathcal{S''}$, the participating token in $L$ is performing a contrary move. At least $|\mathcal{S''}|/2-k^{1/c}$ of these contrary moves are in $\mathcal{M}$, because if $k^{1/c}$ of them are not in $\mathcal{M}$ then we have a contradiction. Then according to the characterization of $\mathcal{M}$ in Observation~\ref{obs:purple} (specifically item~\ref{purp2}), for each these contrary moves in $\mathcal{M}$, the token $s$ swapped with is a \emph{distinct} slot token that was initially in slot gadget $i$. The number of such tokens $s$ is at least $|\mathcal{S''}|/2-k^{1/c}=|\mathcal{S'}|/2-k^{1/c}>|R_4|/4-k^{1/c}\geq k^{1-3/c}/4-k^{1/c}>k^{1-4/c}$ since $c>4$. This completes Case 1.

{\bf Case 2.} Suppose that for at least half of the swaps in $\mathcal{S'}$, the participating token in $L$ is a non-slot token. In this case, for at least half of the swaps in $\mathcal{S''}$,  the participating token in $L$ is performing a contrary move. At least $|\mathcal{S'}|/2-k^{1/c}$ of these contrary moves are in $\mathcal{M}$, because if $k^{1/c}$ of them are not in $\mathcal{M}$ then we have a contradiction. Then according to the characterization of $\mathcal{M}$ in Observation~\ref{obs:purple} (specifically item~\ref{purp3}), for each of these contrary moves of tokens in $L$ in $\mathcal{M}$, the previous time this token in $L$ moved rightward along $e$ it swapped with a \emph{distinct} slot token $s$ that was initially in slot gadget $i$. Importantly, the previous time this token in $L$ moved rightward along $e$ was during $\mathcal{I}$ by the definition of $\mathcal{S'}$. The number of such tokens $s$ is at least $|\mathcal{S'}|/2-k^{1/c}>|R_4|/4-k^{1/c}\geq k^{1-3/c}/4-k^{1/c}>k^{1-4/c}$ since $c>4$. This completes the proof.
\end{proof}

Now we will consider where token $t$ originated. First, we will show that slot gadget $i$ is \emph{not} the initial or final gadget for token $t$. 

\begin{claim}\label{claim:fromelsewhere}
Slot gadget $i$ is \emph{not} the initial or final gadget for token $t$.
\end{claim}
\begin{proof}
A key fact in this proof is that we know that while $t$ is in the nook, many tokens pass by the nook both rightward (in particular $R_4$) and leftward (in particular the tokens from Claim~\ref{claim:Tbig}). The idea of this proof is to show that if $t$ has slot gadget $i$ as its initial or final gadget then these tokens passing by the nook become out of order with $t$, which is impossible by Lemma~\ref{lem:OutOfOrder}.

First, suppose for contradiction that slot gadget $i$ is the final gadget for token $t$. Let $R_5$ be the subset of $R_4$ that are on their main path both immediately before and immediately after interval $\mathcal{I}$. By Lemma~\ref{lem:indest}, $|R_5|\geq |R_4|-2k^{1/2+2/c}$. Since every token in $R_2\supseteq R_5$ also has final gadget~$i$, the notion of \defn{out of order} from  Lemma~\ref{lem:OutOfOrder} can apply to any pair of tokens in $R_5\cup \{t\}$. By Lemma~\ref{lem:OutOfOrder}, when $t$ is at the nook parent both immediately before and immediately after interval $\mathcal{I}$, $t$ cannot be out of order with any subset of $R_5$ of size at least $k^{1/2+2/c}$. Thus, compared to immediately before interval $\mathcal{I}$, right after interval $\mathcal{I}$ $t$ is on the opposite side of less than $2k^{1/2+2/c}$ tokens in $R_5$. Thus, at least $|R_5|-2k^{1/2+2/c}$ of the tokens in $R_5$ are on the same side (left or right) of $t$ immediately before and immediately after interval $\mathcal{I}$. Note that by definition all tokens in $R_5$ traverse the edge $e$ during interval $\mathcal{I}$. Thus, either moving these $|R_5|-2k^{1/2+2/c}$ tokens from their positions immediately before $\mathcal{I}$ to $e$, or moving them from $e$ to their positions immediately after $\mathcal{I}$, are contrary moves. This amounts to a total number of contrary moves of at least $(|R_5|-2k^{1/2+2/c})^2\geq (|R_4|-4k^{1/2+2/c})^2\geq (k^{1-3/c}-4k^{1/2+2/c})^2>k^{2-7/c}>k^{1+1/c}$, since $c>10$. 

Now suppose for contradiction that token $t$ was initially in slot gadget $i$. This case is similar to the previous case, except in this case $R_5$ and $t$ do not have the same initial gadget so the notion of ``out of order'' does not apply to them; instead, $t$ has the same initial gadget as the tokens from Claim~\ref{claim:Tbig} so the notion of ``out of order'' applies to them instead.
Claim~\ref{claim:Tbig} implies that there is a set $L'$ of at least $k^{1-4/c}$ tokens initially in slot gadget $i$ that traverse $e$ in the left direction during interval $\mathcal{I}$. Let $L''$ be the subset of $L'$ that are on their main path both immediately before and immediately after interval $\mathcal{I}$. By Lemma~\ref{lem:indest}, $|L''|\geq |L'|-2k^{1/2+2/c}$. By Lemma~\ref{lem:OutOfOrder}, when $t$ is at the nook parent both immediately before and immediately after interval $\mathcal{I}$, $t$ cannot be out of order with any subset of $L''$ of size at least $k^{1/2+2/c}$. Thus, compared to immediately before interval $\mathcal{I}$, right after interval $\mathcal{I}$ $t$ is on the opposite side of at most $2k^{1/2+2/c}$ tokens in $L''$. Thus, at least $|L''|-2k^{1/2+2/c}$ of the tokens in $L''$ are on the same side (left or right) of $t$ immediately before and immediately after interval $\mathcal{I}$. Thus, either moving these $|L''|-2k^{1/2+2/c}$ tokens from their positions immediately before $\mathcal{I}$ to $e$, or moving them from $e$ to their positions immediately after $\mathcal{I}$, are contrary moves. This amounts to a total number of contrary moves at least $(|L''|-2k^{1/2+2/c})^2\geq (|L'|-4k^{1/2+2/c})^2\geq (k^{1-4/c}-4k^{1/2+1/c})^2>k^{2-9/c}>k^{1+1/c}$, since $c>10$.
\end{proof}

Now that we know that slot gadget $i$ is not the initial or final gadget for $t$, we analyze what happens in $t$'s initial and final gadgets while $t$ is in slot gadget $i$. We will eventually argue when $t$ reaches the root after being in slot gadget $i$, $t$ has become out of order with a large set of tokens, which we have shown in Lemma~\ref{lem:OutOfOrder} is impossible. The first step towards proving this is to show that while $t$ is in slot gadget $i$, many tokens that were initially in slot gadget~$i$ reach the root for the first time.

We need to precisely define the time interval ``while $t$ is in slot slot gadget $i$'':
Let $\mathcal{I}'$ be the time interval starting from the last time $t$ enters slot gadget $i$ before interval $\mathcal{I}$ and ending at the first time $t$ exits slot gadget $i$ after interval $\mathcal{I}$. Both events exist since $i$ is not the initial or final gadget for $t$ by Claim~\ref{claim:fromelsewhere}.

\begin{claim}\label{claim:up}
During interval $\mathcal{I}'$, at least $k+k^{1-6/c}$ tokens that were initially in slot gadget $i$ reach the root for the first time.
\end{claim}

\begin{proof}
At the beginning and end of interval $\mathcal{I}'$, $t$ is at the root, and during interval $\mathcal{I}$, $t$ is in the nook of slot gadget $i$. Thus, during interval $\mathcal{I}'$, $t$ moves from the root to the nook of slot gadget $i$ and then back to the root. Let $\mathcal{S}$ be the set of swaps that $t$ performs while traveling from the root to the nook parent of slot gadget $i$, taking the last such swap that $t$ performs for each edge between the root and the nook parent; thus, $|\mathcal{S}|=k$. Let $S_1$ be the set of tokens that swap with $t$ during swaps in $\mathcal{S}$. By Claim~\ref{lem:unique}, $|S_1|=k$.

By Claim~\ref{claim:Tbig}, at least $k^{1-4/c}$ tokens that were initially in slot gadget $i$ traverse $e$ in the left direction during interval $\mathcal{I}$. Call these tokens $S_2$. 

We will show three properties of $|S_1\cup S_2|$, which together complete the proof. Roughly, the three properties are the following:

(1) many of the tokens in $|S_1\cup S_2|$ were initially in slot gadget $i$, 

(2) $|S_1\cup S_2|$ is large, and 

(3) many of the tokens in $|S_1\cup S_2|$ reach the root for the first time during interval $\mathcal{I}'$.

First we will show item (1). All of the tokens in $S_2$ were initially in slot gadget $i$ by definition, so our goal is to get a lower bound on the number of tokens in $S_1$ that were initially in slot gadget $i$. Less than $k^{1/2+1/c}$ of the tokens in $S_1$ do not have slot gadget $i$ as their initial or final gadget, because otherwise it requires $k^{2(1/2+1/c)}/2$ contrary moves to get these tokens into slot gadget $i$ to the position of the swap in $\mathcal{S}$, or to get them back out of slot gadget $i$ (since each token in $S_1$ swaps with $t$ along a distinct edge). Further, note that if a token $s\in S_1$ has its final position in slot gadget $i$, then during its swap with $t$, $s$ is performing a contrary (left) move, which is not in $\mathcal{M}$ by Corollary~\ref{obs:interface}(\ref{int1}). Thus, less than $k^{1/c}$ of the tokens in $S_1$ have final position in gadget~$i$. Therefore, we have shown that at least $k-k^{1/2+1/c}-k^{1/c}$ tokens in $S_1$ were initially in slot gadget $i$. Call these tokens $S'_1$. From now on we will focus on $|S'_1\cup S_2|$.

Next we will show item (2). Our goal is to get a lower bound on $|S'_1\cup S_2|$. To do this we will get an upper bound on $|S'_1\cap S_2|$. By definition, each token in $S'_1\cap S_2$ swaps with $t$ along a distinct edge between the root and the nook parent of slot gadget $i$, and then subsequently traverses the edge $e$. Each token in $S'_1$ (and thus $S'_1\cap S_2$) is initially in slot gadget $i$, so moving rightward from the location where it swapped with $t$ to edge $e$ are contrary moves. This incurs at least $(|S'_1\cap S_2|)^2/2$ contrary moves. Thus, $|S'_1\cap S_2|<k^{1/2+1/c}$ since otherwise the number of contrary moves would be more than $k^{1+1/c}$. Therefore, $|S'_1\cup S_2|=|S'_1|+|S_2|-|S'_1\cup S_2|>(k-k^{1/2+1/c}-k^{1/c})+k^{1-4/c}-k^{1/2+1/c}>k+k^{1-5/c}$ since $c>10$.

Lastly, we will show item (3). We claim that after the tokens in $|S'_1\cup S_2|$ have traversed $e$, at no point in the future are more than $k^{1/2+1/c}$ of these tokens to the right of the nook parent. This is because right moves are contrary moves for these tokens (since they were initially gadget is $i$) so having $k^{1/2+1/c}$ of these tokens to the right of the nook parent simultaneously would incur $k^{2(1/2+1/c)}/2 > k^{1+1/c}$ contrary moves. In particular, when $t$ reaches the nook parent for the last time during interval $\mathcal{I}'$, no more than $k^{1/2+1/c}$ tokens in $S'_1\cup S_2$ are to the right of $t$. For each token $s\in S'_1\cup S_2$ that is \emph{not} to the right of $t$ at this point in time, note that $s$ either reaches the root before $t$ or swaps with $t$ (since at the end of $\mathcal{I}'$, $t$ is at the root). Less than $k^{1/c}$ such tokens $s$ swap with $t$ since in this case, $s$ performs a contrary (right) move that is not in $\mathcal{M}$ by Corollary~\ref{obs:interface}(\ref{int6}). Thus, at least $|S'_1\cup S_2|-k^{1/2+1/c}-k^{1/c}$ tokens in $S'_1\cup S_2$ reach the root during interval $\mathcal{I}'$.

It remains to show that enough tokens in $S'_1\cup S_2$ reach the root \emph{for the first time} during interval $\mathcal{I}'$. We know that at most $k^{1/c}$ of the tokens in $S'_1\cup S_2$ reached the root for the first time before the beginning of interval $\mathcal{I}'$ because all of these tokens subsequently traverse $e$, which requires $k$ contrary moves per token. Thus, we have shown that at least $k+k^{1-5/c}-k^{1/2+1/c}-2k^{1/c}$ tokens in $S'_1\cup S_2$ reach the root for the first time during interval $\mathcal{I}'$. This quantity of tokens is more than $k+k^{1-6/c}$ since $c>12$.
\end{proof}

Recall that we will eventually argue when $t$ reaches the root at the end of interval $\mathcal{I}'$, $t$ has become out of order with a large set of tokens. To define this large set of tokens, it will be useful to return to the idea of the \defn{core} $C(S)$ of a segment $S$ from Definition~\ref{def:core}.  Next, we will show that during interval $\mathcal{I}'$, for \emph{each} slot gadget $j\not=i$, there is a segment whose initial position was in that slot, such that every token in the core of that segment reaches the root for the first time. 

\begin{claim}\label{claim:oneofeach}
For each slot gadget $j\not=i$, there is a segment $S_j$ whose initial position was in slot gadget $j$ such that during interval $\mathcal{I}'$ every token in $C(S_j)$ reaches the root for the first time.
\end{claim}
\begin{proof}
It will help to count the number of tokens of various types initially in slot gadget $i$. By Lemma~\ref{lem:noncore} the total number of tokens with initial position in slot gadget $i$ that are not in the core of any segment is at most $k^{1/2+5/c}$. Also, the total number of tokens with initial position in slot gadget $i$ in a padding segment is $n(\frac{k}{n^8})>k^{1-7/c}$ since there are $n$ padding segments in each slot gadget and each one is of length $\frac{k}{n^8}$.  Each remaining token initially in slot gadget $i$ is in the core of a big segment. Each big segment is of size $k$. 

By Claim~\ref{claim:up}, there is a set $S$ of at least $k+k^{1-6/c}$ tokens that were initially in slot gadget $i$ that reach the root for the first time during interval $\mathcal{I}'$. Because $|S|>k^{1/2+5/c}+k^{1-7/c}+k$ (which is true since $c>24$), the calculations from the previous paragraph imply that $S$ must contain at least one token from the core of a big segment and at least one token from the core of another segment.

Lemma~\ref{lem:segbyseg} says that the tokens in the cores of segments reach the root for the first time in the same order as in the scaffold solution. Property (P6) of the scaffold solution says that during a time interval in which at least one token from a big segment and at least one token from any other segment from the same slot gadget reach the root, at least one segment from each of the other slot gadgets reaches the root. Combining these two facts, we have that in our token swapping solution, during a time interval in which at least one token from the core of a big segment initially in slot gadget $i$ and at least one token from the core of any other segment initially in slot gadget $i$ reach the root for the first time, the core of at least one segment initially in each of the other slot gadgets reaches the root for the first time. $\mathcal{I}'$ is such  a time interval, since every token in $S$ reaches the root for the first time during $\mathcal{I}'$ and we have shown that $S$ contains at least one token from the core of a big segment and at least one token from the core of another segment (both initially in slot gadget $i$). This completes the proof.
\end{proof}

The final claim towards arguing that when $t$ reaches the root at the end of interval $\mathcal{I}'$, $t$ has become out of order with a large set of tokens, is the following.

\begin{claim}\label{claim:slotfromt}
During interval $\mathcal{I}'$, at least $k^{1-10/c}$ tokens that all have the same initial and final gadget as $t$ exit their initial gadget for the first time and enter their final gadget for the last time.

\end{claim}
\begin{proof}
First, suppose $t$'s initial gadget is a slot gadget. Then, by Claim~\ref{claim:oneofeach}, every token in the core of a segment $S$ with the same initial gadget as $t$ reaches the root for the first time during $\mathcal{I}'$. Then, by Lemma~\ref{lem:indest}, at least $|C(S)|-k^{1/2+2/c}$ tokens from $C(S)$ enter the ordering gadget for the last time during $\mathcal{I}'$. This completes the case where $t$'s initial gadget is a slot gadget since  by Lemma~\ref{lem:noncore}, $|C(S)|-k^{1/2+2/c}>k^{1-9/c}-k^{1/2+2/c}>k^{1-10/c}$ since $c>22$.

Now, suppose $t$'s initial gadget is the ordering gadget. Then, by Claim~\ref{claim:oneofeach}, every token in the core of a segment $S$ whose initial gadget is $t$'s \emph{final} gadget reaches the root for the first time during $\mathcal{I}'$. By Lemma~\ref{lem:exchange}, during interval $\mathcal{I}'$, a set $S'$ of at least $|C(S)|-k^{1/2+3/c}$ tokens with the same final gadget as $t$ enter this gadget for the last time. That is, at the beginning of interval $\mathcal{I}'$, no token in $S'$ had yet entered their final gadget for the last time. Then, by the contrapositive of Lemma~\ref{lem:indest}, at the beginning of interval $\mathcal{I}'$, at least $|C(S)|-k^{1/2+3/c}-k^{1/2+2/c}$ of the tokens in $S'$ had not yet exited the ordering gadget for the first time, and thus first exited the ordering gadget during interval $\mathcal{I}'$. 
This completes the proof since  by Lemma~\ref{lem:noncore}, $|C(S)|-k^{1/2+3/c}-k^{1/2+2/c}>k^{1-9/c}-k^{1/2+3/c}-k^{1/2+2/c}>k^{1-10/c}$ since $c>24$.
\end{proof}

Now we are ready to complete the proof of Lemma~\ref{lem:LeafLeaving} by deriving the contradiction that $t$ has become out of order with a large set of tokens. When interval $\mathcal{I}'$ both begins and ends, $t$ is at the root. By Claim~\ref{claim:slotfromt}, during interval $\mathcal{I}'$, a set $S$ of at least $k^{1-10/c}$ tokens that all have the same initial and final gadget as $t$ exit their initial gadget for the first time and enter their final gadget for the last time. This immediately implies that at the beginning of interval $\mathcal{I}'$, every token in $S$ is in its initial gadget (which is the same as $t$'s initial gadget), and at the end of interval $\mathcal{I}'$, every token in $S$ is in its final gadget (which is the same as $t$'s final gadget). Thus, for every token $s\in S$, $t$ is out of order with $s$ either at the beginning or the end of interval $\mathcal{I}'$. Thus, either at the beginning or the end of interval $\mathcal{I}'$, $t$ is out of order with at least $|S|/2\geq k^{1-10/c}/2$ tokens. Since $c>24$, this contradicts Lemma~\ref{lem:OutOfOrder}, which says that $t$ can only be out of order with at most $k^{1/2+2/c}$ tokens at any given point in time.
\end{proof}

\paragraph{Proof of Lemma~\ref{lem:free-item-token}}\label{sec:Item and Segment Swap}$ $\\

We begin with Lemma~\ref{lem:free-item-token}(1): \emph{At any point in time, apart from the free item token, there is exactly one item token in each slot gadget.}
\begin{proof}[Proof of Lemma~\ref{lem:free-item-token}(1)]

By Lemma~\ref{lem:LeafLeaving}, at all times, each slot gadget contains at least one item token. It remains to show that if the free item token is in a slot gadget $i$, then there is another item token in slot gadget $i$. Consider the token $t$ that was in slot gadget $i$ the last time the free item token $f$ was at the root. We claim that $t$ is still in slot gadget $i$. This is because if $t$ has left slot gadget $i$ then $t$ has passed through the root, which would mean that $t$ is the free item token instead of $f$ (since we were considering the \emph{last} time $f$ was at the root).
\end{proof}

Because any token must pass through the root when moving from one slot gadget to another, Lemma~\ref{lem:free-item-token}(1)
immediately implies the following corollary: 

\begin{corollary}
When the free item token changes from an item token $f$ to another item token $f'$, $f$ is in the slot gadget that $f'$ was just in.
\end{corollary}

Now we move to Lemma~\ref{lem:free-item-token}(2): \emph{The exchange sequence $\chi$ is a  subsequence of the sequence of swaps given as input to the Star STS instance.}

Towards proving this, we show in the following claim that there is a direct correspondence between the free item token changing and segments leaving a slot gadget.

\begin{claim}\label{lem:ItemSwap}

Suppose the free item token changes from item token $f$ to item token $f'$, and let $i$ be the slot gadget containing $f$ when this happens. Then, while $f$ was the free item token, at least $k(1-\frac{1}{3n})$ slot tokens that were initially in slot gadget $i$ reached the root for the first time.

\end{claim}

\begin{proof}

This proof has a similar idea to the proof of Claim~\ref{claim:up} (but is less involved).

Let $v$ be the vertex in slot gadget $i$ of distance exactly $k(1-\frac{1}{4n})$ from the root (this is the same definition of $v$ as in the previous section). By Lemma~\ref{lem:LeafLeaving}, when $f'$ becomes the free item token, $f$ is at $v$ or to the right of $v$. Previously, when $f$ became the free item token, $f$ was at the root. Thus, while $f$ was the free item token, $f$ moved from the root to $v$. Let $\mathcal{S}$ be the set of swaps that $f$ performed while traveling from the root to $v$, taking the last such swap that $f$ performs for each edge between the root and $v$; thus, $|\mathcal{S}|=k(1-\frac{1}{4n})$. Let $S'$ be the set of tokens that swap with $f$ during swaps in $\mathcal{S}$. By Claim~\ref{lem:unique}, $|S'|=k(1-\frac{1}{4n})$.

First, we will show that many of the tokens in $S'$ were initially in slot gadget $i$. Less than $k^{1/2+1/c}$ of the tokens in $S'$ do not have slot gadget $i$ as their initial or final gadget, because otherwise it requires $k^{2(1/2+1/c)}/2$ contrary moves to get these tokens into slot gadget $i$ to the position of the swap in $\mathcal{S}$, or to get them back out of slot gadget $i$ (since each token in $S'$ swaps with $f$ along a distinct edge). Further, note that if a token $s\in S'$ has its final position in slot gadget $i$, then during its swap with $f$, $s$ is performing a contrary (left) move, which is not in $\mathcal{M}$ by Corollary~\ref{obs:interface}(\ref{int5}). Thus, less than $k^{1/c}$ of the tokens in $S'$ have final position in slot gadget $i$. Therefore, we have shown that a set $S''\subseteq S'$ of at least $k(1-\frac{1}{4n})-k^{1/2+1/c}-k^{1/c}$ tokens were initially in slot gadget $i$.

Now, we will show that many of the tokens in $S'$ reached the root for the first time while $f$ was the free item token.  By Lemma~\ref{lem:LeafLeaving}, once $f$ became the free item token and until $f$ reached $v$, $f'$ was in slot gadget $i$ at $v$ or to the right of $v$. Thus, right after $f$ swaps with each token $s\in S'$, $s$ is to the left of $f'$.

We claim that after $f$ reaches $v$, there are never more than $k^{1/2+1/c}$ tokens in $S''$ simultaneously to the right of $v$. This is because right moves are contrary moves for these tokens (since they were initially in slot gadget $i$) so having $k^{1/2+1/c}$ of these tokens to the right of $v$ would incur $k^{2(1/2+1/c)}/2>k^{1+1/c}$ contrary moves. In particular, when $f'$ reaches $v$ for the last time before becoming the free item token, no more than $k^{1/2+1/c}$ tokens in $S''$ are to the right of $f$. For each token $s\in S''$ that is \emph{not} to the right of $f$ at this point in time, note that $s$ either reaches the root before $f'$ or swaps with $f'$. Less than $k^{1/c}$ such tokens $s$ swap with $f'$ since in this case, $s$ performs a contrary (right) move that is not in $\mathcal{M}$ by Corollary~\ref{obs:interface}(\ref{int5}).

Thus, we have shown that at least $|S''|-k^{1/2+1/c}-k^{1/c}$ tokens in $S''$ have reached the root while $f$ is the free item token. It remains to show that enough of these tokens reach the root \emph{for the first time} while $f$ is the free item token. We know that at most $k^{1/2+1/c}$ of the tokens in $S''$ reached the root for the first time before $f$ became the free item token because getting all of these tokens to their positions to swap with $f$ would require more than $k^{2(1/2+1/c)}/2>k^{1+1/c}$ contrary moves. Thus, we have shown that at least $|S''|-2k^{1/2+1/c}-k^{1/c}$ tokens in $S''$ reach the root for the first time while $f$ is the free item token. This quantity of tokens is at least $k(1-\frac{1}{4n})-3k^{1/2+1/c}-2k^{1/c}>k(1-\frac{1}{3n})$.
\end{proof}

Now we are ready to prove Lemma~\ref{lem:free-item-token}(2).

\begin{proof}[Proof of Lemma~\ref{lem:free-item-token}(2)]
Let $s_1,\dots,s_n$ be the input sequence for the Star STS instance. We begin by recalling the correspondence between $s_1,\dots,s_n$ and our token swapping solution. Recall from the construction that there is a direct correspondence between $s_1,\dots,s_n$ and the big segments $y_1,\dots,y_n$ in the sense that if $s_j=g$ then $y_j$ is initially in slot gadget $g$. Further recall that in the scaffold solution, the big segments reach the root in the order $y_1,\dots,y_n$. By Lemma~\ref{lem:segbyseg}, in our token swapping solution the vertices in the \defn{core} of the big segments reach the root for the first time in the same order as in the scaffold solution. That is, in the order $C(y_1),\dots,C(y_n)$. We will refer to this property as the \defn{core property}: 

{\bf The core property:} For all $i,j$ with $i<j$, every token in $C(y_i)$ reaches the root for the first time before any token in $C(y_j)$ does so.

We will maintain a counter $\Lambda$ that intuitively keeps track of where we are up to in the sequence $s_1,\dots,s_n$. We will iterate through the exchange sequence $\chi=\chi[1],\chi[2],\dots$ from beginning to end, and for each element $\chi[i]$, we will increase $\Lambda$ by at least 1 to a new value $i$ such that $s_j=\chi[i]$. It is straightforward to see that if we can successfully maintain such a counter, then $\chi$ is a subsequence of $s_1,\dots,s_n$.

 Now we will define the rule for updating $\Lambda$. For each $\chi[i]$, let $p_i$  be the corresponding point in time when the token that was just in slot gadget $\chi[i]$ becomes the free item token (and let $p_0$ be the time initially). At each $p_i$, we will update $\Lambda$; let $\Lambda(i)$ be the value that we set $\Lambda$ to at $p_i$. 
 
 {\bf Rule for updating $\Lambda$:} $\Lambda(0)=0$ and for all $i\geq 1$, set $\Lambda(i)$ to be the smallest $j>\Lambda(i-1)$ such that before $p_i$ at least one token from $C(y_j)$ reaches the root for the first time, where $y_{j}$ is initially in slot gadget $\chi[i]$. 
 
 Note that the requirement that $y_j$ is initially in slot gadget $\chi[i]$ is equivalent to the above requirement that $s_j=\chi[i]$ (due to the correspondence between $y_j$ and $s_j$ described above). Thus, given this rule for updating $\Lambda$, proving the following claim completes the proof of Lemma~\ref{lem:free-item-token}(2).
 
 \begin{claim}
 For all $i\geq 1$, there exists $j>\Lambda(i-1)$ such that before $p_i$, at least one token from the core of some big segment $y_{j}$ that was initially in slot gadget $\chi[i]$ reaches the root for the first time.
 \end{claim}
 
 \begin{proof}
 Fix $i\geq 1$. Suppose inductively that $\Lambda(i')$ exists for all $i'<i$. By the definition of how we update $\Lambda$, before $p_{i-1}$, at least one token from $C(y_{\Lambda(i-1)})$ has reached the root for the first time, and $y_{\Lambda(i-1)}$ is initially in slot gadget $\chi[i-1]$. By the core property, this means that before $p_{i-1}$, every token in the core of every big segment $y_{j'}$ with $j'<\Lambda(i-1)$ has reached the root for the first time. 
 
Now, we consider what happens between $p_{i-1}$ and $p_i$. By Claim~\ref{lem:ItemSwap}, between $p_{i-1}$ and $p_i$, a set $S$ of at least $k(1-\frac{1}{3n})$ tokens that were initially in slot gadget $\chi[i]$ reach the root for the first time. 

We analyze the composition of the tokens in $S$. Since $S$ contains only tokens that reached the root for the first time after $p_{i-1}$, $S$ contains no tokens in the core of any big segment $y_{j'}$ with $j'<\Lambda(i-1)$. By Lemma~\ref{lem:noncore}, less than $k^{1/2+5/c}$ tokens in $S$ are not in the core of any segment. Also, at most $\frac{k}{n^7}$ tokens in $S$ are in a padding segment since there are $n$ padding segments in each slot gadget and each one is of length $\frac{k}{n^8}$. The remainder of the tokens in $S$ are in the core of a big segment $y_{j'}$ with $j'\geq \Lambda(i-1)$. Thus, letting $\beta=k^{1/2+5/c}+\frac{k}{n^7}$, we have that $S$ contains at least $|S|-\beta$ tokens that are in cores of big segments $y_{j'}$ with $j'\geq \Lambda(i-1)$.
This does not suffice, however, because we are trying to identify a token in $S$ in the core of a big segment $y_{j'}$ with $j'> \Lambda(i-1)$, where the inequality is strict. 

If $i=1$, then $\Lambda(i-1)=0$ so it is trivially true that any big segment $y_{j'}$ has $j'>\Lambda(i-1)$. If $i>1$, then we need to bound the number of tokens in $S$ that are in $C(y_{\Lambda(i-1)})$. 
If $\chi[i]\not=\chi[i-1]$, then $S$ contains no tokens in $y_{\Lambda(i-1)}$ since $y_{\Lambda(i-1)}$ was initially in slot gadget $\chi[i-1]$ while every token in $S$ was initially in slot gadget $\chi[i]$.

Thus, we are concerned with the case where $\chi[i]=\chi[i-1]$. Let $\ell\geq 1$ be the largest integer such that $\chi[i]=\chi[i-1]=\dots =\chi[i-\ell]$. Then by Claim~\ref{lem:ItemSwap}, between $p_{i-\ell-1}$ and $p_{i}$, a set $S'$ of at least $(\ell+1) k(1-\frac{1}{3n})$ tokens that were initially in slot gadget $\chi[i]$ reach the root for the first time.  To complete the proof, we will show that at least one token in $S'$ is in the core of a big segment $y_{j'}$ with $j'> \Lambda(i-1)$. 

We can analyze the composition of $S'$ in the same way as we analyzed the composition of $S$ previously. $S'$ contains no tokens in the core of any big segment $y_{j'}$ with $j'<\Lambda(i-\ell-1)$. Also $S'$ contains at least $|S'|-\beta$ tokens that are in the core of a big segment $y_{j'}$ with $j'\geq \Lambda(i-\ell-1)$. Since $\chi[i-\ell]\not=\chi[i-\ell-1]$, $S'$ contains no tokens in $y_{\Lambda(i-\ell-1)}$. Thus, the number of tokens in $S'$ that are in the core of a big segment $y_{j'}$ with $j'> \Lambda(i-\ell-1)$ is at least 
\begin{align}\label{eqn}
|S'|-\beta&\geq (\ell+1) k(1-\frac{1}{3n})-\beta\notag\\
&= \ell k+k-\frac{\ell k}{3n}-\frac{k}{3n}-\beta\notag\\
    &\geq \ell k +k -\frac{k}{3}-\frac{k}{3n} -\beta \text{\hspace{5mm} since $\ell\leq n$}\notag\\
    &=\ell k +\frac{2k}{3}-\frac{k}{3n} -k^{1/2+5/c}-\frac{k}{n^7}\notag\\
    &\geq \ell k+1.
\end{align}

By the core property, the tokens in $S'$ that are in the core of big segments come from big segments whose initial positions are \emph{consecutive} within slot gadget $\chi[i]$. Let $y^*$ be the big segment with the smallest index larger than $\Lambda(i-\ell-1)$ out of the big segments that were initially in slot gadget $\chi[i]$. Since $S'$ contains no tokens in the core of any big segment $y_{j'}$ with $j'<\Lambda(i-\ell-1)$, $y^*$ is the lowest indexed big segment for which $S'$ could contain tokens from its core. If $S'$ contains any token in $C(y^*)$, then Equation~(\ref{eqn}) implies that $S'$ contains the entire core of the next $\ell-1$ consecutive big segments initially in slot gadget $\chi[i]$, as well as at least one token from the core of the next big segment $y^{**}$ from slot gadget $\chi[i]$ directly after (to the right of in the initial configuration) these $\ell-1$ consecutive big segments. Since $y^*$ is the lowest indexed big segment for which $S'$ could contain tokens from its core, if $S'$ does not contain a token in $C(y^*)$, then $S'$ still contains at least one token from the core of some big segment with index at least that of $y^{**}$.

Together, the core property and the fact that the rule for updating $\Lambda$ chooses the \emph{smallest} such $j$ imply  that $y_{\Lambda(i-\ell)},\dots,y_{\Lambda(i-1)}$ are big segments whose initial positions are consecutive within slot gadget $\chi[i]$. Furthermore, $\Lambda(i-\ell)$ is set to be the minimum value larger than $\Lambda(i-\ell-1)$ such that $y_{\Lambda(i-\ell)}$ is initially in slot gadget $\chi[i]$. That is, $y_{\Lambda(i-\ell)}=y^*$. Then, because $y_{\Lambda(i-\ell)},\dots,y_{\Lambda(i-1)}$ are consecutive big segments initially in slot gadget $\chi[i]$, we have that $y^{**}=y_{j'}$ for $j'>\Lambda(i-1)$.

We have shown that $S'$ contains at least one token from the core of some big segment with index at least that of $y^{**}$. Thus, $S'$ contains at least one token from the core of some big segment $y_{j'}$ with $j'>\Lambda(i-1)$, which completes the proof. 
\end{proof}

Therefore we have proved Lemma~\ref{lem:free-item-token}(2).
\end{proof}

}

%% file: no_better_algs.tex
\section{Known techniques preclude
approximation factors less than 2}

\label{sec:inapprox}

Previous results about sequential token swapping on trees include three different polynomial time 2-approximation algorithms and some lower bounds on approximation factors.
The  three algorithms all have the property that if a token at a leaf is at its destination (i.e., it is a ``happy leaf token'') then the algorithm will not move it.  
Biniaz et al.~\cite{tree-token-swapping} proved that any algorithm 
with  this  property
has a (worst case) approximation factor at least $\frac{4}{3}$.  They also proved via ad-hoc arguments that the approximation factor is exactly 2 for two of the known 
2-approximation algorithms  (the ``Happy Swap Algorithm'' and the ``Cycle Algorithm'').  For the third 2-approximation algorithm, the Vaughan-Portier algorithm, they could not prove  a lower bound better than $\frac{4}{3}$.

We prove that the Vaughan-Portier algorithm has approximation factor exactly 2. We also extend the approximation lower bounds of Biniaz et al.~by proving that the approximation factor is at least 2 for a larger family of algorithms. %
We formalize this family as follows.
For token $t$ let $P_t$ be the  path from $t$'s initial position to  its final position.  
A sequence of token swaps is \defn{$\ell$-straying} if at all intermediate points along the sequence, every token $t$ is within distance $\ell$ of the last vertex of $P_t$ that it has reached up to this point.
A token swapping algorithm is \defn{$\ell$-straying} if it produces $\ell$-straying sequences.

Both approximation lower bounds will be proved for the same family of trees that was used by Biniaz et al.~\cite{tree-token-swapping}.  For any $k$ and  any  odd $b$ we  define a tree $T_{k,b}$ together with initial and  final positions of tokens. The tree $T_{k,b}$ has $b$ paths of length $k$ attached to a central vertex $c$,  and a set $L$ of $k$ leaves  also attached to $c$. 
See Figure~\ref{fig:branch-example}.
The tokens at $c$ and $L$ are \defn{happy}---they are at their final positions. The tokens in branch $i, 0 \le i \le b-1$, have their final positions in  branch $i+1$, addition  modulo $b$, 
with the initial and final positions equally far from the center $c$.

\begin{figure}[htb]
    \centering
    \includegraphics[width=.35\textwidth]{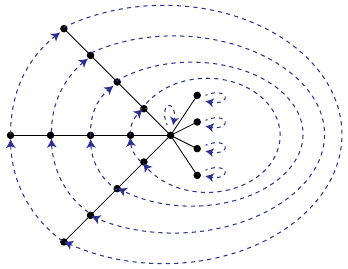}
    \caption{Tree $T_{k,b}$ with $b=3$ branches each of length $k=4$, and with $k=4$ leaves attached to the center node.  The dashed arrows go from a token's initial to final  position. The figure is from \cite{tree-token-swapping}.}
    \label{fig:branch-example}
\end{figure}

Biniaz et al.~\cite{tree-token-swapping} 
proved that the optimum number of swaps for $T_{k,b}$ is at most $(b+1) ({k+1 \choose 2} +2k)$.  The solution repeatedly exchanges the tokens in branch $i, 0 \le i \le b - 1 $, modulo  $b$,  with the tokens at $L$.  The first exchange moves the tokens initially at $L$ into branch 0, and the $(b+1)^{\rm st}$  exchange moves those tokens back to $L$.
We prove that for $T_{k,b}$ the approximation factor is not better than 2 for $\ell$-straying algorithms (Subsection~\ref{sec:ell-straying}) 
and for the Vaughan-Portier algorithm (Subsection~\ref{sec:Vaughan}). .

\later{

\subsection{\texorpdfstring{$\ell$}{l}-straying algorithms do not achieve better than 2-approximation}
\label{sec:ell-straying}

Before proving the  main result, we compare $\ell$-straying algorithms to previous algorithms/properties.
It is easy to show that the Happy Swap Algorithm and the Cycle Algorithm are 1-straying   (see ~\cite{tree-token-swapping}  for descriptions  of these algorithms).  Thus our lower bound of 2 on the approximation factor replaces the ad hoc arguments  of Biniaz et al.~\cite{tree-token-swapping}.

We next compare the property of being $\ell$-straying to the  property of fixing happy leaves.
A swap sequence is \defn{minimal} if it never swaps the same two tokens more than once.  It is  easy to test this property and eliminate duplicate swaps. 

\begin{claim} A minimal 1-straying sequence does not move happy leaf tokens. 
\end{claim}
\begin{proof} Suppose there is a happy leaf vertex $v$ with a  token  $t$ that is at its destination. Let $u$ be the vertex  adjacent to $v$. The path $P_t$ from $t$'s initial to final vertex  consists only of vertex $v$ so in a 1-straying swap sequence, token $t$ can  only be at vertex $v$ or $u$.  Suppose that at some point  in the sequence there is a swap that moves $t$ to $u$ and moves a token, say $s$, from $u$ to $v$.  The next time token $t$ moves, it must move back to $v$, and since $v$  is a leaf, token $s$ is stuck at $v$ until then. %
Thus tokens $t$ and $s$ must re-swap on edge $uv$ at some later point in  the  sequence.  Removing the two swaps of tokens $t$ and $s$ yields a shorter equivalent 1-straying sequence.
\end{proof}

We now prove the main result of this section.

\begin{theorem} %
$\ell$-straying  algorithms do not achieve approximation factor less than 2.
More precisely, 
for every  small $\epsilon, \gamma > 0$, there exists an $n$ and an instance of token swapping on an $n$-node tree so that any $\ell$-straying algorithm for $\ell\leq n^{1-\gamma}$ cannot achieve an approximation factor better than $2-\epsilon$. 
\end{theorem}
\begin{proof}
We will prove a lower bound on  the number of swaps used by  any   $\ell$-straying algorithm on $T_{k,b}$.
Then we will  choose values of $b$ and $k$ to get the ratio of the algorithm versus the optimum arbitrarily close to 2.

We count two kinds of swaps performed by  an $\ell$-straying algorithm: type-A swaps between tokens with initial positions in the same branch; and type-B swaps between tokens with initial positions in different branches.

\begin{lemma}
\label{lemma:typeA}
The number of type-A swaps is at least $\frac{1}{2}  b (k-b\ell)^2$.
\end{lemma}
\begin{proof} We count  the number of type-A swaps for one branch  $B$.
Fix a token $t$  whose initial position is branch $B$ and at distance $i$ from the center, $1 \le i \le k$.  Let $S(t)$ be the set of tokens whose initial position is in branch $B$ and closer to the center than $t$'s initial position.  Then $|S(t)| = i-1$.  Intuitively, we can think of $S(t)$ as a set of tokens that $t$ must swap with if we restrict to $t$'s initial and final branches.
Let $H$ be the set of vertices within distance $\ell$ of the center.   Observe that $|H| = b\ell + k + 1$. Let $H(t)$ be the tokens at vertices of $H$ at the point in time right before $t$ moves from the center vertex into  its final branch for the first time. 
Intuitively, the only tokens of $S(t)$ that can ``hide'' and avoid swapping with $t$ are in   $H(t)$.  

\begin{claim}
\label{claim:same-branch-swaps}
Token $t$ swaps with every token in $S(t) \setminus H(t)$.
\end{claim}

Before proving the claim we work out the bounds. 
The token initially at the center and the tokens at the $k$ leaves attached to the center were happy initially  so they must stay within distance $\ell$ of their initial positions, and thus within distance $\ell$  of the center, which  means that these tokens lie in $H(t)$. Also $t$ lies in $H(t)$.  Thus $|S(t) \cap  H(t)| \le b\ell -1$, and $|S(t) \setminus H(t)| \ge (i-1) - (b\ell-1) = i - b \ell$. 
Then the number of swaps per branch is at least 
$$\sum_{i=b\ell}^{k} (i-b\ell) = \sum_{i=0}^{k-b\ell} i  = \frac{1}{2}(k-b\ell +1)(k-b\ell).$$
Summing over all $b$ branches gives a total of at least $\frac{1}{2} b (k-b\ell )^2$ type-A swaps.

\begin{proof}[Proof of Claim~\ref{claim:same-branch-swaps}]
Fix a token $t' \in  S(t) \setminus H(t)$.  Note that $t'$ always lies in its initial or final branch or at a vertex of $H$.  At the point in time right before $t$ moves from the center vertex into  its final branch for the first time, $t'$ is outside $H$, so it must be more than distance $\ell$ from the center and either: (1) in its initial branch; or (2) in its final branch.
In case (1) $t'$ has never reached the center by the definition of $\ell$-straying.    But $t$ has reached the center, and was initially further from the center than $t'$, so $t$ and $t'$ must have swapped in the past.
In case (2) $t'$ 
will never leave its final branch in the future by the definition of $\ell$-straying. Meanwhile, $t$ is currently closer to the center than $t'$ and will end farther from the center than $t'$ and on the same branch. This means that $t$ and $t'$ must swap in the future. 
\end{proof}

This completes the proof of Lemma~\ref{lemma:typeA}. 
\end{proof}

\begin{lemma}
\label{lemma:typeB}
The number of type-B swaps is at least
$\frac{1}{2}b(k-2\ell)^2$.
\end{lemma}
\begin{proof}
Fix a branch $B$, and consider its \defn{tail} $T$ which consists of its last (furthest from the  center) $k-2\ell$ vertices.
Let $I(B)$ [$F(B)$] be the set of tokens whose initial [final, respectively] position is in the branch, and let $I(T)$ [$F(T)$] be the set of tokens whose initial [final] position is in the tail. If a token $t$ outside $I(B)$ enters the tail at some point in time, then, by the definition of $\ell$-straying, $t$'s final position 
cannot be outside $B$, nor among the first $\ell$ vertices of $B$.  This implies that,  
from this time on, $t$ must stay in $B$. 
Thus at all points in time, the  tokens in the tail are a subset of $I(B) \cup F(B)$.  
Every one of the $k  - 2\ell$ tokens in $F(T)$ must enter the tail at some point in time.  (Note that other tokens of $F(B)$ may  enter the tail and then leave the  tail.)   Let $t_1, t_2, \ldots, t_{k-2\ell}$ be a list of the first $k-2\ell$ tokens of $F(B)$ 
that enter the tail, ordered by their first entry time.   When  $t_i$ is about to enter the tail, there are at most $i-1$ tokens of $F(B)$ in the tail, so there are at least $k-2\ell-(i-1)$ tokens of $I(B)$ in  the tail.    Furthermore,  $t_i$ must swap with all of them, since they leave $B$ and $t_i$ stays inside $B$ from this point on.  
Thus the total number of swaps between a token of $F(B)$ and a token  of $I(B)$ is at least 
$$\sum_{i=1}^{k-2\ell} (k-2\ell-(i-1)) = \sum_{j=1}^{k-2\ell}  j  = \frac{1}{2}(k-2\ell+1)(k-2\ell).$$ 
The total over all $b$ branches is then at least  $\frac{1}{2} b (k  - 2\ell)^2$.

\end{proof}

To complete the proof of the theorem we calculate the approximation factor for $T_{k,b}$. 
Let ALG be the  number  of swaps performed by an $\ell$-straying algorithm, and  let OPT be the optimum number of swaps.
By the above two lemmas, ${\rm ALG} \ge \frac{1}{2}b (k-b\ell)^2 + \frac{1}{2}b(k-2\ell)^2  \ge b(k-b\ell)^2$ for $b \ge 2$. 
Biniaz et al.~\cite{tree-token-swapping} proved that ${\rm OPT} \le (b+1) ({k+1 \choose 2} +2k)$. 

Choose  $b =  k^\delta$ for any small $\delta$, $0 < \delta < \frac{1}{2}$.
Then ${\rm OPT} \le \frac{1}{2} k^{2+\delta} + O(k^{2})$ and ${\rm ALG} \ge k^{2+\delta} -  O(k^{1 + 2 \delta}\ell)$.
This gives ${\rm ALG}/{\rm OPT} \ge 2 - O(\frac{1}{k^\delta})$. 

To prove the  stronger claim in the theorem, let $\epsilon,\gamma>0$ be given. Pick $k$ large enough so that $\epsilon$ is bigger than $24/k^\delta$. Let $b=k^\delta$ as before.
However, now pick $\delta$ in terms of $\gamma$, 
so that $k^{1-2\delta}=(kb)^{1-\gamma}=k^{1+\delta-\gamma-\delta\gamma}$, and hence $\delta=\gamma/(1.9-\gamma)$ suffices.
\end{proof}

The fact that the lower bound on the approximation factor holds even if $\ell$ grows almost linearly in the size of the tree is interesting because it means that any approximation algorithm that performs better than the current best-known algorithms must be drastically different.

\subsection{The Vaughan-Portier algorithm does not achieve better than 2-approximation}
\label{sec:Vaughan}

The previous section shows that any $\ell$-straying algorithm for token swapping on a tree cannot achieve an approximation factor better than 2.
Of the three known approximation algorithms, two of them---the Happy Swap Algorithm and the Cycle Algorithm---are 1-straying. (This follows easily from  the descriptions of those algorithms  in ~\cite{tree-token-swapping}.)
The third algorithm, the Vaughan-Portier algorithm, does not have this property, since tokens can stray far from their shortest paths. In this section we prove that the Vaughan-Portier algorithm cannot achieve an approximation factor better than 2.

\subsubsection{Description of the Vaughan-Portier algorithm}

Most token swapping algorithms find the sequence of swaps in order  from first to last, i.e., based on the initial and final token configurations, the algorithm decides on the first swap to be made.  This swap changes the initial token configuration, and the algorithm is applied recursively to the new initial configuration and the unchanged final configuration.  The Vaughan-Portier algorithm is different.  It exploits the symmetry that the  reverse of a token swapping sequence changes the final configuration to the initial one. In particular, her algorithm may add a swap to  either the start or the end of the sequence. If a swap is added at the start of  the sequence, the initial configuration changes.   If a swap  is added at  the end  of the sequence, the final configuration changes.  In  either case, the  algorithm is applied recursively.

We describe and analyze her algorithm by thinking of each token as a pair of tokens, $t,t_f$, where $t$ is our usual token and $t_f$ is a ``destination token'' that is placed initially on $t$'s destination node.  The goal is to get $t$ and $t_f$ to the same node of the tree.  At any point  in time any tree node has one token and one destination token.
Swapping adjacent tokens $s$ and $t$ is the usual  operation.  Swapping adjacent destination tokens $s_f$ and $t_f$ means that the \emph{last} swap in the sequence will be of tokens $s$ and $t$. 

\paragraph*{Notation:} If a token and its destination are at the same node we say that the token and the destination token are \defn{happy}. We say that an unhappy token \defn{wants} to move along its incident edge towards its destination token. Symmetrically, we say that a destination token \defn{wants} to move along its incident edge towards its corresponding token.

The Vaughan-Portier algorithm uses the following three  operations.  See also Figure~\ref{fig:Vaughan-swaps}.
\begin{description}
\item[Happy swap.] If $s$ and $t$ are two adjacent tokens that want to move towards each other, then 
output the swap $(s,t)$ and recursively solve the problem that results from swapping $s$ and $t$.
\item[Happy destination swap.] If $t_f$ and $s_f$ are two adjacent destination tokens that want to move towards each other, then  recursively solve the  problem that results from swapping $t_f$ and $s_f$, and then output the swap $(s,t)$.
\item[Symmetric shove.]  Suppose $u$ and $v$ are adjacent nodes, where $u$ has  token $t$ and destination token $r_f$, and $v$ has token $s$ and destination $s_f$ (so $s$ and $s_f$ are happy). Suppose that $t$ and $r_f$ both want to move towards $v$.   
Then output the swap $(s,t)$, recursively solve the problem that results from swapping  
$s,t$ and swapping $r_f,s_f$, then output the swap $(r,s)$.  
\end{description}

\begin{figure}[h]
    \centering
    \includegraphics[width=.7\textwidth]{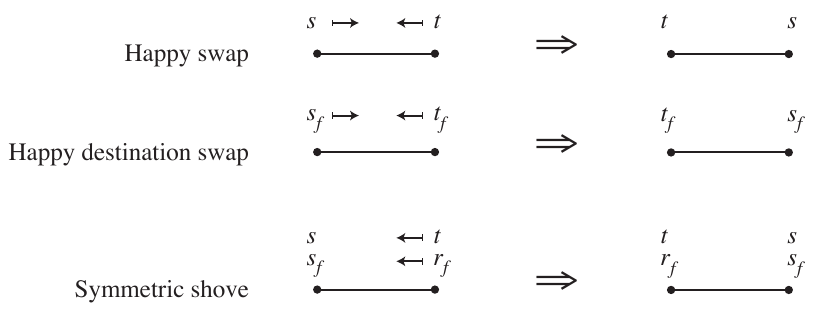}
    \caption{The three operations of the Vaughan-Portier algorithm.  An arrow beside a [destination] token indicates the direction it wants to move.}
    \label{fig:Vaughan-swaps}
\end{figure}

The Vaughan-Portier algorithm applies these operations in order.  If there is a happy  swap or a happy  destination swap that can be  performed, do  so.   Otherwise perform a symmetric shove---she  proves that one exists in this situation.   Then repeat.

The Vaughan-Portier algorithm has several important properties:
\begin{description}
\item{(P1)} Once a token is happy it stays happy.
\item{(P2)} Unhappy tokens and destination tokens only move in the direction they want to move (while happy tokens can move arbitrarily). 
\item{(P3)} If $\sigma$ is a sequence of swaps returned by the algorithm then the reverse of $\sigma$ is a possible output of the algorithm for the  instance where 
tokens and their destinations are flipped, i.e., each pair $t, t_f$ is replaced by  $t_f,  t$.
\end{description}

Property (P1) is true because the only operation that involves a happy token is a symmetric shove and the happy token $s$ remains happy. 

Property (P2) is true because every operation that moves an unhappy [destination] token moves it in the direction it wants to go.

Property (P3) is true because in the  reversed problem happy  swaps and happy  destination swaps exchange roles, and a symmetric shove remains a symmetric shove.

\begin{remark} We note that the Vaughan-Portier algorithm is not $\ell$-straying even if $\ell$ is $(n-1)/2$, and even if the tree is just a path. Say $n=2k+1$ is odd. The tokens in consecutive order on the path are $k+2,k+3,\ldots,2k+1,k+1,1,2,\ldots,k$, and the target of token $i$ is position $i$ of the path. Token $k+1$ is happy. A valid application of the Vaughan-Portier algorithm
repeats the following for each $i$ from $k$ down to $1$: perform a symmetric shove of token $k+1+i$ with token $k+1$, and then use only happy swaps to move token $k+1+i$ to position $k+1+i$. 
This causes token $k+1$ to end up in position $1$, which is at distance $(n-1)/2$ from its desired destination and path.
\end{remark}

\subsubsection{The Vaughan-Portier algorithm on \texorpdfstring{$T_{k,b}$}{T\_\{k,b\}}}

\begin{theorem}
\label{thm:Vaughan-bound}
The approximation factor of the Vaughan-Portier algorithm is not less than 2.
\end{theorem}
\begin{proof}
We will show that on  the tree $T_{k,b}$ the  number of swaps, ALG, performed by the Vaughan-Portier algorithm is at least $bk(k-3)$ swaps.
Recall that the optimum number of  swaps, OPT, is at most $(b+1)({{k+1} \choose 2} +2k) $.
For $k=b$, we get  ${\rm OPT} \le \frac{1}{2} k^3 + \Theta(k^2)$, and ${\rm ALG} \ge k^3 - \Theta(k^2)$. 
Then ${\rm ALG} / {\rm OPT} \ge 2 - \Theta(\frac{1}{k})$, which proves the theorem.  

It remains to  show that the Vaughan-Portier algorithm performs at least $bk(k-3)$ swaps on $T_{k,b}$. 
We will consider the dynamically changing quantity $D=\sum_t d(t,t_f)$, that is, the sum over all tokens $t$ of the distance from the current position of $t$ to the current position of the destination token $t_f$. Note that a happy swap and a happy destination swap both decrease $D$ by 2. A symmetric shove also decreases $D$ by 2, however it performs two swaps so a symmetric shove decreases $D$ by an average of 1 per swap. Since happy swaps and happy destination swaps are ``cheaper'' in this sense than symmetric shoves, to prove a lower bound on the number of swaps, we will prove an \emph{upper bound} on the number of happy swaps and happy destination swaps. Specifically, we will prove that the total number of happy swaps is at most $bk$.

\begin{lemma}
\label{lemma:happy-swap-bound}
The total number of happy swaps is at most $bk$.
\end{lemma}

We first show that this lemma gives the bound we want.  By Property (P3) the lemma also  
implies that the total number of happy destination swaps is at most $bk$.
Now consider the quantity $D$ defined above. Initially, $D=b\sum_{i=1}^k 2i=bk(k+1)$ and at the end, $D=0$. Since each happy swap and each happy destination swap decreases $D$ by 2, these two operations can each decrease $D$ by a total of at most $4bk$. Thus, symmetric shoves account for decreasing $D$ by at least $bk(k+1)-4bk=bk(k-3)$. Since each swap in a symmetric shove decreases D by an average of 1, the total number of swaps is at least $bk(k-3)$, as required.

It remains to prove 
Lemma~\ref{lemma:happy-swap-bound}.  We begin with two claims about the relative ordering of tokens and their destination tokens during the course of the Vaughan-Portier algorithm on the tree of branches.

\noindent
\textbf{Notation:} For any token $t$, let $Q_t$ be the path consisting of the branch initially containing $t$ followed by the center node followed by the branch initially containing $t_f$.

\begin{claim}
\label{claim:ordering}
If $t$ is unhappy, then $t$ and $t_f$ are on path $Q_t$, and appear in the order $t,t_f$.  
\end{claim}
\begin{proof}
By Property (P2), since $t$ is unhappy, so far $t$ and $t_f$ have only moved in the direction they want to go, which is along $Q_t$. Thus, $t$, and $t_f$ are both on $Q_t$. Note that $t$ initially appears before $t_f$ on $Q_t$. We will show that $t$ and $t_f$ have not switched order. Since $t$ and $t_f$ have never left $Q_t$, the only way for them to have switched order is for them to have swapped or for them to have arrived at the same node. It is not possible for $t$ and $t_f$ to swap because every swap either involves two tokens or two destinations. If they arrive at the same node then $t$ becomes happy and, by Property (P1), stays happy, contradicting our assumption that $t$ is unhappy. 
\end{proof}

\begin{claim}
\label{claim:pair-ordering}
Suppose tokens $t$ and $s$ have the same initial branch, with $t$ further from the center.  If $t$ and $s$ are unhappy, then their  order along the path $Q_t=Q_s$ is $t, s, s_f, t_f$.  
\end{claim}
\begin{proof}
Initially, the  tokens appear in  the  order $t,s,s_f,t_f$ along the  path.  
By  Claim~\ref{claim:ordering}, $s$ and $s_f$ maintain that order while $s$ is unhappy.  Thus, the only  way the ordering could change is if $t$ and $s$  change order, or $s_f$ and  $t_f$ change order.   Consider $t$ and $s$.  They cannot lie at the same node, so  the only way  they can  change order is by swapping, but this involves $s$ moving away from $s_f$.  Similarly, if $s_f$ and $t_f$ change order then  $s_f$ moves away from $s$.  Thus the four tokens  remain  in their  initial order.
\end{proof}

In order to prove Lemma~\ref{lemma:happy-swap-bound} (that there are at most $bk$ happy swaps), we will prove that during every happy swap, one of the tokens involved becomes happy.  Since there are $bk$ unhappy  tokens in the initial  configuration, the result follows.  Thus it remains to  prove:

\begin{lemma}
\label{lemma:every-swap-happy}
Suppose there is a happy  swap  of tokens $t$ and $s$, with $s$ closer to  (or at) the center.  Then the swap causes $s$ to  become happy.
\end{lemma}
\begin{proof}
Suppose $t$ and $s$ are at nodes $u$ and $v$, respectively and that $u$ is part of branch $B$.  Node $v$ is either part of branch $B$ or the center node. 
Suppose the destination token at  $u$ is $r_f$.  Note that $r \ne t$ since $t$ is not happy.  
We aim to  prove that $r=s$ which implies that $s$ and $s_f$ are at the same node after the swap.   

If $r_f$ wants to move away  from the  center, then $r$ must lie in branch  $B$, further from the center.  Then tokens $r$ and $t$ both had their initial positions in $B$.  Along the path  $Q_r=Q_t$ we find $r$, then $t,r_f$ at the  same node, then $t_f$, a contradiction to Claim~\ref{claim:pair-ordering}.

Thus $r_f$ wants to move towards the center. Suppose $r \ne s$.  Because $s$ and $r$ both have their destinations in this branch, they had a  common initial branch.
Along the path $Q_r=Q_s$, the ordering is $r,s,r_f,s_f$, a contradiction to 
 Claim~\ref{claim:pair-ordering}.
 Therefore $r=s$.
\end{proof}
This completes the proof of Theorem~\ref{thm:Vaughan-bound}.
\end{proof}

}

%% file: parallel.tex
\section{Parallel token swapping on trees is NP-complete}
\label{sec:parallel}

The parallel token swapping problem is like sequential token swapping, except that swaps that do not involve the same token can occur simultaneously ``in parallel''. In particular, time is measured in rounds where in each round a matching of adjacent tokens
is selected and all of those swaps then occur simultaneously during that round. Instead of asking how many swaps are required to bring the tokens to their target destinations, we now ask how many rounds are needed.

Precisely, in the \defn{parallel token swapping} problem, the input consists of a graph $G$, an integer $K$, and a permutation $\pi$ on the vertices of $G$. Place a unique token on every vertex of $G$. During a \defn{round}, we choose a matching in $G$ and then for every edge in the matching, \defn{swap} the tokens at the endpoints. In the \defn{parallel token swapping} problem, the goal is to determine whether it is possible to rearrange the tokens via a sequence of $K$ rounds such that for every vertex $v$ in $G$, the token starting in $v$ ends in $\pi(v)$.

In this section, we will prove that the parallel token swapping problem remains hard even when the input graph is restricted to be a particular type of tree: a subdivided star. A \defn{subdivided star} is a graph consisting of any number of paths all joined at a single endpoint which we will call the \defn{root}. See Figure~\ref{figure:star_example} for an example. We will prove this by reducing from Star STS. As in previous sections, since we are reducing from a token swapping problem to another token swapping problem, we will introduce a difference in terminology: as before, we will call the tokens from the Star STS instance \defn{items}, and we will call the leaves of the graph in the Star STS instance \defn{slots}. 

\begin{figure}[!htb]
  \centering
  \includegraphics[scale=.4]{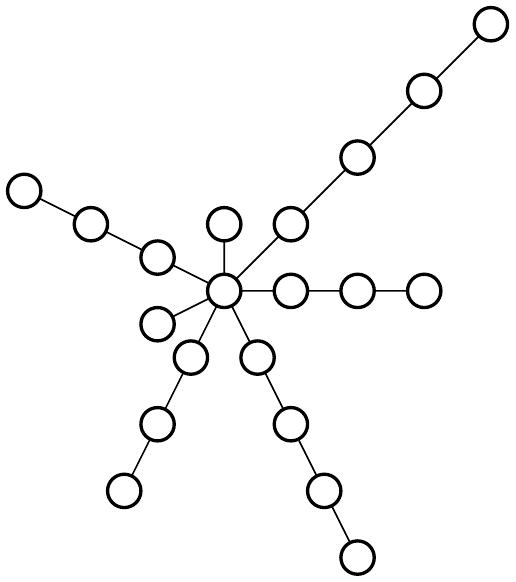}
  \caption{An example of a subdivided star}
  \label{figure:star_example}
\end{figure}

In Section~\ref{sec:parallel-intuition}, we will lay out the intuition behind the reduction. Then in Section~\ref{sec:parallel-reduction}, we will formally describe the reduction. Next, in Section~\ref{sec:parallel-forward} we will show how to solve the parallel token swapping instance produced by the reduction when the input Star STS instance was possible to solve. Finally, in Section~\ref{sec:parallel-backward} we will show the reverse direction: that any solution to a parallel token swapping instance produced by the reduction can be used to construct a solution to the input Star STS instance. Together, these last two sections combine to prove that the reduction is answer-preserving, allowing us to conclude:

\begin{theorem}
\label{thm:parallel-NP-hard}
Parallel token swapping on trees in NP-hard.
\end{theorem}

\subsection{Proof idea}
\label{sec:parallel-intuition}

The most basic concept of this reduction is the same as the prior reduction for Token Swapping on trees. We represent each item in the Star STS instance with a token in the parallel token swapping instance (an item token) and associate some of the branches of the subdivided star (slot branches) with the slots of the Star STS graph. If an item is located at the root of the star, then the corresponding item token will be located at the root of the subdivided star. If an item token is located in a slot, then the corresponding item token will be located somewhere in the corresponding slot branch.

In order to make that basic idea work, we need to be able to control which tokens can pass through the root and at which time. We want to restrict the possible transitions of an item token through the root rather than allow arbitrary traversal of the root. Using the root as a bottleneck for tokens swapping into their target branch will be a keystone idea in the proof. For this purpose, we introduce a second set of tokens, which we call \defn{enforcement tokens}. The enforcement tokens will enforce the behavior of the swaps near the root by creating ``congestion'' at the root. This will be achieved by placing the enforcement tokens a distance $K$ away from their targets, forcing the enforcement tokens to take the shortest path to their destination and thus needing to swap into the root at specific times.

Each enforcement token will start at some distance $d$ from the root and end at distance $K-d$ from the root in a different branch. The token's destination is a distance of $K$ from its starting location, so therefore the enforcement token will have to move one step towards its target location every round in order to reach its destination by the $K$ round deadline. In particular, during round $d$ the token will swap into the root from its starting branch and during round $d+1$ the token will swap out of the root into its destination branch.

Since the root vertex can only be involved in one swap on any given round, these swaps are going to be the only swaps involving the root during rounds $d$ and $d+1$. We want the root to have a swap involving an enforcement token at every round, so we can create one enforcement token, call it $e_1$, with $d=1$, another ($e_3$) with $d=3$, another ($e_5$) with $d=5$, and so on up to a final enforcement token, $e_{K-1}$, with $d=K-1$. This will make it so that there is a forced swap involving the root vertex during every single round (rounds $1$ and $2$ due to $e_1$, rounds $3$ and $4$ due to $e_3$, etc...).

So let's say we set this up exactly like that. When can an item token (or any other non-enforcement token) pass through the root? There is a swap involving an enforcement token and the root during every round, so a token can pass through the root only if it is swapping with an enforcement token. An enforcement token enters the root every odd numbered round and exits every even numbered round. Therefore a token could be able to enter the root on some round $2t$ while enforcement token $e_{2t-1}$ is exiting the root and then exit the root on round $2t+1$ while enforcement token $e_{2t+1}$ is entering the root. If this occurs, however, the token must start in the branch that is the destination branch of token $e_{2t-1}$ and end in the branch that is the source branch of token $e_{2t+1}$. Therefore, not only have we restricted the possible traversals through the root, we also have a way to fine tune exactly which branches we allow swaps between and when by setting the source and destination branches of the enforcement tokens.

Unfortunately, setting things up this way also gives us an additional constraint: for each branch, it cannot be the case that an enforcement token leaves that branch after another enforcement token enters the branch. The problem is that the two tokens, each on its own inexorable path towards its destination will end up colliding inside the branch. Two enforcement tokens can never swap with each other due to a parity argument (all enforcement tokens always have the same parity of distance from the root, as they all swap into the root on even numbered rounds), so there is no way for the two tokens to pass each other. In other words, if the constraint is violated then a solution to the parallel token swapping instance will be impossible.

The constraint that all enforcement tokens leave a branch before any other enforcement tokens enter is equivalent to the constraint that all non-enforcement tokens that are going to enter a branch do so before any non-enforcement tokens leave the branch. This seems like a problem. How can we use a branch to represent a slot if the branch has the constraint that item tokens never enter after other item tokens have exited? After all, in Star STS, slots are used and re-used repeatedly.

The solution to this problem is the second key idea of the reduction. Unlike in the sequential token swapping case, we will allow the correspondence between slots and branches to change over time. At any given time a particular branch will represent a slot, but which branch represents that slot can change. This makes it so that we don't have to keep returning tokens into the same branches again after tokens have already left the branch; instead, we just move those tokens into new branches and re-label the correspondence between branches and slots.

So how do we actually position the start and end branches of the enforcement tokens so that the desired swaps between item tokens are possible? We will use 8 rounds for every swap. If we are trying to swap some slot $s$, here's the effect of these 8 rounds on the non-enforcement tokens:

\begin{enumerate}
    \item the item token at the root moves into a branch $w$ dedicated to this swap
    \item the item token in the branch $b_s$ currently associated with slot $s$ moves into the root
    \item the item token now at the root moves into branch $w$; there are now two item tokens in branch $w$
    \item some non-enforcement token moves into the root 
    \item the token at the root moves into another branch (these two swaps provide an extra rounds in which the optional swap of the item tokens can occur)
    \item an item token from $w$ moves into the root
    \item the item token at the root moves into a new branch $b_s'$, which as the name suggests will become the new branch associated with slot $s$
    \item the remaining item token in $w$ moves into the root
\end{enumerate}

Of course, to achieve these effects we set the start and destination branches of the enforcement tokens, rather than setting the start and destination branches of the item tokens as described in the list, since the latter is not something we can directly control.

Notice that during rounds 4 and 5 in the above list, there are two item tokens in branch $w$: the item token that was at the root and the one which was in branch $b_s$ (corresponding with the item in slot $s$). During these rounds, the two tokens can (optionally) swap. Then the rest of the rounds move these two item tokens into their two new locations: branch $b_s'$ and the root. In other words, depending on which token ends up where, there may have been a swap.

Since we can simulate a swap, we can just use the target locations of the item tokens to encode the target permutation of the items. The only remaining question is what to do about all the other tokens we haven't introduced. For this we bring back the concept of the scaffold solution. We know that certain swaps, the swaps of the enforcement tokens, are forced. Therefore, we can consider what happens if we do only those swaps; that is the scaffold solution. Now consider what happens if we place a token on every vertex, and then go through the scaffold solution. The enforcement tokens will all reach their destinations. The item tokens might end up shuffled among themselves (relative to where they should end up), but as a group of tokens, the item tokens will end up occupying exactly the group's target locations. Then for all non-item non-enforcement tokens, we can just set the target destination to be whatever the destination ends up under the scaffold solution.

\subsection{Reduction}
\label{sec:parallel-reduction}

Suppose we have a Star STS instance consisting of: a star with center 0 and slots $1, \ldots, m$, each of which initially has an item of its same label; a permutation $\pi$ of the items; and a sequence $s_1, \ldots, s_n$ of slots that specify the allowed swaps.

We will construct a parallel token swapping instance consisting of: the round limit $K$, a subdivided star, and a permutation $\pi$ of the vertices indicating for every vertex $v$ the target location $\pi(v)$ for the token that starts at $v$. For the parallel token swapping construction we will generally use a superscript to refer to the index of the branch and a subscript to refer to a distance from the root. For example, we might discuss the $i^{th}$ branch $b^i$ or the vertex $b^i_j$ which is in branch $i$ and a distance of $j$ away from the root. 

\subsubsection{Time limit \texorpdfstring{$K$}{K} and the subdivided star}

As described in the proof idea section, each swap in the Star STS instance will be represented by $8$ rounds in the Parallel Token Swapping instance. Therefore, we set $K = 8n$.

For the subdivided star, we will use the root $r$, together with a total of $n+(m+n)+2$ branches, each of length $K$. For each branch $b$, we will refer to the vertices of the branch starting from the end closest to the root as $b_1, b_2, \ldots, b_K$. These are also shown in Figure~\ref{fig:ParallelTokenOverview}.

\begin{itemize}
    \item $n$ of these branches are \defn{swap branches} $w^t$ for $t=1,\ldots,n$. Each swap branch will be used for one of the swaps to temporarily hold the two swapping item tokens together (so that they can optionally swap). 
    \item $m+n$ of the branches are \defn{slot branches} $s^i$ with $i = 1, 2, \ldots, m+n$. Obviously there are more slot branches than slots, but this is because not every slot branch will represent a specific slot at any given time. Rather, at any given time, $m$ of the branches will be \defn{active}, and those will be the branches which currently represent a slot. Since a swap needs a new branch to occur, we need an additional $n$ branches some of which will only temporarily represent a specific slot. More details on how the mapping will change over time are given in the next section. 
    \item The final two branches are \defn{garbage branches} $g$ and $g'$. Tokens will move between these branches as a way of wasting time.
\end{itemize}

\begin{figure}
    \centering
    \includegraphics[scale=.5]{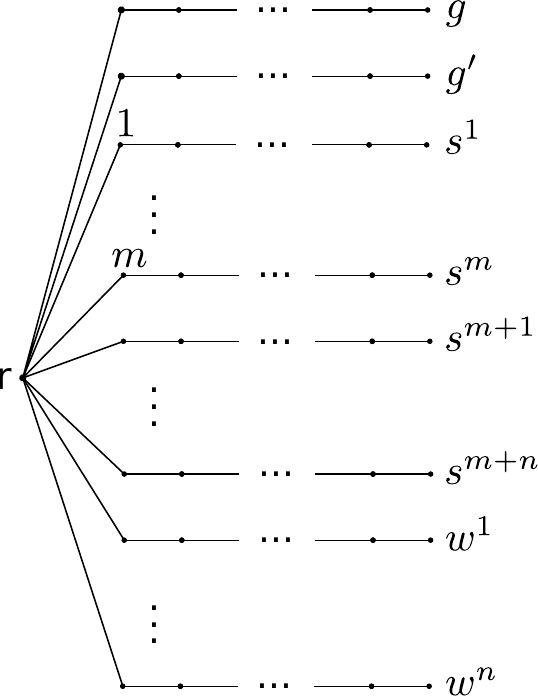}
    \caption{The overall structure of the Parallel Token Swapping instance. Each branch is labeled with its name on the right, and each vertex initially containing an item token is labeled with the name of that token.}
    \label{fig:ParallelTokenOverview}
\end{figure}

All that's left is to define the permutation $\pi$. We will do this a piece at a time by considering the three types of tokens described in the proof idea section: item tokens, enforcement tokens, and the remaining, non-item non-enforcement, tokens, which we will refer to as \defn{filler tokens}.

\subsubsection{Item tokens}

Before considering the item tokens, we first turn our attention to the active branches. We define a function $a$ that tells us which branches are active at any time. In particular, if $1 \le i \le m$ and $0 \le t \le n$ then $a(i, t)$ is defined such that $s^{a(i,t)}$ is the slot branch which represents slot $i$ at time $t$. We say that, at time $t$, $\{s^{a(i, t)} \mid 1\le i \le m\}$ are the active slot branches and the other slot branches are inactive. 

At time $t=0$, slot branches $1$ through $m$ will represent slots $1$ through $m$, and so $a(i, 0) = i$. Then for every $t > 0$, we define $a(\cdot, t)$ recursively. In particular, during swap $t$, the slot branch that used to represent slot $s_t$ becomes inactive and instead slot branch $m+t$ becomes active in order to represent that slot. In other words, $a(i, t) = a(i, t-1)$ for $i \ne s_t$ and $a(s_t, t) = m+t$.

With that done, we can define the initial and final positions of the item tokens. There are $m+1$ item tokens, which we refer to as $0, 1, \ldots, m$, named after their corresponding items.

Item $0$ starts at the root of the star, so we will initially place item token $0$ at the root $r$. Each other item $i$ starts in slot $i$, so we will initially place the corresponding item token $i$ in the active branch (at time $t=0$) representing slot $i$ next to the root. In other words, we place item token $i$ at vertex $s^{a(i, 0)}_1=s^i_1$.

For the final positions, consider the item which, according to $\pi$, should end up at the root of the star. We give this item token a final destination of $r$. For every other item $i$, the item is supposed to end up in slot $\pi(i)$. Therefore, we give the corresponding item token $i$ a final destination which is in the active branch (at time $t=n$) representing slot $\pi(i)$ next to the root. In other words, we give item token $i$ a final destination of $s^{a(\pi(i), n)}_1$.

Notice that since the initial branches assigned to each item token are disjoint (with the exception of one item token starting at the root, which is not assigned an initial branch at all), each item token is given a unique start location. Similarly, the final locations of the item token ends are in disjoint branches (with the exception of one item token at the root), so each item token is given a unique end location.

\subsubsection{Enforcement Tokens}

Now we will describe the \defn{enforcement tokens}. There will be $\frac{K}{2}$ enforcement tokens $e_1, e_3, \ldots, e_{K-1}$. Enforcement token $e_d$ will start in some branch $x^d$ at distance $d$ from the root and end in some other branch $y^d$ at distance $K-d$ from the root. In other words, the token will start at vertex $x^d_d$ and end at vertex $y^d_{K-d}$. All that's left to make these positions precise is to identify the start and end branches $x^d$ and $y^d$ for $d = 1, 3, \ldots, K-1$.

Here is how these branches are assigned:

For each $t$ with $1 \le t \le n$, we assign branches as follows
\begin{itemize}
    \item $x^{8t-7} = w^t$
    \item $y^{8t-7} = s^{a(s_t, t-1)}$
    \item $x^{8t-5} = w^t$
    \item $y^{8t-5} = g$
    \item $x^{8t-3} = g'$
    \item $y^{8t-3} = w^t$
    \item $x^{8t-1} = s^{a(s_t, t)}=s^{m+t}$
    \item $y^{8t-1} = w^t$
\end{itemize}

Note that these 8 branches in this order are exactly the 8 branches from the proof idea section list of token movements through the root corresponding with a swap.

It is simple to argue using the distance from the root that each enforcement token has a unique start location and a unique end location. Furthermore, we also wish to claim that no enforcement token shares a start location or an end location with an item token. Since item tokens start and end at distance $0$ or $1$ from the root, $e_1$ is the only enforcement token that could possibly share a start location with an item token and $e_{K-1}$ is the only enforcement token that could possibly share an end location with an item token. However, $e_1$ starts in branch $w^1$ and $e_{K-1}$ ends in branch $w^n$, neither of which are slot branches. No item token starts or ends in a non-slot branch, so we can conclude that the item and enforcement token start locations are all unique, as are their end locations.

\subsubsection{Filler Tokens}
\label{sec:Filler Tokens}

For every vertex that does not initially contain an enforcement token or an item token, we place a filler token initially at that vertex. In order to identify the target destination for these tokens, we make use of the scaffold solution.

Define the \defn{scaffold solution} to be a set of swaps for times $1 \le t \le K$. At every time $t$, we will include one swap for each enforcement token: the swap that brings that token one step towards its destination. \footnote{Note, scaffold solution only defines a subset of swaps which must occur based on the enforcement token. It is not based on an `intended solution' like the scaffold solution for the single swap reduction. Thus, although both set up a structure of moves which must be followed, they have somewhat different intuition and technical differences which should not be confused with each other. }

We will later prove the following lemma:

\begin{lemma}
\label{lemma:parallel-scaffold-legal}
The scaffold solution does not include any pair of simultaneous swaps that share a vertex.
\end{lemma}

Using this lemma and the fact that no two item or enforcement tokens share a start position, we can unambiguously execute the scaffold solution starting with the initial position of all the tokens. (Note that if we tried to execute a timeline of swaps including simultaneous swaps sharing a vertex, we would have an ambiguity; where would the token at that vertex go?). Define the target destination of each filler token to be the destination of that token when we execute the scaffold solution.

We showed previously that every vertex is the start position of at most one item or enforcement token. Since we started a filler token in every vertex not starting with an item or enforcement token, we now have that every token has a unique starting position and that every vertex is the starting location of some token.

What about ending locations? We saw previously that every item and enforcement token has a unique end location. In order to consider the end locations of the filler tokens we have to examine the behavior of the tokens under the scaffold solution. First of all, every enforcement token is clearly moved to its target location by the scaffold solution. This is because the moves included in the scaffold solution are exactly the ones that are necessary in order to move the enforcement tokens to their target destinations. For the item tokens, we use the following lemma:

\begin{lemma}
\label{lemma:parallel-scaffold-item}
The scaffold solution brings each item token to a final location that is the target location of some item token (though not necessarily the token's own target destination).
\end{lemma}

Then we can ask the question of where the scaffold solution brings the filler tokens. The enforcement tokens each end up in their own target locations. The $m+1$ item tokens each end up in a target location of an item token, and since there are only $m+1$ such locations (and the scaffold solution certainly doesn't bring two item tokens into the same location), we can conclude that the $m+1$ item tokens end up in the $m+1$ item token target locations. Therefore the end locations of the filler tokens under the scaffold solution are exactly those vertices that are not target locations of item or enforcement tokens. Since the target location of a filler token is defined to be its end location under the scaffold solution, we can conclude that every token is assigned a unique target location.

\subsubsection{Permutation \texorpdfstring{$\pi$}{σ}}

Since we have defined the start and end positions of the three types of tokens, we can define the permutation $\pi$. For every vertex $v$, let $\pi(v)$ be the end location of the token that starts at $v$. Since every vertex has exactly one token starting there, this is well defined. Since no two tokens share an end location, this is a permutation. 

\subsubsection{Proof of Lemmas~\ref{lemma:parallel-scaffold-legal} and \ref{lemma:parallel-scaffold-item}}

In Section~\ref{sec:Filler Tokens}, we left out the proofs of two lemmas. Provided that those two lemmas hold, the reduction is clearly well defined and furthermore can be computed in polynomial time. Thus, all that's left in order to conclude NP-hardness of parallel token swapping is to prove the lemmas, which we do here, and prove the answer-preserving property of the reduction, which we do in the following sections.

\begin{proof}[Proof of Lemma~\ref{lemma:parallel-scaffold-legal}]
\emph{The scaffold solution does not include any pair of simultaneous swaps that share a vertex.}

Every swap in the scaffold solution includes an enforcement token. All enforcement tokens always share the same parity of distance from $r$. Therefore, the only two ways for the scaffold solution to include a pair of simultaneous swaps that share a vertex is (1) two enforcement tokens both swap into the same vertex, or (2) two enforcement tokens at the same vertex both swap out of that vertex. In either case, there must be two enforcement tokens in the same place in order for the scaffold solution to have a pair of simultaneous swaps sharing a vertex.

Clearly no pair of enforcement tokens is ever at $r$ at the same time (since each $e_d$ is at $r$ immediately after round $d$). Therefore, in order to be in the same place at the same time, the enforcement tokens would have to meet inside some branch. But then, some enforcement token would have to enter a branch before a different enforcement token exits that branch. Note that this is the condition that the proof idea section explicitly said we would need to avoid. 

Checking that no branch has an enforcement token enter before another one leaves is tedious but not difficult. We now consider, in turn, each type of branch: the swap branch, the garbage branch, and the slot branch.

For each \textbf{swap branch} $w^t$ only four enforcement tokens ever enter or leave. The two that leave are $e_{8t-7}$, which exits during round $8t-7$, and $e_{8t-5}$, which exits during round $8t-5$. The two that enter are $e-{8t-3}$, which enters during round $8t-2$, and $e_{8t-1}$, which enters during round $8t$. Clearly, enforcement tokens never enter before another enforcement token leaves.

The \textbf{garbage branch} $g$ only ever has enforcement tokens enter it, while the other garbage branch $g'$ only ever has enforcement tokens exit it.

Each \textbf{slot branch} $s^i$ has at most one enforcement token enter it. This is because the only enforcement tokens to enter a slot branch are the ones of the form $e_{8t-7}$, which enter branch $y^{8t-7} = s^{a(s_t,t-1)}$ during round $8t-6$. For every $t$, $a(s_t,t-1)$ is a unique value (since by the definition of $a$, once $a(s_t,t-1)$ is replaced by $a(s_t,t)=m+t$, the value $a(s_t,t-1)$ can never occur again as an output of $a$). Therefore each such enforcement token enters exactly one such slot branch.

Similarly, each slot branch $s^i$ has at most one enforcement token exit it. This is because the only enforcement tokens to exit a slot branch are the ones of the form $e_{8t-1}$, which exit branch $x^{8t-1} = s^{a(s_t,t)}=s^{m+t}$ during round $8t-1$. 

In order for an enforcement token to enter a slot branch before a different enforcement token exits the branch, it would have to be the case that $i = a(s_t,t-1)$ and $i = m+t'$ with $8t-6 < 8t'-1$. This is equivalent to $t \le t'$. But notice that the maximum value of $a(\cdot, \tau)$ is always $m+\tau$. Therefore $i = a(s_t,t-1) \le m + t - 1 \le m+t' - 1 = i - 1$. Clearly it cannot be that $i \le i-1$, and so we conclude that an enforcement token enters a slot branch $s^i$ only after other  enforcement tokens have exited the branch.

This concludes the casework, proving that two enforcement tokens never collide in the scaffold solution and therefore that the scaffold solution never includes two simultaneous swaps sharing a vertex.
\end{proof}

\begin{proof}[Proof of Lemma~\ref{lemma:parallel-scaffold-item}]
\emph{The scaffold solution brings each item token to a final location that is the target location of some item token (though not necessarily the token's own target destination).}

In order to prove this, we will prove by induction that for $0 \le t \le n$, after round $8t$, the item tokens are located at $r$ and at vertices $s^{a(i, t)}_1$ for $1 \le i \le m$.

By definition this is true for $t = 0$:

After round $0$ (that is, right at the beginning), the item tokens were defined to be at locations $r$ and $s^{a(i,0)}_1$ for $1 \le i \le m$, exactly where we wanted to show they would be.

Next suppose that for some $1\le t \le n$, it was the case that after round $8(t-1)$, the item tokens are located at $r$ and at vertices $s^{a(i, t-1)}_1$ for $1 \le i \le m$. Let's see what happens in the scaffold solution over the course of rounds $8t-7$ through $8t$. This is depicted in Figure~\ref{fig:ParallelSwap}. The item tokens are all located within distance $1$ of $r$, and so the only enforcement tokens that can move the item tokens are the ones that get close to $r$ during these 8 rounds. In particular, it is easy to verify that the only relevant enforcement tokens are $e_{8t-7}$, $e_{8t-5}$, $e_{8t-3}$, and $e_{8t-1}$.

Token $e_{8t-7}$ moves from $w^t_1$ to $s^{a(s_t, t-1)}_7$ over these 8 rounds. $e_{8t-5}$ moves from $w^t_3$ to $g_5$. $e_{8t-3}$ moves from $g'_5$ to $w^t_3$. $e_{8t-1}$ moves from $s^{a(s_t, t)}_7$ to $w^t_1$.

\begin{figure}
    \centering
    \includegraphics[width=.8\linewidth]{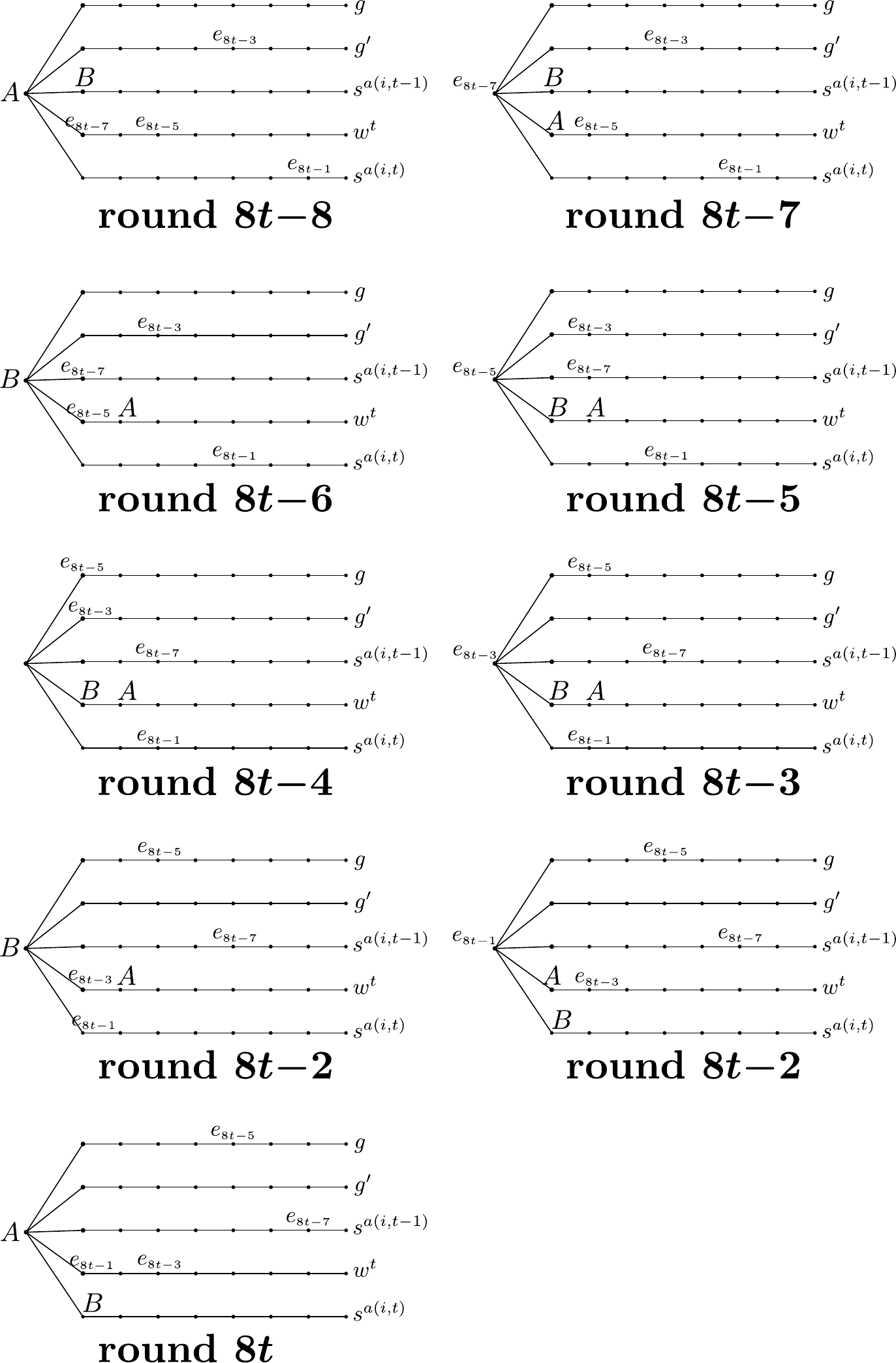}
    \caption{Position of item and enforcement tokens during rounds $8t-7$ to $8t$. Branches as well as vertices containing enforcement and item tokens are labeled.}
    \label{fig:ParallelSwap}
\end{figure}

Here's the effect of these enforcement tokens on the item tokens. On round $8t-7$, token $e_{8t-7}$ moves the item token at $r$, call it $A$, into $w^t_1$. Next, on round $8t-6$, token $e_{8t-5}$ moves that same item token, $A$, into $w^t_2$, while simultaneously $e_{8t-7}$ moves the item token at $s^{a(s_t, t-1)}_1$, call it $B$, into $r$. On round $8t-5$, token $e_{8t-5}$ moves $B$ into $w^t_1$. During rounds $8t-4$ and $8t-3$, no item tokens move. On round $8t-2$, token $e_{8t-3}$ moves into $w^t_1$ from $r$, and as a result $B$ moves back into $r$. Next, on round $8t-1$, token $e_{8t-3}$ moves $A$ into $w^t_1$, while simultaneously token $e_{8t-1}$ moves $B$ from $r$ into $s^{a(s_t,t)}_1$. Finally, on round $8t$, token $e_{8t-1}$ moves $A$ into $r$.

All other movements of non-enforcement tokens due to the scaffold solution are not movements of item tokens. The result on item tokens is that $A$, which started at $r$, is moved back into $r$, while $B$, which started at $s^{a(s_t, t-1)}_1$, is moved into $s^{a(s_t, t)}_1$.

Note that $s^{a(i, t-1)}_1 = s^{a(i,t)}_1$ whenever $i \ne s_t$. Therefore, all the other item tokens, which started at $s^{a(i, t-1)}_1$ with $i \ne s_t$, end up at $s^{a(i,t)}_1$. And as we saw, the item token ($A$) that was previously at $r$ stays at $r$ and the item token ($B$) that was previously at $s^{a(i, t-1)}_1$ with $i = s_t$ ends up at $s^{a(i, t)}_1$. Thus, as desired, after round $8t$, the item tokens are located at $r$ and at vertices $s^{a(i, t)}_1$ for $1 \le i \le m$.

By induction, this same statement holds for every $t$ with $0 \le t \le n$. In particular, applying it to $t = n$, we see that the item tokens are located at $r$ and the vertices $s^{a(i, n)}_1$ for $1 \le i \le m$ after $8n$ rounds of the scaffold solution (i.e. at the end of it). But those are exactly the target locations of item tokens, which is what we wanted to prove.
\end{proof}

\subsection{Star STS solution \texorpdfstring{$\to$}{-->} parallel token swapping solution}
\label{sec:parallel-forward}

Suppose we have a solution to a Star STS instance. Then consider the parallel token swapping instance produced by the reduction from the previous section.

Let $S$ be the scaffold solution. We will build a different solution $S'$ by adding some swaps to $S$. Consider the solution to the Star STS instance. For $1\le t \le n$, if swap $t$ is used in the Star STS solution, then add a swap of vertices $w^t_1$ and $w^t_2$ during round $8t-4$ to $S'$.

First of all, it is easy to verify that no enforcement token swaps into or out of either $w^t_1$ or $w^t_2$ during round $8t-4$. Therefore, $S'$ is a valid schedule of swaps for the parallel token swapping instance. In fact, we will show that $S'$ solves the instance.

The main idea of this proof is that $S'$ maintains the following invariant:

\begin{lemma}
\label{lemma:parallel-forward-lemma-1}
If $0 \le t \le n$ then after round $8t$ of $S$, one of the item tokens is located at $r$ and the other $m$ item tokens are in the active slot branches $s^{a(\cdot, t)}$, each in the vertex of that branch which is adjacent to the root. In particular, consider the state of the Star STS instance after the first $t$ swaps are each either used or not according to the solution. The item token located at $r$ will correspond to the item at the root, while the item token in active slot branch $s^{a(i, t)}$ will correspond to the item in slot $i$.
\end{lemma}

In the process of proving this, we will also show the following:

\begin{lemma}
\label{lemma:parallel-forward-lemma-2}
During round $8t-4$ of $S'$, the two tokens in vertices $w^t_1$ and $w^t_2$ will always be item tokens. 
\end{lemma}

We can use these two lemmas to prove that $S'$ brings every token to its target destination (within $K$ rounds), and therefore that $S'$ is a solution to the parallel token swapping instance.

Consider first any token that is not an item token. We know that the scaffold solution $S$ brings every such token to its target location. But $S'$ is just $S$ with some swaps added. In particular, by Lemma~\ref{lemma:parallel-forward-lemma-2}, we know that every swap that was added is between two item tokens. Therefore, the result of $S'$ will be the same as the result of $S$, but with the item tokens permuted in some way. In particular, just as $S$ brings the non-item tokens to their target locations, $S'$ does the same.

Now consider the item tokens. Let $i$ be an item token. There are two cases. Either the final position of the corresponding item $i$ is $\pi(i) = 0$ or the final position of item $i$ is some slot $\pi(i)$. 

In the first case, item $i$ ends up at the root ($\pi(i) = 0$). Therefore, in the Star STS solution, item $i$ is at the root after the entire solution. Phrased a little more verbosely, we have that after the first $t=n$ swaps are each either used or not according to the Star STS solution, item $i$ is at the root.  But then by Lemma~\ref{lemma:parallel-forward-lemma-1} applied to $t = n$, we know that the item token at $r$ after $8t = 8n = K$ rounds of $S'$ (the entirety of $S'$) is exactly the item token $i$. But when we were defining the target location of an item token, we gave the item token $i$ for which $\pi(i)=0$ the vertex $r$ as its target location. Therefore in this case the item token is brought to its target location by $S'$. 

In the second case, the final position of item $i$ is some slot $\pi(i)$. Then in the Star STS solution, item $i$ is in slot $\pi(i)$ after the entire solution, or equivalently after the first $t=n$ swaps are each either used or not according to the Star STS solution. Then by Lemma~\ref{lemma:parallel-forward-lemma-1} applied to $t = n$, we know that the item token at $s^{a(\pi(i),t)} = s^{a(\pi(i),n)}$ after $8t = 8n = K$ rounds of $S'$ (the entirety of $S'$) is exactly the item token $i$. If we look back at the target location that was assigned to an item token $i$ with $\pi(i) \ne 0$, the target location was $s^{a(\pi(i), n)}$. Thus, in this case also we see that the item token is brought to its target location by $S'$.

Since $S'$ consists of $K$ legal rounds of swaps which bring every token to its target location, we can conclude that as desired, $S'$ solves the parallel token swapping instance.

\subsubsection{Proof of Lemma~\ref{lemma:parallel-forward-lemma-1} and~Lemma~\ref{lemma:parallel-forward-lemma-2}}

\begin{proof}[Proof of Lemma~\ref{lemma:parallel-forward-lemma-1}]
\emph{
If $0 \le t \le n$ then after round $8t$ of $S$, one of the item tokens is located at $r$ and the other $m$ item tokens are in the active slot branches $s^{a(\cdot, t)}$, each in the vertex of that branch which is adjacent to the root. In particular, consider the state of the Star STS instance after the first $t$ swaps are each either used or not according to the solution. The item token located at $r$ will correspond to the item at the root, while the item token in active slot branch $s^{a(i, t)}$ will correspond to the item in slot $i$.
}

This Lemma can be proved by an inductive argument almost identical to the proof of Lemma~\ref{lemma:parallel-scaffold-item}. The only difference is that instead of tracking the positions of the item tokens as a group, we track the position of each item token individually. Also, due to the swap that may have been added, the two item tokens that were called $A$ and $B$ in the proof may end up swapped (this is because the swap that might have been added to $S$ to get $S'$ depending on the Star STS solution, is a swap of two vertices which contain $A$ and $B$ at the time). Figure~\ref{fig:ParallelSwap} remains an accurate depiction, except that $A$ and $B$ will be swapped after round $8t-4$.

In this new induction, if the Star STS solution does not involve a swap, then the analysis is exactly the same as in the previous proof: $A$ starts at $r$ and ends at $r$, $B$ starts at $s^{a(s_t, t-1)}_1$, and ends at $s^{a(s_t, t)}_1$, and every other item token starts at $s^{a(i, t-1)}_1$ and ends at $s^{a(i, t)}_1$ (which happens to be the same vertex). Under the correspondence with the Star STS instance, this corresponds with none of the items moving.

If, on the other hand, the Star STS solution involves a swap, then as previously mentioned, the end result is that item tokens $A$ and $B$ are swapped. $A$ and $B$ are the tokens which were previously at $r$ and $s^{a(s_t, t-1)}_1$, which means that they corresponded to the items previously at the root and in slot $s_t$. Then the result of these 8 rounds is that the two item tokens are swapped relative to their positions in the previous case, which exactly corresponds to the fact that during step $t$ of the Star STS solution, the item in slot $s_t$ and the item at the root swap.
\end{proof}

\begin{proof}[Proof of Lemma~\ref{lemma:parallel-forward-lemma-2}]
\emph{
During round $8t-4$ of $S'$, the two tokens in vertices $w^t_1$ and $w^t_2$ will always be item tokens.}

This follows directly from the previous proof, where we saw that the two tokens in these positions at that time are the item tokens referred to as $A$ and $B$.
\end{proof}

\subsection{Parallel Token Swapping solution \texorpdfstring{$\to$}{-->} Star STS solution}
\label{sec:parallel-backward}

Suppose we started with a Star STS instance and used the reduction to construct the corresponding parallel token swapping instance, and that this instance had a solution $S'$.

We will show that $S'$ can be used to construct a solution to the Star STS instance, which together with the results of the previous section shows that the reduction is answer-preserving.

The key idea will be to consider the location of each item token after rounds $0,8,16,\ldots, 8n=K$ of $S'$. In particular, we will care only about whether the token is at the root or otherwise which branch the token is in, not where in a branch the item token can be found.

We know that each enforcement token must move a distance of $K$ from its start location to its target destination. Therefore, the token must be swapped towards its destination during every round of $S'$. Thus every swap from the scaffold solution $S$ must be present in $S'$.

We wish to know how the branches of the item tokens change over time. But a token switches branches only by passing through $r$. Therefore, we can consider which swaps of $S'$ involve $r$. In fact, the enforcement tokens were specifically chosen so that $S$ would include a swap involving $r$ on every round. Therefore, the only swaps involving $r$ in $S'$ are the swaps involving $r$ from $S$. What are those swaps?

Well if $d$ is odd, then on round $d$, enforcement token $e_d$ swaps into $r$ from vertex $x^d_1$, and on round $d+1$, enforcement token $e_d$ swaps out of $r$ and into vertex $y^d_1$. 

Phrased differently, if $d$ is odd, then on round $d$, a non-enforcement token moves from $r$ into branch $x^d$, and on round $d+1$, a non-enforcement token moves from branch $y^d$ into $r$.

Then let us consider the 8 such movements starting after round $8t-8$ and ending after round $8t$:

\begin{itemize}
    \item On round $8t-7$, the token at $r$ moves into branch $x^{8t-7}=w^t$
    \item On round $8t-6$, some non-enforcement token in branch $y^{8t-7}=s^{a(s_t, t-1)}$ moves into $r$
    \item On round $8t-5$, the token at $r$ moves into branch $x^{8t-5}=w^t$
    \item On round $8t-4$, some non-enforcement token in branch $y^{8t-5}=g$ moves into $r$
    \item On round $8t-3$, the token at $r$ moves into branch $x^{8t-3}=g'$
    \item On round $8t-2$, some non-enforcement token in branch $y^{8t-3}=w^t$ moves into $r$
    \item On round $8t-1$, the token at $r$ moves into branch $x^{8t-1}=s^{a(s_t, t)}$
    \item On round $8t$, some non-enforcement token in branch $y^{8t-1}=w^t$ moves into $r$
\end{itemize}

As before, these correspond to swaps depicted in Figure~\ref{fig:ParallelSwap}. Some of these bullet points can be combined together. Overall, the effect of these 8 rounds (starting after round $8t-8$ and ending after round $8t$) on the branches of non-enforcement tokens is the following:

\begin{enumerate}
    \item the token at $r$ moves into branch $w^t$
    \item some non-enforcement token in branch $s^{a(s_t, t-1)}$ moves into branch $w^t$
    \item some non-enforcement token in branch $g$ moves into branch $g'$
    \item some non-enforcement token in branch $w^t$ moves into branch $s^{a(s_t, t)}$
    \item some non-enforcement token in branch $w^t$ moves into $r$
\end{enumerate}

In fact, we will show that most of these movements are movements of not just non-enforcement tokens, but more specifically item tokens. 

\begin{lemma}
\label{lemma:parallel-backward-lemma}
Every movement of non-enforcement tokens between branches in $S'$ (as described by the above list) is either a movement from branch $g$ into branch $g'$ or a movement of an item token.
\end{lemma}

Assuming that can be proven, we can then show that the movement of item tokens between branches in $S'$ corresponds to the movement of items between slots in the Star STS instance. In particular, after round $8t$ for $0 \le t \le n$, there will be one item token at $r$, and one item token in slot branch $s^{a(i, t)}$ for each $i \in \{1, \ldots, m\}$. Furthermore, the change between round $8(t-1)$ and round $8t$ will always correspond to either using or not using the swap $s_t$.

Certainly, for $t=0$, this holds: there is one item token, $0$, at $r$, and each item token $i$ for $i\in \{1, \ldots, m\}$ is in slot branch $s^{a(i, 0)} = s^i$. In the Star STS instance, item $0$ is at the root, while each other item $i$ is in slot $i$.

Next, suppose by the inductive hypothesis that for some $1 \le t \le n$, the above statement holds with $t-1$: after round $8(t-1)$ of $S'$, there was one item token at $r$, and one item token in slot branch $s^{a(i, t)}$ for each $i \in \{1, \ldots, m\}$. Let's see what movements of item tokens happen between rounds $8(t-1)$ and $8t$. According to the list above, (1) the item token at $r$ moves into branch $w^t$, (2)
the item token in branch $s^{a(s_t, t-1)}$ moves into branch $w^t$, (3) some non-enforcement token in branch $g$ moves into branch $g'$, (4) an item token in branch $w^t$ moves into branch $s^{a(s_t, t)}$, and (5) the final item token in branch $w^t$ moves into $r$. Note that since there are no item tokens in $g$, the token moving in (3) must be a filler token. Thus, the overall movement of item tokens is this: first the item tokens at $r$ and in branch $s^{a(s_t, t-1)}$ are moved into $w^t$, then one of those tokens moves into branch $s^{a(s_t, t)}$, and the other moves to $r$. 

Every other item token, in branch $s^{a(i, t-1)}$ with $i \ne s_t$, stays in that branch. But by the definition of $a$, $a(i, t) = a(i, t-1)$ precisely when $i \ne s_t$.
So the item token from branch $s^{a(i, t-1)}$ with $i \ne s_t$ ended up in branch $s^{a(i, t)}$, while the two item tokens at $r$ and in branch $s^{a(s_t, t-1)}$ ended up at $r$ and in branch $s^{a(s_t, t)}$, in some order.

Then the inductive hypothesis holds: after round $8t$ of $S'$, there is one item token at $r$, and one item token in slot branch $s^{a(i, t)}$ for each $i \in \{1, \ldots, m\}$. Furthermore, as we wanted to prove, the change between round $8(t-1)$ and round $8t$ always corresponds to either using or not using the swap $s_t$, depending on the arrangement of the two item tokens that moved from vertex $r$ and branch $s^{a(s_t, t-1)}$ to vertex $r$ and branch $s^{a(s_t, t)}$.

Thus, we can simply read off the movement of the items in a solution to the Star STS instance from the movement of the item tokens in $S'$. Now all that remains is to prove the needed lemma.

\begin{proof}[Proof of Lemma~\ref{lemma:parallel-backward-lemma}]
\emph{
Every movement of non-enforcement tokens between branches in $S'$ (as described by the above list) is either a movement from branch $g$ into branch $g'$ or a movement of an item token.}

The list is reproduced below:

\begin{enumerate}
    \item the token at $r$ moves into branch $w^t$
    \item some non-enforcement token in branch $s^{a(s_t, t-1)}$ moves into branch $w^t$
    \item some non-enforcement token in branch $g$ moves into branch $g'$
    \item some non-enforcement token in branch $w^t$ moves into branch $s^{a(s_t, t)}$
    \item some non-enforcement token in branch $w^t$ moves into $r$
\end{enumerate}

For the moment, assume the following three statements which will be argued at the end of the proof:

\begin{enumerate}
    \item For every slot branch $s^i$, there is exactly one way that an item token can get into the branch; either the slot token starts there or there is one movement according to the list which can move a non-enforcement token into the branch.
    \item Similarly, for every slot branch $s^i$, either there is an item token that has that branch as a target location, or there is exactly one movement according to the list which can move a non-enforcement token out of the branch, but not both.
    \item Finally, every swap branch $w^t$ has two movements on the list that can move a non-enforcement token into $w^t$, followed by two movements that can move a non-enforcement token out of $w^t$.
\end{enumerate}

We can combine these facts into a proof by induction. Consider all of the movements from the list other than those from $g$ to $g'$ in order according to $S'$. We will prove that all of these movements are movements of an item token.

Suppose this is true up to a certain movement, call it $M$ in the list. 

Case 1: $M$ moves a non-enforcement token from a slot branch $s^i$ into some other branch. By applying fact 2 to $s^i$, we know that $M$ is the only time that a non-enforcement token can leave branch $s^i$ and also that no item token has a target location in $s^i$. We can also apply fact 1 from the list above to $s^i$, yielding two cases.

Case 1a: $s^i$ started with an item token in it. Since no item token has a target location in $s^i$ and $M$ is the only opportunity for the item token to leave $s^i$, the token must take that opportunity in order to be able to reach its target location. Thus, $M$ must have been used to move an item token.

Case 1b: there is one movement $M'$ somewhere in the list which can move a non-enforcement token into $s^i$. $M'$ must have occurred before $M$. This is because $M'$ corresponds to an enforcement token leaving $s^i$, $M$ corresponds to an enforcement token entering $s^i$, and as we saw in Lemma~\ref{lemma:parallel-scaffold-legal}, an enforcement token never enters a branch before another enforcement token leaves that same branch. Since $M'$ occurred before $M$, the inductive hypothesis applies: $M'$ must have moved an item token. But then there is an item token currently in $s^i$. As in case 1a, $M$ is the only opportunity for that item token to leave $s^i$, and so the token must take that opportunity in order to be able to reach its target location. Thus, $M$ must have been used to move an item token.

Case 2: $M$ moves a non-enforcement token from a swap branch $w^t$ into another branch. Applying fact 3 from above, we see that $M$ must be one of the two movements $M_{out}^1$ or $M_{out}^2$ out of $w^t$ and must be preceded by exactly two movements $M_{in}^1$ and $M_{in}^2$ into $w^t$. The inductive hypothesis applies to both $M_{in}^1$ and $M_{in}^2$ since they precede $M$. Therefore, each of them moved an item token into $w^t$. As a result, after $M_{in}^1$ and $M_{in}^2$, there were two item tokens in $w^t$, neither of which has a target location in $w^t$. Since there are only two opportunities for an item token to leave $w^t$ ($M_{out}^1$ and $M_{out}^2$), it must be the case that the two item tokens make use of those two opportunities to leave the swap branch, since otherwise they would be unable to reach their target locations. Thus, both $M_{out}^1$ and $M_{out}^2$ must move an item token, and so $M$, which is one of the two, must have been used to move an item token.

Since in all cases we conclude that $M$ moved an item token, the proof by induction goes through; every movement on the list other than those from $g$ to $g'$ must move an item token.

All that's left is to confirm the three facts listed at the top of this proof.

Proof of Fact 1: The only movements on the list that are directed into a slot branch $s^i$ are the movements from $w^t$ to $s^{a(s_t, t)}$ for $1\le t \le n$. There are two cases. If $i \in \{a(j, 0) \mid 1 \le j \le m\}=\{1, \ldots, m\}$, then $1 \le i \le m$. On the other hand $a(s_t, t) = m+t > m$, and so in this case the branch $s^{a(s_t, t)}$ will never equal $s^i$. We also see that in this case $s^i$ is exactly one of the slot branches that initially contains an item token. On the other hand, if $i \not\in \{a(j, 0) \mid 1 \le j \le m\}$ then $i > m$, and so $a(s_{i-m}, {i-m}) = m + (i-m) = i$. Thus, exactly one of the movements from $w^t$ to $s^{a(s_t, t)}$ (in particular the one with $t = i-m$) will move a non-enforcement token into $s^i$. Notice, however, that in this case $s^i$ does not start with an item token. Thus we have shown exactly fact 1: for every slot branch $s^i$, there is exactly one way that an item token can get into the branch; either the slot token starts there or there is one movement according to the list which can move a non-enforcement token into the branch.

Proof of Fact 2: The only movements on the list that are directed out of a slot branch $s^i$ are the ones from $s^{a(s_t, t-1)}$ to $w^t$. Recall the recursive definition of $a$: $a(\cdot, 0)$ takes on the values $1$ through $m$, $a(\cdot, 1)$ replaces one of those values, $a(s_1,0)$ with $a(s_1, 1) = m+1$, then $a(\cdot, 2)$ replaces one of the remaining values, $a(s_2, 1)$, with $a(s_2, 2) = m+2$, etc... If $i \in \{a(j, n) \mid 1 \le j \le m\}$, then the value $i$ was never replaced out of the range of $a(\cdot, t)$ all the way until $t = n$. Therefore, in this case we see that $a(s_t, t-1)$ was not equal to $i$ for any $t$, and so none of the movements from $s^{a(s_t, t-1)}$ to $w^t$ are directed out of $s^i$. Note, however, that since $i \in \{a(j, n) \mid 1 \le j \le m\}$, we have that $s^i = s^{a(j, n)}$ for some $j$, and therefore that $s^i$ contains the target destination of some item token. On the other hand, if $i \not\in \{a(j, n) \mid 1 \le j \le m\}$, then the value $i$ was replaced out of the range of $a(\cdot, t)$ at some $t$. Therefore, in this case we see that $a(s_t, t-1)$ was equal to $i$ for exactly one $t$, and so exactly one of the movements from $s^{a(s_t, t-1)}$ to $w^t$ must have been directed out of $s^i$. In other words, we have demonstrated fact 2: for every slot branch $s^i$, either there is an item token that has that branch as a target location, or there is exactly one movement according to the list which can move a non-enforcement token out of the branch, but not both.

Proof of Fact 3: Consider any swap branch $w^t$, the only occurrences of that branch are for the instance of the list of movements which corresponds with this specific value of $t$. In this instance of the list, there are four movements involving $w^t$. There are two movements on the list that can move a non-enforcement token into $w^t$, followed by two movements that can move a non-enforcement token out of $w^t$.
\end{proof}

%% file: conclusion.tex
\section{Open Problems}

Many interesting related problems remain open.
For sequential token swapping on trees, where is the divide between NP-complete and polynomial-time?  
Our reduction's tree is 
a subdivided star
(as
for parallel token swapping)
with the addition of one extra leaf (the nook) per path.
By contrast, there is a polynomial-time algorithm for
the case of a broom (a star with only one edge subdivided)
\cite{vaughan1999broom,kawahara2019time,tree-token-swapping}. What about a star with two subdivided edges?

\iffull
\fi
For parallel token swapping, even the case of a single long path is open.
Kawahara et al.~\cite{kawahara2019time} gave an additive approximation
algorithm that uses at most one extra round.  Is there an optimal algorithm,
or is parallel token swapping NP-hard on paths?

There are also open problems in approximation
algorithms.  In parallel token swapping, we know that there is no PTAS~\cite{kawahara2019time}. 
Is there an $O(1)$-approximation for trees or general graphs? For sequential token swapping, 
there is a 4-approximation algorithm~\cite{miltzow2016approximation}. 
Is 4 a lower bound on the approximation factor of this algorithm?
Is $4$-approximation the best possible for general graphs?  

Although this paper focused on the best reconfiguration sequence, much research is devoted to understanding the worst-case behavior for a given graph; is it NP-hard to determine the diameter of the Cayley graph, i.e., the maximum number of reconfiguration steps that can be
required for any pair of token configurations?  This problem is open for
both sequential token swapping (implicit in~\cite{chitturi2019sorting})
and parallel token swapping~\cite{alon1994routing}.

%% file: appendix.tex
\appendix

\section{Star Subsequence Token Swapping Reachability (Star STS) is NP-complete}
\label{appendixA}

Our starting point is the NP-complete problem called 
``permutation generation'' in Garey and Johnson~\cite[MS6]{Garey-Johnson}, and called 
WPPSG (Word Problem for Products of Symmetric Groups) in~\cite{garey1980complexity}:

\begin{quote}
\textbf{WPPSG.} Given sets $X_1, X_2,  \ldots, X_n$, where $X_i \subseteq \{1, 2, \ldots , m\}$, and a permutation 
$\pi$ of  $\{1, \ldots, m\}$,
can $\pi$ be written as $\pi = \pi_1 \cdot \pi_2 \cdots \pi_n$ where
$\pi_i$ is a permutation of $X_i$ for $1 \le i \le n$?
\end{quote}

As noted  in~\cite[Corollary, p.~224]{garey1980complexity} WPPSG remains NP-complete when each subset $X_i$ has size 2, i.e.,~the  input is a sequence of swaps (transpositions) and the given permutation must be expressed as a subsequence of  the given sequence of  swaps.
This  special case is called WPPSG2 in~\cite{garey1980complexity}, and called ``Swap or not Reachability'' in~\cite{THEGAME_JCDCGGG}.
By  considering the transpositions as edges of a (complete) graph, the problem can be expressed as the following token swapping reachability problem.

\begin{quote}
\textbf{Subsequence Token Swapping Reachability (STS):}  Given a graph on $m$ vertices where vertex $i$ initially has a token $i$, and given a  target permutation $\pi$ of the tokens  and a sequence  of swaps $s_1, s_2, \ldots, s_{n}$, where $s_j$ is an  edge of the graph, is there a subsequence of the swaps that realizes $\pi$?
\end{quote}

We prove that STS remains NP-complete when the underlying graph is a star.  

\begin{lemma} 
\label{lemma:StarSTS}
Star STS is NP-complete,  even  in the special case where the target permutation places token 0 back on the center.
\end{lemma}
\begin{proof}[Proof of Lemma~\ref{lemma:StarSTS}]
We give a polynomial time reduction from STS to Star STS.
Suppose the input to STS consists of the target  permutation $\pi \in S_m$  and the swaps $s_i = (a_i, b_i)$, $1 \le i \le n$, where $a_i, b_i \in  \{1, \ldots, m\}$.

Construct an  instance of Star STS as follows.
Make a leaf of the star for each of the $m$ vertices of the input graph. Also, for each swap $s_i$ in the input sequence, make two leaves $s^{\rm in}_i$  and  $s^{\rm out}_i$. 
Each leaf  has a token of its same  label, and the center vertex of the star  has token  $0$. 
Define the target permutation such that token  $0$ remains on the center vertex, tokens $1,  \ldots, m$ permute according to the input permutation $\pi$, and tokens $s^{\rm in}_i$  and  $s^{\rm out}_i$ must switch positions for $1 \le i \le n$. 

To construct the swap sequence, replace each $s_i= (a_i, b_i)$ by the 6-element sequence $s^{\rm in}_i,  a_i, \allowbreak b_i, \allowbreak a_i, \allowbreak s^{\rm out}_i,  s^{\rm in}_i$ (recall that an element $x$ indicates a swap on edge $(0,x)$). 
Observe that this construction  can be carried out in polynomial time.

If the original instance of STS has a solution, then we get a solution for the constructed instance of Star STS as follows.
If $s_i$ is chosen in the original, then choose the  corresponding 6-element sequence, and otherwise, choose only $s^{\rm in}_i$, $s^{\rm out}_i$, $s^{\rm in}_i$. 

In  the other direction,  suppose the instance of  Star STS has a solution.  For each $i$, we must use the subsequence $s^{\rm in}_i$, $s^{\rm out}_i$, and $s^{\rm in}_i$ because 
no other swaps operate on vertices $s^{\rm in}_i$  and  $s^{\rm out}_i$, and all three of these swaps are needed to re-position the tokens correctly, i.e.~there
is no other way to switch the  positions of tokens $s^{\rm in}_i$  and  $s^{\rm out}_i$. 
Note that the first swap places token $s^{\rm in}_i$ at the  center vertex  of the star.  Now consider the swaps $a_i, b_i, a_i$.  If we use all three, then this corresponds to using $s_i$,  and  if we use none  of them,  then this corresponds to not using $s_i$.  
We claim that these are the only possibilities because, if we
choose $a_i,a_i$ it does nothing, and if we choose any other proper non-empty subsequece  of $a_i, b_i, a_i$ then token $s^{\rm in}_i$, would  not be  at the  center  vertex when we perform $s^{\rm out}_i$, and $s^{\rm in}_i$,  and  thus would not reach its target vertex. 
\end{proof}
